\documentclass[
notitlepage,
floatfix,
aps,
prl,
reprint,
twocolumn,
superscriptaddress,
longbibliography
]{revtex4-2}
\usepackage[utf8]{inputenc}
\usepackage[normalem]{ulem}

\usepackage[caption=false]{subfig}
\usepackage{graphicx}
\usepackage{epstopdf}
\usepackage{array}
\usepackage{verbatim}

\usepackage{amsmath,amssymb,amscd,mathtools}
\usepackage{amsthm}
\usepackage{tabularx}
\usepackage{stmaryrd}
\usepackage{enumerate}
\usepackage{wasysym}
\usepackage{microtype}
\usepackage{hyperref}
\usepackage[dvipsnames]{xcolor}
\hypersetup{
	bookmarksnumbered=true, 
	unicode=false, 
	pdfstartview={FitH}, 
	pdftitle={}, 
	pdfauthor={}, 
	pdfsubject={}, 
	pdfcreator={}, 
	pdfproducer={}, 
	pdfkeywords={}, 
	pdfnewwindow=true, 
	breaklinks=true,
	colorlinks=true, 
	linkcolor=NavyBlue, 
	citecolor=NavyBlue, 
	filecolor=NavyBlue, 
	urlcolor=NavyBlue 
}
\usepackage{tikz}
\usepackage[lmargin=.7in,rmargin=.7in,tmargin=.7in,bmargin=1in]{geometry}

\usepackage{newtxtext,newtxmath}
\makeatletter
\def\maketitle{
	\@author@finish
	\title@column\titleblock@produce
	\suppressfloats[t]}
\makeatother

\usepackage{booktabs}

\theoremstyle{plain}
\newtheorem{thm}{Theorem}
\newtheorem{cor}[thm]{Corollary}
\newtheorem{lem}[thm]{Lemma}
\newtheorem{pro}[thm]{Proposition}

\theoremstyle{definition}
\newtheorem{defn}[thm]{Definition}

\newcommand{\x}[1]{}
\newcommand{\eq}[1]{(\hyperref[eq:#1]{\ref*{eq:#1}})}

\renewcommand{\sec}[1]{\hyperref[sec:#1]{Section~\ref*{sec:#1}}}
\newcommand{\thrm}[1]{\hyperref[thrm:#1]{Theorem~\ref*{thrm:#1}}}
\newcommand{\lemm}[1]{\hyperref[lemm:#1]{Lemma~\ref*{lemm:#1}}}
\newcommand{\prop}[1]{\hyperref[prop:#1]{Proposition~\ref*{prop:#1}}}
\newcommand{\corr}[1]{\hyperref[corr:#1]{Corollary~\ref*{corr:#1}}}
\newcommand{\fig}[1]{\hyperref[fig:#1]{~\ref*{fig:#1}}}
\newcommand{\deff}[1]{\hyperref[deff:#1]{~\ref*{deff:#1}}}

\newcommand{\mN}{\mathcal{N}}

\newcommand{\mE}{\mathcal{E}}

\newcommand{\Upsirec}[1]{\hat{U}^{[#1]}}
\newcommand{\Upsirecch}[1]{\hat{\mathbf{U}}^{[#1]}}
\newcommand{\Upsirecapp}[1]{\hat{\mathcal{U}}^{[#1]}}

\newcommand{\mH}{\mathcal{H}}

\newcommand{\hmN}{\hat{\mathcal{N}}}
\newcommand{\hN}{\hat{N}}
\newcommand{\mO}{\mathcal{O}}
\newcommand{\mB}{\mathcal{B}}

\newcommand{\mS}{\mathcal{S}}

\newcommand{\mbN}{\mathbb{N}}

\newcommand{\mbR}{\mathbb{R}}

\newcommand{\psivec}[1]{\ket{\psi_{#1}}}

\newcommand{\Eop}[2]{\hat{E}^{[#2]}_{#1}}
\newcommand{\Eopch}[2]{\hat{\mathbf{E}}^{[#2]}_{#1}}
\newcommand{\hE}[2]{\hat{\mE}^{[#2]}_{#1}}
\newcommand{\Gammaop}[2]{\hat{\Gamma}^{[#2]}_{#1}}
\newcommand{\Gammaopch}[2]{\hat{\mathbf{\Gamma}}^{[#2]}_{#1}}
\newcommand{\Gammaerr}[2]{\hat{\bbGamma}^{[#2]}_{#1}}

\DeclareMathOperator{\Tr}{Tr}
\DeclareMathOperator{\id}{id}

\DeclareMathAlphabet{\mathmybb}{U}{bbold}{m}{n}
\newcommand{\1}{\mathmybb{1}}
\newcommand{\bbGamma}{\mathmybb{\Gamma}}
\DeclareMathOperator{\arcosh}{arcosh}
\DeclareMathOperator{\sech}{sech}

\DeclareMathOperator{\arcsech}{arcsech}

\newcommand{\ket}[1]{|{#1}\rangle}
\newcommand{\bra}[1]{\langle{#1}|}
\newcommand{\braket}[2]{\langle{#1}|{#2}\rangle}
\newcommand{\ketbra}[2]{|{#1}\rangle\!\langle{#2}|}

\newcommand{\no}{\nonumber}
\newcommand{\ba}{\begin{eqnarray}}
	\newcommand{\ea}{\end{eqnarray}}
\newcommand{\bann}{\begin{eqnarray*}}
	\newcommand{\eann}{\end{eqnarray*}}
\newcommand{\bal}{\begin{align}\begin{aligned}}
		\newcommand{\eal}{\end{aligned}\end{align}}
\newcommand{\dm}[1]{\ketbra{#1}{#1}}

\newcolumntype{L}[1]{>{\raggedright}p{#1}}
\newcolumntype{C}[1]{>{\centering}p{#1}}
\newcolumntype{R}[1]{>{\raggedleft}p{#1}}
\newcolumntype{D}{>{\centering\arraybackslash}X}

\newcommand{\B}{{\rm L}}

\newcommand{\tp}{\tau^{\perp}}

\DeclarePairedDelimiter\abs{\lvert}{\rvert}%

\newcommand{\lset}{\left\{ }
\newcommand{\rset}{\right\}}

\usepackage{soul}
\usepackage{lipsum}
\usepackage{tcolorbox}
\tcbuselibrary{breakable}
\newtcolorbox{mybox}[1]{colback=red!5!white,colframe=red!65!black,fonttitle=\bfseries,title=#1,boxrule=0.7pt}
\usepackage{makecell}

\newcommand{\spec}{\mathrm{spec}}
\newcommand{\suppl}{Supplemental Materials}

\begin{document}
\title{Quantum Dynamic Programming}
\author{Jeongrak Son}
\affiliation{School of Physical and Mathematical Sciences, Nanyang Technological University, 637371, Singapore}
\author{Marek Gluza}
\email{marekludwik.gluza@ntu.edu.sg}
\affiliation{School of Physical and Mathematical Sciences, Nanyang Technological University, 637371, Singapore}
\author{Ryuji Takagi}
\affiliation{Department of Basic Science, The University of Tokyo, Tokyo 153-8902, Japan}
\affiliation{School of Physical and Mathematical Sciences, Nanyang Technological University, 637371, Singapore}
\author{Nelly H. Y. Ng}
\email{nelly.ng@ntu.edu.sg}
\affiliation{School of Physical and Mathematical Sciences, Nanyang Technological University, 637371, Singapore}
\date{\today}
	

\begin{abstract} 
We introduce a quantum extension of dynamic programming, a fundamental computational method that efficiently solves recursive problems using memory. 
Our innovation lies in showing how to coherently generate recursion step unitaries by using memorized intermediate quantum states. 
Quantum dynamic programming achieves an exponential reduction in circuit depth for a broad class of fixed-point quantum recursions, though this comes at the cost of increased circuit width. 
Interestingly, the trade-off becomes more favourable when the initial state is pure.
By hybridizing our approach with a conventional memoryless one, we can flexibly balance circuit depth and width to optimize performance on quantum devices with fixed hardware constraints.
Finally, we showcase applications of quantum dynamic programming to several quantum recursions, including a variant of Grover's search, quantum imaginary-time evolution, and a new protocol for obliviously preparing a quantum state in its Schmidt basis.
\end{abstract}

\maketitle

\emph{Introduction.}---Storing intermediate results in memory can drastically reduce computation runtime. 
Computing the Fibonacci sequence provides a perfect example. 
Direct calculation of the $N$th element using the recursive definition $F(n) = F(n-1)+F(n-2) $, starting with $F(0) = 0$ and $F(1) = 1$ and without storing intermediate results, requires calculating $F(N-2)$ twice---once for $F(N)$ and again for $F(N-1)$. 
In turn, computing $F(N-2)$ requires two computations of $F(N-4)$, and so on. 
This repetition causes the computation time to grow exponentially with $N$. 
On the other hand, if we store intermediate values of $F(n)$ in memory and retrieve them whenever necessary, $F(N)$ can be computed in time linear to $N$, achieving an exponential speedup.
This technique, called memoization~\cite{Michie1968memoization}, is a form of dynamic programming~\cite{Bellman1952_DP, Bellman_dynamic} that utilizes a small amount of memory to yield vastly shorter runtime by avoiding re-computations.

We introduce \emph{quantum dynamic programming (QDP)} as a technique to utilize memoization in quantum computation.
The quantum analogue of computing the Fibonacci sequence would be solving quantum recursions, which we define as
\begin{align}\label{eq:qrecursion}
	\psivec{n}\mapsto\psivec{n+1} = \Upsirec{\psi_{n}}\psivec{n}\ ,
\end{align}
a sequence of quantum states $\{\psivec{n}\}_{n}$ related by unitary operators $\Upsirec{\psi_{n}}$ that explicitly depend on the past states $\psivec{n}$.
Such recursions appear in important quantum algorithms, including nested fixed-point search~\cite{Grover2005FP, Yoder2014Grover} and double-bracket quantum algorithms~\cite{Gluza2024DBI, Robbiati2024DBQA, Xiaoyue2024DBQA, QITEDBF}.

There are various approaches to solving Eq.~\eqref{eq:qrecursion}, yet a uniquely quantum challenge underlies them all.
Unlike classical computing, reading out a quantum state $\ket{\psi}$ requires many copies, and this cannot be circumvented~\cite{Gisin1997Cloning} as dictated by the no-cloning theorem~\cite{Wootters1982Nocloning}.
If one compiles the recursion step $\Upsirec{\psi_n}$ by learning each intermediate state $\psivec{n}$ through, e.g. tomography~\cite{Haah2016_tomography}, each step must be repeated many times before learning how to compile $\Upsirec{\psi_n}$. An alternative approach, termed \emph{unfolding}, is analogous to classical computation without memoization.
It makes multiple queries to an operation preparing $\psivec{n}$ in order to compile $\Upsirec{\psi_{n}}$.
However, since $\psivec{n}$ is not learnt, operations preparing it must be compiled using queries to the initial step $\Upsirec{\psi_{0}}$, which is assumed to be given.
This leads to the circuit depth exponential in $N$ for preparing $\psivec{N}$. 
We expound on this scaling later in this Letter with an example.

Our approach to QDP belongs to a lineage of studies on circuits instructed by quantum states~\cite{Lloyd2014quantum, Marvian2016_emulator, Pichler2016DME, Kimmel2017DME_OP, Kjaergaard2022DME, Wei2023hermpreserving, Patel2023WML1, Patel2023WML2, Rodriguezgrasa2023cloningDME, Go2024DME, Schoute2024QProgrammableReflections}, where desired operations are implemented by injecting quantum instruction states that encode the operation, rather than by compiling circuits based on classical information.
This paradigm offers several advantages.
First, the circuit can be implemented obliviously to the quantum instructions, i.e. without learning them, while requiring significantly fewer samples compared to tomography-based learning-and-compiling methods~\cite{Kimmel2017DME_OP, Go2024DME}.
Furthermore, the circuits instructed by quantum states are universal~\cite{Kimmel2017DME_OP}, enabling all queries to $\Upsirec{\psi_{0}}$ (required in the unfolding approach) to be replaced by the injection of $\psivec{0}$.

The main breakthrough of our work is the following observation: instead of unfolding the recursion (i.e. compiling $\Upsirec{\psi_{n}}$ with exponentially many queries to $\Upsirec{\psi_{0}}$), we can dynamically evolve memory states from $\psivec{0}$ to $\psivec{n}$ in parallel and directly implement $\Upsirec{\psi_{n}}$ by injecting these memory states. 
In essence, QDP `folds up' the unfolded recursion step $\Upsirec{\psi_{n}}$ into the memory state as it evolves to $\psivec{n}$.
By doing so, we achieve an exponential circuit depth reduction compared to the unfolding approach, as formalized in Theorem~\ref{thm:qdp_pure_nontech}, en par with the classic example of computing Fibonacci numbers.

However, the same no-cloning challenge prevails QDP.
Many copies of each intermediate state $\psivec{n}$ are required, necessitating circuits with greater width.
The crucial distinction from precomputation approaches~\cite{Marvian2016_emulator, Huggins2023_precomputation} emerges here: our memory states are not given a priori, but we prepare them as a result of previous recursion steps.  
This results in a trade-off, converting circuits with exponential depth into those with exponential width. 
Nonetheless, QDP remains the strategically correct choice.
Any protocol must conclude within the coherence time of the quantum processor, and QDP precisely offers an explicit tool that circumvents such hard limitations---providing a concrete strategy that radically reconfigures circuit structures to leverage quantum memoization. 

In this manuscript, we systematically introduce QDP for generic quantum recursions and establish its feasibility.
We exemplify the utility of QDP on known algorithms, e.g. recursive Grover's search~\cite{Yoder2014Grover} and quantum imaginary-time evolution~\cite{QITEDBF}. Lastly, we propose a natively dynamic algorithm for transforming a state into its Schmidt basis without learning the state itself.
With this, we establish that QDP is not only a strategic conceptual framework but also a highly applicable one.

\emph{Memory-calls in quantum recursions.}---Quantum recursion unitaries can be regarded as a mapping from a quantum state $\ket{\psi}$ to a unitary operator $\Upsirec{\psi}$.
In general, $\Upsirec{\psi} = e^{if(\dm{\psi})}$, where $f$ is a function from Hermitian operators to Hermitian operators.
Then, $\Upsirec{\psi}$ can be approximated by 
\begin{align}\label{eq:dynamicunitary0}
	\Upsirec{\{\hmN\},\psi}= \hat{V}_{L} e^{i\hmN_{L}(\dm{\psi})} \hat{V}_{L-1}  \ldots \hat{V}_{1} e^{i\hmN_{1}(\dm{\psi})} \hat{V}_{0}\ ,
\end{align}
where $\{\hmN\}$ is a collection of Hermitian-preserving (linear or polynomial) maps $\hmN_{i}$, while $\hat{V}_{i}$ are static unitaries independent of the instruction $\ket{\psi}$.
We say that a quantum recursion unitary written as Eq.~\eqref{eq:dynamicunitary0} contains $L$ \emph{memory-calls} of the form $e^{i\hmN(\dm{\psi})}$.
The most prominent example is found in Grover's algorithm~\cite{Grover96, Grover2005FP}, where alternating reflections
\begin{align}\label{eq:GammaLdef1}
	\Upsirec{\psi}_{L} = \displaystyle\prod_{i=1}^L e^{-i\alpha_{i}\dm{\psi}}e^{-i\beta_{i}\dm{\tau}}\ ,
\end{align}
with respect to the initial state $\ket{\psi}$ and the target state of the search $\ket{\tau}$ are applied to $\ket{\psi}$.
Eq.~\eqref{eq:GammaLdef1} contains $L$ memory-calls, each $e^{i\hmN_{i}(\psi)}$ defined by $\hmN_{i}(\psi) = -\alpha_{i}\dm{\psi}$ and a set of predetermined angles $ \{\alpha_{i},\beta_{i}\}_{i=1}^{L} $. 
Another example is the double-bracket iteration $\Upsirec{\psi} = e^{s[\hat D,\dm{\psi}]}$ with a fixed operator $\hat D$, which facilitates diagonalization~\cite{Gluza2024DBI}.
This corresponds to choosing $L = 1$, $V_{0} = V_{1} = \1$, and $\hmN_{1}(\dm{\psi}) = -is[\hat D,\dm{\psi}]$.

\emph{Unfolding quantum recursions.}--- We start by introducing the standard approach in implementing quantum recursions.
Consider the Grover unitary Eq.~\eqref{eq:GammaLdef1}, but applied recursively $ \psivec{n+1} = \Upsirec{\psi_n}_{L} \psivec{n}$ starting from a root state $\psivec{0}$.
By assumption, $e^{-i\alpha_{i}\dm{\psi_{0}}}$ and $e^{-i\beta_{i}\dm{\tau}}$ are readily available.
However, $\Upsirec{\psi_n}_{L}$ involves memory-calls $e^{-i\alpha_{i}\dm{\psi_{n}}}$, which are not directly accessible.
Using the relation $\dm{\psi_{n}}=\Upsirec{\psi_{n-1}}_{L} \dm{\psi_{n-1}}\bigl(\Upsirec{\psi_{n-1}}_{L}\bigr)^\dagger$, we observe that
\begin{align}\label{eq:reflectorcovariance}
	e^{-i\alpha_{i}\dm{\psi_{n}}} = \Upsirec{\psi_{n-1}}_{L} e^{-i\alpha_{i}\dm{\psi_{n-1}}} \bigl(\Upsirec{\psi_n-1}_{L}\bigr)^\dagger\ .
\end{align}
We name this property \emph{covariance} and use it to `unfold' the memory-call $e^{-i\alpha_{i}\dm{\psi_{n}}}$ into $2L+1$ memory-calls to $\psivec{n-1}$.
This unfolding continues until each memory-call to $\psivec{n}$ in $\Upsirec{\psi_n}_{L}$ is implemented with $(2L+1)^{n}$ calls to $\psivec{0}$.
The nested fixed-point Grover search~\cite{Yoder2014Grover} exactly follows this method.

To our knowledge, all existing quantum recursions are implemented through unfolding.
However, there are two severe drawbacks.
First, memory-calls must be covariant; otherwise they need to be approximated by covariant ones, via group commutators~\cite{dawson2006solovay} as in Refs.~\cite{Gluza2024DBI, Robbiati2024DBQA, Xiaoyue2024DBQA, QITEDBF} or a general scheme described in Sec.~\ref{app:LAI_blackbox} of~\cite{suppl}.
Moreover, unfolding yields exponentially deep circuits, with the final recursion step $\Upsirec{\psi_{N-1}}$ dominating the depth, as it requires $L(2L+1)^{N-1}$ calls to $\psivec{0}$.

\emph{Memory-usage queries.}---To resolve the exponential depth of unfolded circuits, QDP invokes dynamically evolving memory states to directly instruct memory-calls. 
The goal is to apply the unitary operation $e^{i\hmN(\dm{\psi})}$ to a working state $\ket{\phi}$: 
\begin{align}\label{eq:executing_singlememorycall}
    \ket{\phi}\mapsto e^{i\hmN(\dm{\psi})}\ket{\phi}\ .
\end{align}
A crucial observation is that Eq.~\eqref{eq:executing_singlememorycall} can be approximately implemented using multiple copies of the (unknown) state $\ket{\psi}$~\cite{Kimmel2017DME_OP,Wei2023hermpreserving}.
Consider the channel 
\begin{align}\label{eq:HME_def_maintext}
	\hE{s}{\hmN,\psi}(\dm{\phi}) = \Tr_{1}\left[e^{-i\hN s}\left(\dm{\psi}\otimes\dm{\phi}\right)e^{i\hN s}\right]\ ,
\end{align}
where the bipartite unitary $e^{-i\hN s}$ is oblivious to both the memory and working register.
Specifically, $\hN$ depends solely on the Hermitian-preserving map $\hmN$ that characterizes the memory-call.
By choosing $\hN = \Lambda^{\transp_{1}}$, the partial transpose of the Choi matrix $\Lambda$ of $\hmN$ (see \suppl~\cite{suppl} Section~\ref{app:HME} for an explicit expression), Eq.~\eqref{eq:HME_def_maintext} approximates the memory-call $e^{is\hmN(\dm{\psi})}$ with error $\mO(s^{2})$.
For instance, when $\hmN = \id$, which corresponds to density matrix exponentiation~\cite{Lloyd2014quantum}, $\hN = \mathrm{SWAP}$ is used in Eq.~\eqref{eq:HME_def_maintext}.
We call this channel a single \emph{memory-usage query}, as it consumes one copy of $\ket{\psi}$.
Each memory-usage query invokes an error scaling quadratically to the duration $s$; by setting $s = 1/M$ and repeating the same memory-usage queries Eq.~\eqref{eq:HME_def_maintext} $M$ times, we obtain the channel $(\hE{1/M}{\hmN,\psi})^{M}$ that approximates Eq.~\eqref{eq:executing_singlememorycall} with total error $\mO(s^{2}/M)$.
Therefore, the error in the memory-call approximation can be kept arbitrarily small by increasing $M$.

\emph{Exponential depth reduction with QDP.}---We define QDP as the process of making memory-usage queries on copies of $\psivec{n}$ to approximate a recursion step towards $\psivec{n+1}$.
Specifically, we initialize a quantum program with $\sigma_0 = \dm{\psi_{0}}$ and define the QDP iteration as
\begin{align}\label{eq:iteratedmemoryusage}
	\sigma_{n+1} = 	\hat{V}_{1} \left(\hE{1/M}{\hmN,\sigma_{n}}\right)^{M}\left(\hat{V}_{0}\sigma_{n}\hat{V}_{0}^{\dagger}\right)\hat{V}_{1}^{\dagger}\ .
\end{align}
Eq.~\eqref{eq:iteratedmemoryusage} approximates a single memory-call recursion (Eq.~\eqref{eq:dynamicunitary0} with $L = 1$) by replacing the memory-call with $M$ memory-usage queries.
In general, memory states $\sigma_{n\neq0}$ are not pure because the channel $\hE{1/M}{\hmN,\sigma_{n}}$ is not exactly unitary. 
However, we show in Theorem~\ref{thm:qdp_pure_nontech} that $\sigma_{N}$ can be made arbitrarily close to the desired state $\psivec{N}$. 
If recursions involve multiple memory-calls, QDP generalizes by replacing all memory-calls with $M$ total memory-usage queries to the memory state $\sigma_{n}$.

To prepare one copy of $\sigma_{1}$, we make $M$ memory-usage queries with the memory state $\sigma_{0}$.
This requires $(M+1)$ root state copies ($M$ in the memory register and one in the working register).
Likewise, preparing $\sigma_{2}$ requires $(M+1)^{2}$ copies of $\sigma_{0}$; thus, the iteration of Eq.~\eqref{eq:iteratedmemoryusage} consumes $(M+1)^{n}$ copies of the root state for $\sigma_{n}$.
Meanwhile, the circuit depth remains linear: 
multiple instruction states can be prepared in parallel, resulting in a maximum depth of $nM$ memory-usage queries.
An intuitive explanation can be drawn by comparing QDP to unfolding, where a memory-call to $\psivec{n}$ is executed by making $(2L+1)^{n}$ memory-calls to the root state.
In contrast, QDP folds all root state calls into the memory state $\sigma_{n}$, allowing a memory-call to be implemented with a fixed-depth circuit.

\emph{The accuracy of QDP.}---Seeking an exact solution to quantum recursions is unrealistic, as circuit compilation always entails inaccuracies. 
Hence, it is necessary to ensure that each memory-call approximation in QDP is sufficiently accurate to prepare the final state $\sigma_{N}$ within the same error threshold as unfolding. 
Setting $M = \mO( 1/\epsilon)$ in Eq.~\eqref{eq:iteratedmemoryusage} ensures that the QDP prepared state $\sigma_{1}$ approximates the exact result $\psivec{1}$ within an error $\|\sigma_1-\dm{\psi_{1}}\|_1 \le \mO(\epsilon) $.
However, subsequent steps may amplify this error, as $\sigma_{1}\mapsto\sigma_{2}$ becomes instructed by $\sigma_1$, not $\psivec{1}$.
Indeed, the triangle inequality gives $\|\sigma_2-\dm{\psi_{2}}\|_1 \le \mO(M\epsilon) $; see Sec.~\ref{app:locally_accurate_implementations} of~\cite{suppl} for the full analysis. 

To address this, we identify sufficient conditions to prevent such destructive error propagation. 
It turns out that if the first few steps are sufficiently accurate, then later steps benefit from stabilization properties of typical quantum recursions.
In particular, recursions with fixed-points (e.g. via the Polyak-Łojasiewicz inequality in gradient descent iterations~\cite{karimi2016linear} or asymptotic stability in time-discrete dynamical systems~\cite{Moore1994DiscreteDBI}) typically exhibit exponential convergence. 
For quantum recursions, this implies $\|\psivec{N}-\psivec{\infty}\|_1 \le \alpha^N\|\psivec{0}-\psivec{\infty}\|_1$, for some $\alpha<1$.
However, this relies on three assumptions: \emph{i)} a stable fixed-point $\psivec{\infty}$ as $N\rightarrow \infty$, \emph{ii)} uniqueness of the fixed-point, and that \emph{iii)} the fixed-point is sufficiently strongly attracting. 
We term quantum recursions satisfying \emph{i,ii,iii)} (or their generalization to mixed state recursions as in Sec.~\ref{subsec:main_Thm_proof} of~\cite{suppl}) as exhibiting \emph{fast spectral convergence}.
Without \emph{iii)}, recursions become unstable, suggesting limited physical relevance, as practically achievable protocols must withstand small perturbations.
In such cases, quantum computation would require infinite precision and resources for a successful operation.

It is worth noting that unfolding is also subject to similar stability constraints: if $\Upsirec{\psi_{0}}$ is not compiled exactly, subsequent recursion unitaries $\Upsirec{\psi_{n}}$, which rely on exponentially many applications of $\Upsirec{\psi_{0}}$, become exponentially unstable. 
Thus, whether using QDP or unfolding, quantum recursions must exhibit fast spectral convergence; otherwise, the computational task is ill-conditioned. When this condition is met, QDP delivers an exponential depth reduction compared to unfolding:
\begin{thm}[Exponential circuit depth reduction]\label{thm:qdp_pure_nontech}
	Suppose a quantum recursion starting from $\psivec{0}$ satisfies fast spectral convergence. Then QDP combined with a mixedness reduction protocol yields a final state $\sigma_{N}$ satisfying
	\begin{align}
		\|\sigma_N - \dm{\psi_{N}}\|_1 \leq \epsilon,
	\end{align}
	using a circuit of depth $\mO(N\epsilon^{-1})$ and width $\epsilon^{-N}e^{\mO(N)}$.
	Here, $\psivec{N}$ is the exact solution to the recursion. 
\end{thm}
The technical version of this theorem and its proof is given as Theorem~\ref{thm:qdp_pure} in~\cite{suppl}.
We prove this theorem by addressing two different types of error. In particular, we note that unitary errors are efficiently suppressed by fast spectral convergence, while non-unitary errors destabilize the convergence and need to be mitigated explicitly. 
We employ an additional protocol~\cite{Cirac1999_purification, Childs2024purification} to reduce
such errors, and show in~\cite{suppl} Sec.~\ref{subsec:QDP_pure} that this protocol incurs only a slightly faster exponential growth in circuit width.
Furthermore, a similar theorem can be developed for mixed state extensions, which we display in End Matters.

\emph{Implications of the depth-width trade-off: distributed quantum computing and hybrid strategy.}---Theorem~\ref{thm:qdp_pure_nontech} establishes the desired circuit depth-width trade-off.
At its core, this trade-off is valuable because it exponentially reduces computational time, which can turn a practically impossible task into a feasible one.
In addition, the capacity to realize a shorter depth circuit comes with potential advantages when combined with quantum error mitigation. In particular, it was shown that error mitigation generally requires exponential resource overhead with circuit depth~\cite{Takagi2022fundamental,Takagi2023universal,Tsubouchi2023universal}. Although there are observations indicating that error mitigation could require exponential cost also with circuit width~\cite{Temme2017error,Quek2024exponentially}, a rigorous characterization applicable to general classes of circuits and error mitigation methods has not been established. In this sense, there is still room for some error mitigation techniques to run efficiently in large-width short-depth circuits of practical interest.
Beyond this immediate advantage, we focus on how QDP facilitates implementation. 

One simple measure is the quantum circuit size, defined as the product of the maximum circuit depth and width~\cite{Yoder2016universal,Takagi2017error}. 
We illustrate how this measure benefits from the trade-off.
Suppose that in each recursion step, $L$ non-covariant memory-calls are made, each approximated by $k$ \emph{covariant} ones. 
The unfolding circuit size therefore scales as $(2kL+1)^{N}$ for $N$ recursion steps. 
Meanwhile, memory-usage queries directly implement even non-covariant memory-calls, but require $k'$ memory-usage queries to manage implementation error, i.e. $M = k'L$ memory-usage queries per recursion step. 
In this scenario, the QDP circuit width scales as $(M+1)^{N} = (k'L+1)^{N}$, which dominates the total circuit size for large $N$.
Whenever $L$ is large and $k'<2k$, QDP approximately achieves a $(2k/k')^{N}$-fold improvement in circuit size compared to unfolding.
For instance, in the most error-tolerant case of $k = k' = 1$, QDP significantly reduces circuit size.

Beyond circuit size, QDP offers additional practical advantages: its local modularity is particularly well-suited for distributed quantum computing~\cite{Wehner2018_QInternet, Cacciapuoti2019Distributed, Davarzani2020Distributed}. 
By localizing and decoupling circuits~\cite{Wang2024Decoupling}, QDP allows copies of $\psivec{n}$ to be prepared independently before being injected into recursion steps $\Upsirec{\psi_{n}}$.
As the recursion progresses, the circuit width decreases exponentially as most copies are consumed as memory and traced out during memory-usage queries.
This enables greater flexibility in parallelization, since processors with shorter coherence times can be allocated for preparing $\psivec{n}$ with smaller $n$, optimizing the use of available quantum devices.

The flexibility of QDP is reinforced by a hybrid strategy.
Realistically, quantum devices can neither execute exponentially many sequential gates nor operate on exponentially many qubits simultaneously. 
A practical approach is to initiate QDP after several rounds of unfolding, distributing the exponential factor between circuit depth and width to prevent either from becoming prohibitive.
For example, one could begin with $N_{1}$ unfolding steps, using a circuit of depth $e^{\mO(N_{1})}$, nearing the device’s depth limit. 
These unfolding steps are performed in parallel, maximizing circuit width to produce $M$ copies of intermediate states $\psivec{N_{1}}$.
For many quantum recursions, this starting strategy also conveniently steers the system into a regime where the recursion starts to exhibit fast spectral convergence. 
This allows  QDP to subsequently take over, executing an additional $ N_{2} $ recursions to attain the final state $ \sigma_{N_{1}+N_{2}} $. 
This hybrid strategy fully utilizes the device capacity, which is otherwise limited to producing either $ \psivec{N_{1}} $ without QDP or $ \sigma_{N_{2}} $ without unfolding.

\emph{Example 1: Grover search.}---The nested fixed-point Grover search~\cite{Yoder2014Grover} has a recursion unitary defined in Eq.~\eqref{eq:GammaLdef1}, where $N$ recursion steps of this algorithm prepare $\psivec{N}$ such that
\begin{align}\label{eq:Grover_performance}
	\frac{1}{2}\lVert \dm{\psi_{N}} - \dm{\tau} \rVert_{1} \simeq e^{-\lvert \braket{\psi_{0}}{\tau}\rvert (2L+1)^{N}}\ .
\end{align}
If the recursion is implemented through unfolding, the total circuit depth scales exponentially as $ (2L+1)^{N}$.
This matches the performance of the optimal Grover search~\cite{Grover96, Brassard2002ampamp} with the same circuit depth; refer to~\cite{suppl} Sec.~\ref{appendix:nested} for a self-contained discussion. In contrast, the QDP implementation reduces the circuit depth to $\mO(N)$, implying an exponential runtime speedup, thanks to Eq.~\eqref{eq:Grover_performance} exhibiting fast spectral convergence.
This speedup comes with the requirement of exponential copies of the initial state $\psivec{0}$.
Unfortunately, it is known that~\cite{Kimmel2017DME_OP} there is no quantum advantage in sample-based Grover search algorithms, which include the QDP implementation. 
Nevertheless, this example serves as a valuable demonstration of QDP’s applicability for a historically significant quantum algorithm. 

\emph{Example 2: Quantum imaginary-time evolution.}---We consider quantum imaginary-time evolution in the sample-based Hamiltonian simulation scenario~\cite{Kimmel2017DME_OP}, which aims at achieving the task of ground state preparation for quantum many-body systems. 
Suppose that copies of $\rho\propto \hat{H}$ are given, where $\hat{H}$ is the Hamiltonian with a zero ground state energy. 
We also prepare copies of some pure state $\psivec{0}$ having non-zero overlap with the ground state. 
Using them, recursions of the form
\begin{align}
	\psivec{n} \mapsto \psivec{n+1} = e^{s[\dm{\psi_{n}},\hat{H}]}\psivec{n},
\end{align}
with some duration $s$ can be implemented oblivious to both $\hat{H}$ and $\psivec{n}$; see~\cite{suppl}~Secs.~\ref{app:DBI_QDP} and~\ref{app:QITE_detail}.
This unitary has been shown~\cite{QITEDBF} to be a discretized version of the quantum imaginary-time evolution~\cite{Wick1954QITE} and yield a sequence of states $\{\psivec{n}\}$ whose infidelity to the ground state decays exponentially as $n$ grows, i.e. with fast spectral convergence. An alternative QDP implementation of the same task with the usual Hamiltonian simulation, instead of access to copies of $\rho\propto\hat{H}$, is studied in Ref.~\cite{QITEDBF}.

\emph{Example 3: Oblivious Schmidt decomposition.}---For pure bipartite states, the Schmidt decomposition~\cite{Schmidt1907, Ekert1995Schmidt} fully characterizes their entanglement.
Thus we may envision a protocol that unitarily transforms any unknown $\psivec{0} \in \mathcal{H}_{A}\otimes \mathcal{H}_{B}$ into a form where its Schmidt basis aligns with the computational basis: $\hat{V}_{A} \otimes \hat{V}_{B}\psivec{0} = \sum_{k=1}^{D} \sqrt{\lambda_{k}} \ket{k}\otimes \ket{k}$ with $\{\lambda_{k}\}_{k}$ denoting the Schmidt spectrum.
Traditionally, prior knowledge of $\psivec{0}$ is required to construct $\hat{V}_{A}$ and $\hat{V}_{B}$. 
However, QDP enables this task while completely circumventing such a need. 
Given that this is a novel algorithm, we outline the key ideas here, and provide a more detailed exposition in~\cite{suppl} Sec.~\ref{app:OSD_detail}.  
Let $\hat{D}$ be a non-degenerate diagonal operator on subsystem $A$.
The recursion unitary is defined as
\begin{align}\label{eq:stateDBI_main}
	\Upsirec{\hmN,\psi_{n}} = e^{s\left[\hat{D},\Tr_B[\dm{\psi_n}]\right]}\otimes \1_B\ .
 \end{align}
This is a single memory-call type recursion with $\hmN(\rho_{AB})=-is[\hat{D},\Tr_{B}[\rho_{AB}]]\otimes\1_{B}$ that is Hermitian-preserving; hence, with copies of $\psivec{n}$, the recursion step $\psivec{n} \mapsto \psivec{n+1}= \Upsirec{\psi_{n}}\psivec{n}$ can be approximated.
When $s$ is small, $\{\psivec{n}\}$ exponentially converges to $\hat{V}_{A}\otimes\1_{B}\psivec{0}$~\cite{HelmkeMoore1994Book, Smith_Thesis} allowing Theorem~\ref{thm:qdp_pure_nontech} to apply.
A similar procedure implements $\hat{V}_{B}$, completing the oblivious Schmidt decomposition.

Let us discuss the implications of oblivious Schmidt decomposition in the context of quantum information processing.
The replica method~\cite{Horodecki2002method} extracts classical information $\lambda_k$ but requires exponentially many swap operations.
Similarly, we expect that oblivious Schmidt decomposition may also require a long runtime to converge.
However, it not only enables sampling from the Schmidt spectrum but also provides the bipartite quantum state $\sum_k \sqrt{\lambda_k}\ket{k}\otimes\ket{k}$ \emph{coherently} in the computational basis.
This is useful, e.g. for entanglement distillation of an \emph{unknown} state $\psi$, contrasting with standard settings~\cite{Bennett1996concentrating} that require the knowledge about the initial state for compiling local operations.
While oblivious Schmidt decomposition via QDP may not attain the optimal asymptotic rates as derived from Ref.~\cite{Matsumoto2007universal}, it provides a constructive approach to oblivious entanglement distillation with explicit circuit implementations.

\emph{Conclusion and outlook.}---%
We demonstrated that QDP enables the efficient implementation of a broad class of quantum recursions involving memory-calls.
For well-behaved recursions, QDP yields accurate results with polynomial depth---an \emph{exponential} speed-up over the previously proposed zero-memory implementations based on unfolding~\cite{Yoder2014Grover,Gluza2024DBI}.
Additionally, QDP enables hybrid strategies of balancing circuit depth and width, making quantum recursions more feasible on real devices. By structuring quantum recursions into circuits with parallel rather than sequential operations, QDP renders them suitable as a central use-case for distributed quantum computing~\cite{Wehner2018_QInternet, Cacciapuoti2019Distributed, Davarzani2020Distributed}.

The general QDP proposal assumes the preparation of many independent copies of the same quantum state, analogous to sample-based algorithms~\cite{Lloyd2014quantum, Pichler2016DME, Kimmel2017DME_OP, Low2019qubitization, Kjaergaard2022DME,  Go2024DME, Schoute2024QProgrammableReflections}, where copies of the quantum state replace oracle access to unitary operators. 
For QDP applications involving only a marginal amount of quantum coherence, approximate cloning heuristics~\cite{Rodriguezgrasa2023cloningDME} could potentially provide exponential speed-ups without inflating the circuit width.
Nevertheless, quantitative bounds relating the coherence of input quantum states to cloning fidelity are currently unknown~\cite{cloningAcinRevModPhys.77.1225}.
It is also unclear whether coherence could signal simulation hardness for quantum recursions, akin to how entanglement affects tensor network simulations~\cite{VidalPhysRevLett.91.147902}.
Investigating the role of coherence in QDP presents an intriguing avenue for future research.

Currently, we are unaware of quantum algorithms where QDP ensures a clear computational advantage over classical methods. 
These may appear in settings robust against imperfect unitary implementations, similar to those present in diagonalizing double-bracket iterations~\cite{HelmkeMoore1994Book,Smith_Thesis,Gluza2024DBI}.
Indeed, QDP allows us to add oblivious Schmidt decomposition to the quantum algorithmic toolkit. 
We hope that once it will be feasible to experimentally implement memory-usage queries with high fidelity, oblivious Schmidt decomposition and QDP in general will facilitate practical state preparations that will advance our knowledge of quantum properties in materials, e.g. magnets or superconductors.

\emph{Acknowledgements.}---We thank Debbie Leung for drawing our attention to noise reduction algorithms applicable to higher dimensions. JS, MG, and NN are supported by the start-up grant for Nanyang Assistant Professorship of Nanyang Technological University, Singapore. MG additionally acknowledges support through the Presidential Postdoctoral Fellowship of the Nanyang Technological University. RT acknowledges the support of JSPS KAKENHI Grant Number JP23K19028, JST, CREST Grant Number JPMJCR23I3, Japan, and the Lee Kuan Yew Postdoctoral Fellowship of Nanyang Technological University Singapore.

\newpage

\normalsize
\renewcommand{\theequation}{A\arabic{equation}}
\setcounter{equation}{0}  

\onecolumngrid
\section*{End Matters}\label{sec:end_matter}
\twocolumngrid
\emph{Appendix A: mixed state quantum recursions.}---Similarly to the pure state case in Eq.~\eqref{eq:qrecursion}, we define the mixed state recursion starting from a mixed state $\rho_{0}$ as
\begin{align}\label{eq:qrecursion_mixed}
	\rho_{n}\mapsto\rho_{n+1} = \Upsirec{\rho_{n}} \rho_{n} \left(\Upsirec{\rho_{n}}\right)^{\dagger}\ ,
\end{align}
where $\Upsirec{\rho_{n}}$ is the unitary operator that depends on $\rho_{n}$.

Furthermore, the assumptions \emph{i,ii,iii)} in the main text need to be slightly refined to account for the mixed state cases. 
Specifically, the uniqueness of the fixed-point (assumption \emph{ii)}) needs only to be up to unitary equivalence between root states. 
In other words, any two root states of the recursion $\rho_{0}$ and $\rho'_{0}$ with the same spectrum have the same fixed-point $\rho_{\infty} = \rho'_{\infty}$. 
This we refer to as \emph{ii')} the fixed point being ‘spectrally’ unique.
See \cite{suppl} Sec.~\ref{subsec:main_Thm_proof} for details. 

The exponential circuit depth reduction established in Theorem~\ref{thm:qdp_pure_nontech} can also be achieved for mixed state recursions, albeit with lower efficiency. 
\begin{thm}[Mixed state quantum dynamic programming]\label{thm:qdp_nontech}
	Suppose the quantum recursion satisfies fast spectral convergence.
	For any $\epsilon>0$, implementing $N$ recursion steps with each memory-call replaced by $\mO(N\epsilon^{-1})$ memory-usage queries, yields a final state $\sigma_{N}$, such that 
	\begin{align}
			\|\sigma_N - \rho_{N}\|_1 \leq \epsilon,
		\end{align}
	using a circuit of depth $\mO(N^{2}\epsilon^{-1})$.
	Here, $\rho_N$ is the exact solution to the recursion. 
\end{thm}

The full, technical version of this theorem can be found in \suppl~\cite{suppl} Section~\ref{subsec:main_Thm_proof}, as Theorem~\ref{thm:qdp}. 
Note that we do not have an explicit exponential scaling of the width in this case. 
This is due to the lack of an additional protocol, corresponding to the mixedness reduction for the pure state case, that can preserve the spectrum of the density matrices.
Hence, the non-unitary error from memory-usage queries must be suppressed by requiring each memory-usage query to be very close to a unitary channel. 
Unfortunately, the recently developed channel purification protocol~\cite{Liu2024ChannelPurification}, which is a quantum channel version of the state mixedness reduction protocols, also does not improve the scaling, as the number of copies needed for the channel purification is comparable to that of improving each memory-usage query.

\emph{Appendix B: Non-asymptotic prospects of QDP.}---%
Without fault-tolerance, noise in quantum hardware places limitations on the use of QDP, particularly when implementations are too small to match its asymptotic performance scaling. 
However, we argue qualitatively how QDP can help mitigate noise-induced constraints.

Ref.~\cite{Robbiati2024DBQA} studied explicit gate counts preparing ground states of XXZ models using double-bracket quantum algorithms (DBQA).
It was found that for 10 qubits, a single DBQA step achieves 99.5\% fidelity to the ground state with a gate cost as low as 50 CZ gates per qubit.
This can be considered as a satisfactory  circuit compilation for the 10 qubit XXZ problem as it is both short and accurate.
However, for larger system sizes, the warm-start procedure used in Ref.~\cite{Robbiati2024DBQA} is expected to be insufficient due to optimization challenges~\cite{bharti2022noisy, Larocca2024BPreview}.

Theoretically, DBQA can systematically improve ground state preparation by running additional steps in such scenarios.
As exemplified in Ref.~\cite{Robbiati2024DBQA} for $L=10$ qubits, if the warm-start circuit is inaccurate (i.e. initialized at a relatively high energy), additional steps gradually decrease the energy. 
Since DBQA reduces energy locally, its performance is independent of system size, and in practice, CZ gate counts per qubit also remain roughly independent. 
Specifically, for $L=10$ qubits, $k=1$ required $50$ CZ gates per qubit, $k=2$ required about $200$, and $k=3$ demanded thousands.
Depending on the precise hardware and use of error mitigation, $k=2$ steps roughly exhaust the capacities of existing quantum hardware, while $k=3$ can be considered out of reach at the time of writing.

We thus expect that the transition from $k=2$ to $k=3$ could provide significant improvements for larger systems (e.g. $L=50$ or $L=150$ qubits), as available in system sizes made available by leading quantum computing companies.
However, at $k=3$ the required circuit depth would likely be prohibitively high, and the implementation is bound to fail almost every time.
Instead, the $k=3$ step can be implemented by running a single QDP step.
One should employ memory-usage queries for this step, which in turn must be compiled within tight runtime requirements.

For a qualitative estimate, suppose one density matrix exponentiation (DME)~\cite{Lloyd2014quantum} memory-usage query can be implemented using 100 CZ gates per qubit.
Then the $k=3$ DBQA step using $M=5$ memory-usage queries halve the circuit depth compared to a standard DBQA implementation.
While the compilation cost of DME memory-usage queries has not been explicitly studied except for $L=1$ qubits~\cite{Kjaergaard2022DME}, if it were, say, $1000$ CZ gates per qubit, the same argument would apply to future scenarios, e.g. optimizing the transition from $k$ to $k+1$ steps for $k\approx 4$.
Note that heuristic improvements like these will remain essential even in the fault-tolerant regime, because quantum algorithms based on quantum phase estimation rely on high-quality initializations, ensuring that QDP continues to be useful in such settings.

Let us also comment on further open questions pertaining to implementation.
When running large-scale QDP, error-mitigation can be used to suppress errors and enable an unfolded implementation of the recursion until failures appear.
This suggests that practical QDP implementations may start from an advanced stage rather than from scratch, reducing the number of required QDP steps and thereby lowering the circuit width.
Furthermore, one could consider creating memory states $
\ket{\psi_{k^*}}$ at an intermediate step $k^*$ using physical qubits and then encoding these states $
\ket{\psi_{k^*}}$ into a fault-tolerant quantum computer via heuristic (non-fault-tolerant) methods; see Refs~\cite{Buchbinder2013, QITEDBF}.

Finally, let us comment that increasing circuit width exposes the system to errors.
However, since intermediate states serve as quantum instructions, these errors may statistically average out, i.e.  $M$ memory usage queries using noisy instructions may  still yield a reduction in the cost function as the `average' instruction points in the `correct' direction.
This robustness is made precise by quantum imaginary-time evolution (see Example 2 in the main text).
Here, in every step, energy is guaranteed to decrease by a finite amount, and maintaining this decrease only requires approximating a unitary channel whose effective Hamiltonian aligns with the steepest-descent direction.
For individual errors to overhaul the energy decrease, they would need to entirely flip the direction of the effective Hamiltonian relative to the gradient operator.
We anticipate that such a flip is unlikely, even in the presence of noisy quantum instructions.
In general, recursions with fast spectral convergence should be expected to benefit from such stabilizing effects.

\clearpage
\newpage

\renewcommand{\theequation}{S\arabic{equation}}
\setcounter{equation}{0}

\title{Supplemental Materials for ``Quantum Dynamic Programming''}

\maketitle
\onecolumngrid
\tableofcontents
\vspace{1cm}

\section{Preliminaries and distance measures}

We denote Hilbert spaces as $\mH$ and the set of quantum states having support in $\mH$ as $\mS(\mH)$, which is a subset of $\mB(\mH)$, the set of generic bounded operators acting on $\mH$. 
When not specified, quantum channels are assumed to map $\mS(\mH) \rightarrow \mS(\mH)$. 
The identity operator acting on $\mH$ is denoted as $\1$, whereas the identity map from $\mS(\mH)$ to $\mS(\mH)$ is denoted as $\id$.

A few different distance measures between operators and channels are utilized in this paper. 
The trace norm distance between two normal operators acting on the same space $\mB(\mH)$ is defined to be
\begin{align}\label{eq:trace_norm_operators_def}
	\frac{1}{2}\lVert \hat{A} - \hat{B}\rVert_{1} \coloneq \frac{1}{2}\Tr\left[\sqrt{ (\hat{A} - \hat{B})^{\dagger}(\hat{A} - \hat{B}) }\right]\ .
\end{align}
Between two pure states $ \psi $ and $ \phi $, the trace norm distance can be written with respect to their overlap as $\frac{1}{2}\left\lVert \psi - \phi\right\rVert_{1} = \sqrt{1 - \lvert \braket{\psi}{\phi}\rvert^{2}}$.
Another measure uses the Hilbert-Schmidt norm and is defined as
\begin{align}\label{eq:HSnorm_operators_def}
	\frac{1}{2}\lVert \hat{A} - \hat{B}\rVert_{2} \coloneq \frac{1}{2}\sqrt{\Tr\left[ (\hat{A} - \hat{B})^{\dagger}(\hat{A} - \hat{B}) \right]}\ .
\end{align}
The operator norm of a normal operator $\hat{A}$ is defined to be $\Vert \hat{A}\Vert_{\infty} = \max\{\vert\lambda_{j}\vert\}$, where $\lambda_{j}$ are eigenvalues of $\hat{A}$; likewise, the operator norm distance between two normal operators $\hat{A}$ and $\hat{B}$ is $\frac{1}{2}\Vert \hat{A} - \hat{B}\Vert_{\infty}$.

One can define distances between two quantum channels $ \mS(\mH)\rightarrow\mS(\mH') $, leveraging on operator distances. 
The trace norm distance between two quantum channels $ \Phi_{1},\Phi_{2} $ is defined as the maximum distance between the outputs of two channels given an identical input state, i.e. 
\begin{align}\label{eq:trace_norm_dist_def}
	\frac{1}{2}\Vert \Phi_{1} - \Phi_{2}\Vert_{\rm Tr} \coloneq \frac{1}{2}\max_{\varrho\in\mS(\mH)}\left\lVert\Phi_{1}(\varrho) - \Phi_{2}(\varrho)\right\rVert_{1}\ .
\end{align}
The diamond norm distance between $ \Phi_{1},\Phi_{2} $ is a stronger measure of distance between channels, because it further optimizes over input states that can be entangled with an external reference, 
\begin{align}\label{eq:diamond_norm_dist_def}
	\frac{1}{2}\Vert \Phi_{1} - \Phi_{2}\Vert_{\diamond} \coloneq \frac{1}{2}\max_{\varrho\in\mS(\mH\otimes\mH)}\left\lVert(\Phi_{1}\otimes \id)(\varrho) - (\Phi_{2}\otimes \id)(\varrho)\right\rVert_{1}\ .
\end{align}
It is useful to note that the diamond norm distance always upper bounds the trace norm distance, $ \Vert \Phi_{1} - \Phi_{2}\Vert_{\diamond} \geq \Vert \Phi_{1} - \Phi_{2}\Vert_{\rm Tr}$.

For both operators and channels, norms $\Vert \cdotB \Vert$ with no subscript denotes that some unitary invariant norm (including trace norm, operator norm, etc.) is used.

\section{Locally accurate implementations}
\label{app:locally_accurate_implementations}

In Supplemental Materials, we focus on quantum recursions whose initial states are in general mixed.
Then a memory-call unitary Eq.~\eqref{eq:executing_singlememorycall} becomes
\begin{align}\label{eq:executing_singlememorycall_mixed}
	\sigma\mapsto\Eopch{}{\hmN,\rho}(\sigma) \coloneq e^{i\hmN(\rho)}\sigma e^{-i\hmN(\rho)}\ ,
\end{align}
with a memory state $\rho$, a working state $\sigma$, and a hermitian-preserving map $\hmN$.

From this, we can define quantum recursions. 
Consider the simplest case with only one memory-call unitary of the form Eq.~\eqref{eq:executing_singlememorycall_mixed} and trivial constant unitaries $V_{0} = V_{1} = \1$ when written in the form Eq.~\eqref{eq:dynamicunitary0}.
Starting from the root state $\rho_{0}$, the recursion is defined by $\rho_{n+1} = \Eopch{}{\hmN,\rho_{n}}(\rho_{n})$; the exact implementation would lead to a sequence of states $\rho_0, \rho_1,\ldots$
For both unfolding and QDP implementations, some memory-calls cannot be executed without error, and we need to consider the approximate implementations.

For example, consider the QDP implementation. 
If we continue the recursion with approximate memory-calls, at $(n-1)$th recursion step, the instruction states available to us in the memory register are $\sigma_n$, which might be different from the exact intermediate state $\rho_n$. 
Thus, the QDP implementation is natural in that it works with what it has in hand (placeholder quantum memory) and leap-frogs forward (implements the recursion step using memory-usage queries with the memoized intermediate state being a quanum instruction).
We say that 
the sequence of states $\sigma_1,\ldots,\sigma_N$ is an \emph{$\epsilon$-locally accurate solution} of $N$ recursion steps of Eq.~\eqref{eq:executing_singlememorycall}, if for $k=0,\ldots N-1$,
\begin{align}\label{eq:locally_accurate}
	\frac{1}{2}\left \lVert \hE {\text{QDP}}{\hmN,\sigma_k,M}({\sigma_{k}}) - \Eopch{}{\hmN,\sigma_k}(\sigma_k)\right\rVert_1=\frac{1}{2}\left\lVert \sigma_{k+1} - \Eopch{}{\hmN,\sigma_k}(\sigma_k)\right\rVert_1 \leq \epsilon. 
\end{align}
In other words, the memory-usage queries implement at each step in Eq.~\eqref{eq:executing_singlememorycall} an accurate recursion step to the best of the memory's knowledge ($\sigma_k$ rather than $\rho_k$ instructs the recursion step).
Exact solutions of a quantum recursion Eq.~\eqref{eq:qrecursion} using $e^{i\hat{\mathcal{N}}(\rho)}$ can be viewed as locally accurate with $\epsilon =0$.
Note that the $\epsilon$-locally accurate solution can always be achieved by using $M = \mO(\epsilon^{-1})$ memory-usage queries with the instruction $\sigma_{k}$ at each step.

However, even when the implementation is always locally accurate, the discrepancy in the instruction state might culminate in $\Xi_n = \frac{1}{2}\|\sigma_n-\rho_n\|_1$ deviating uncontrollably, and as discussed in the main text, the recursion itself would occur unstable.
Let us discuss how this potential instability is accounted for by worst-case upper bounds using the triangle inequality.

For clarity we assume the recursion step consists of a single memory-call implemented by $M$ memory-usage queries.
We then have that the ideal implementation is $\rho_1=\Eopch{}{\hmN,\rho_0}(\rho_0)=\Eopch{}{\hmN,\sigma_0}(\sigma_0)$ and the practical implementation $\sigma_1 = \hE {\text{QDP}}{\hmN,\sigma_0,M}(\sigma_0)$, such that
\begin{align}
	\Xi_1 &\coloneq \frac{1}{2}\|\sigma_1-\rho_1\|_1  = \frac{1}{2}\left\|\hE {\text{QDP}}{\hmN,\sigma_0,M}(\sigma_0)-\Eopch{}{\hmN,\sigma_0}(\sigma_0)\right\|_1 = \mO(1/M).
\end{align}
We observe that the distance bound also holds in the channel level when two channels use the same quantum instruction, i.e. 
\begin{align}
	\frac{1}{2}\left\|\hE {\text{QDP}}{\hmN,\sigma_0,M}- \Eopch{}{\hmN,\sigma_0}\right\|_{\Tr} = \mO(1/M)\ .
\end{align}
The QDP channel $\hE {\text{QDP}}{\hmN,\sigma_1,M}$  consists of $M$ memory-usage queries.
For each such memory-usage query defined by 
\begin{align}\label{eq:HME_def_sm}
	\hE{s}{\hmN,\sigma_1}(\omega) = \Tr_{1}\left[e^{-i\hN s}\left(\sigma_1\otimes\omega\right)e^{i\hN s}\right]\ ,
\end{align}
we have that
\begin{align}
	\hE{s}{\hmN,\sigma_1}(\omega) = \hE{s}{\hmN,\rho_1}(\omega) + \Tr_{1}\left[e^{-i\hN s}\left( (\sigma_1-\rho_1)\otimes\omega\right)e^{i\hN s}\right]\ ,
\end{align}
which leads to the bound
\begin{align}
	\frac{1}{2}\left\| \hE{s}{\hmN,\sigma_1} - \hE{s}{\hmN,\rho_1} \right\|_{\Tr} \leq \frac{1}{2}\|\sigma_1-\rho_1\|_1 = \Xi_1\ .
\end{align}
Here we see how the preparation error $\Xi_1$ of $\sigma_1$ compared to $\rho_1$ influences the distance of the recursion step maps that define $\sigma_2 = \hE{\text{QDP}}{\hmN,\sigma_1,M}(\sigma_1)$ and 
$\rho_2 = \hE{\text{QDP}}{\hmN,\rho_1,M}(\rho_1)$ so that
using telescoping and the triangle inequality
\begin{align}
	\frac{1}{2}\left\| \hE {\text{QDP}}{\hmN,\sigma_1,M} -\hE {\text{QDP}}{\hmN,\rho_1,M}\right\|_{\Tr}\leq
	\frac{M}{2}\left\| \hE{1/M}{\hmN,\sigma_1} - \hE{1/M}{\hmN,\rho_1}\right\|_{\Tr} = \mO(M \Xi_1)\ .
\end{align}
This means that
\begin{align}
	\Xi_2 &= \frac{1}{2}\left\| \hE {\text{QDP}}{\hmN,\sigma_1,M}(\sigma_1) -\Eopch{}{\hmN,\rho_1}(\rho_1)\right\|_1\le
	\frac{1}{2}\left\| \hE {\text{QDP}}{\hmN,\sigma_1,M}(\sigma_1) -\hE{\text{QDP}}{\hmN,\sigma_1,M}(\rho_1)\right\|_1 +  \frac{1}{2}\left\| \hE {\text{QDP}}{\hmN,\sigma_1,M} -\Eopch{}{\hmN,\rho_1}\right\|_{\Tr}\nonumber\\ 
	& \leq  \Xi_{1}+   \frac{1}{2}\left\| \hE {\text{QDP}}{\hmN,\sigma_1,M} -\hE {\text{QDP}}{\hmN,\rho_1,M}\right\|_{\Tr}+ \frac{1}{2}\left\| \hE {\text{QDP}}{\hmN,\rho_1,M} -\Eopch{}{\hmN,\rho_1}\right\|_{\Tr} = \mO(M\Xi_{1}).
\end{align}
Analogously, we obtain $\Xi_{n} = \mO(M^n \Xi_1)$ after $n$ iterations, potentially displaying an exponential amplification of the state preparation error, when the bound is saturated.

However, this exponential instability might not necessarily be the case, when there are other factors stabilizing the recursion.
In fact, we provide general sufficient conditions for insensitivity to intermediate errors and obtain the highly accurate final solution from locally accurate implementations with circuit depth polynomial to the number of steps $N$.
This is captured by the notion of fast spectral convergence in Section~\ref{subsec:main_Thm_proof}.

\subsection{Locally accurate implementation and unfolding of the dynamic unitary using black box evolutions}\label{app:LAI_blackbox}

In this subsection, we demonstrate how to locally accurately implement a recursion unitary, as in Eq.~\eqref{eq:locally_accurate}, given access to black box evolutions with respect to the instruction state.
In particular, we consider the generalized recursion 
\begin{align}\label{eq:dynamicunitary_poly}
	\Upsirec{\{f\},\rho} = \hat{V}_{L}e^{if_{L}(\rho)}\hat{V}_{L-1}e^{if_{L-1}(\rho)}\cdots\hat{V}_{1}e^{if_{1}(\rho)}\hat{V}_{0}\ ,
\end{align}
with $L$ memory-calls to the exponentials of the Hermitian-preserving polynomials $\{f\}$. 
Note that our definition of polynomials includes those with constant operator coefficients, e.g. $\hat{A}\rho \hat{B} \rho \hat{C}$ with operators $\hat{A},\hat{B},\hat{C}$ is a polynomial of $\rho$ with degree two.
As another example, for a fixed operator $\hat{D}$, the commutator $[\hat{D},\rho] = \hat{D}\rho - \rho\hat{D}$ is a degree one polynomial as well as the scalar multiplication $f(\rho) = s\rho$. 

In general, polynomials are not covariant $f_k(\hat U \rho \hat U^\dagger)\neq \hat U f_k(\rho) \hat U^\dagger$, but the black box queries $e^{-i\rho t}$ are.
In Lemma~\ref{lem:Murao_technique} below, we will show that memory-calls $e^{if(\rho)}$ of polynomials $f$ of degree $d$ can be locally accurately implemented by making queries to evolutional oracles $e^{-it\rho^{\otimes d}}$ generated by the product of the quantum instruction $\rho^{\otimes d}$.
Then the unfolding implementation follows from making appropriate substitutions: each time the memory-call $e^{if(\rho_{n})}$ is made, the locally accurate implementation needs queries to the black box evolution $e^{-it\rho_{n}^{\otimes d}}$, which can be replaced by the query to the root state evolution $e^{-it\rho_{0}^{\otimes d}}$ using the covariant form
\begin{align}\label{eq:covariance_polynomial}
	e^{-it\rho_{n}^{\otimes d}} = \hat{U}^{\otimes d} e^{-it\rho_{0}^{\otimes d}}\left(\hat{U}^{\dagger}\right)^{\otimes d}\ ,
\end{align}
where $\hat{U} = \Upsirec{\{f\},\rho_{n-1}} \Upsirec{\{f\},\rho_{n-2}} \cdots\Upsirec{\{f\},\rho_{0}} $.

Ref.~\cite{Odake2023higherorder} introduced a versatile technique which allows the locally accurate implementation of all dynamic unitaries in the form of Eq.~\eqref{eq:dynamicunitary0} that are physically compilable using black box evolutions $e^{-it\rho}$. 
To be compilable using black box evolutions, each map $\mN_{k}$ for the memory-calls must satisfy $\mN_{k}(\1)\propto \1$, which stems from the fact that the evolution with the global phase added $e^{-it(\rho + c\1)}$ should yield the same output unitary $e^{-it\mN_{k}(\rho+\1)} = e^{-it\mN_{k}(\rho)}e^{-ic't}$ up to a global phase factor. 
\begin{lem}[Algorithm~1 in Ref.~\cite{Odake2023higherorder}]\label{lem:Murao_technique}
	Let $\mN$ be some Hermitian-preserving linear map such that $\mN(\1) \propto \1$.
	Given access to $e^{-it\rho}$ for $t>0$ and an unknown state $\rho$, it is possible to approximate $e^{i\mN(\rho)}$ with an error bounded by $\epsilon$ using circuit of depth $\mO(\beta^{2}\epsilon^{-1})$.
	Here, $\beta$ is the constant determined by the Pauli transfer matrix representation of $\mN$. 
	It is always possible to write a linear map $\mN(\cdot) = \sum_{\vec{\mu},\vec{\nu}}\gamma_{\vec{\nu},\vec{\mu}}T_{\vec{\nu},\vec{\mu}}(\cdot)$, with coefficients $\gamma_{\vec{\nu},\vec{\mu}}$ and maps $T_{\vec{\nu},\vec{\mu}}(\sigma_{\vec{\xi}}) = \delta_{\vec{\mu},\vec{\xi}}\sigma_{\vec{\nu}} $ that transform a Pauli matrix $\sigma_{\vec{\mu}}$ to $\sigma_{\vec{\nu}}$ for labels $\vec{\nu},\vec{\xi},\vec{\mu}\in\{0,1,2,3\}^{n}$.
	Then, $\beta = 2\sum_{\vec{\mu},\vec{\nu}}\lvert \gamma_{\vec{\nu},\vec{\mu}}\rvert$.	
\end{lem}
In other words, a locally accurate implementation $\hE{\text{BB}}{\{\mN\},\sigma_{k},M}(\sigma_{k}) $, such that
\begin{align}
	\frac{1}{2}\left\lVert \hE{\text{BB}}{\{\mN\},\sigma_{k},M}(\sigma_{k}) - \Eopch{}{\{\mN\},\sigma_{k}}(\sigma_{k}) \right\rVert_{1}\leq \epsilon\ ,
\end{align}
can be prepared using $M = \mO(1/\epsilon)$ queries to the black box evolution $e^{-it\rho}$.

Lemma~\ref{lem:Murao_technique} can be extended to approximate exponentials of Hermitian-preserving polynomials $f(\rho)$ of degree $d$, satisfying $f(\1)\propto\1$.
For that, access to the evolution $e^{-it\rho^{\otimes d}}$ is needed.
The idea is to find a linear Hermitian-preserving map $\mN$, such that $g(\rho^{\otimes d}) = f(\rho)$ and implement $\mN$ using Lemma~\ref{lem:Murao_technique}.
\begin{lem}\label{lem:poly_corresponding_lin_map}
	For any polynomial $f(\rho)$ of degree $d$, there exists a linear map $\mN$, such that $\mN(\rho^{\otimes d}) = f(\rho)$.
	If $f$ is Hermitian-preserving and $f(\1)\propto\1$, $g$ can also be made Hermitian-preserving and $\mN(\1^{\otimes d})\propto\1$.
\end{lem}
\begin{proof}
	Any polynomial of degree $d$ can be written as $f(\rho) = \sum_{j=1}^{J}P_{s_{j}}^{(j)}$ with $J$ different terms, where
	\begin{align}
		P_{s_{j}}^{(j)} = \left(\prod_{i = 1}^{s_{j}}A_{i}^{(j)}\rho\right)A_{0}^{(j)},
	\end{align}
	for some operators $A_{i}^{(j)}$ and $s_{j}\leq d$. 
	Now we define a linear map $\mN_{s_{j}}^{(j)}$ such that
	\begin{align}\label{eq:g_sj_j}
		\mN_{s_{j}}^{(j)}\left(\bigotimes_{k=1}^{d}\ketbra{v_{k}}{w_{k}}_{k}\right) = A^{(j)}_{s_{j}}\ketbra{v_{s_{j}}}{w_{s_{j}}}_{s_{j}} A^{(j)}_{s_{j}-1}\cdots A^{(j)}_{1}\ketbra{v_{1}}{w_{1}}_{1} A^{(j)}_{0}.
	\end{align}
	for any matrix element $\bigotimes_{k=1}^{d}\ketbra{v_{k}}{w_{k}}_{k}$, where each $v_{k},w_{k}\in\{1,\cdots,\dim(\rho)\}$.
	Then, it immediately gives $\mN_{s_{j}}^{(j)}(\rho^{\otimes d}) = P_{s_{j}}^{(j)}(\rho)$.
	The summation of such linear maps, $\mN = \sum_{j=1}^{J}\mN_{s_{j}}^{(j)}$ is also a linear map and yields the desired result $\mN(\rho^{\otimes d}) = f(\rho)$.
	
	The second part of the Lemma can be shown similarly. 
	If $f(\1)\propto\1$, then $\mN(\1^{\otimes d}) = f(\1)\propto\1$.
	If $f$ is Hermitian-preserving, we can make a decomposition where each $P_{s_{j}}^{(j)}$ and thus the construction Eq.~\eqref{eq:g_sj_j} are also Hermitian-preserving. 
\end{proof}

Combining Lemmas~\ref{lem:Murao_technique} and~\ref{lem:poly_corresponding_lin_map}, we can implement $e^{-if(\rho)}$ with an arbitrarily small error by making queries to $e^{-i\rho^{\otimes d}t}$.
Hence, the dynamic unitary Eq.~\eqref{eq:dynamicunitary_poly} with $L$ such memory calls can also be locally accurately implemented with error $\epsilon$ by making $\mO(1/\epsilon)$ queries to $e^{-i\rho^{\otimes d}t}$.

\subsection{Hermitian-preserving map exponentiation and its extension}\label{app:HME}

We briefly recap the technique of Hermitian-preserving map exponentiation introduced in Ref.~\cite{Wei2023hermpreserving}.
The intuition behind the idea is the following: 
if a map $\hmN$ is not completely-positive and trace-preserving, then the operation $\sigma \mapsto \hmN(\sigma)$ is not physically implementable. 
However, 
as long as $\hmN$ is Hermitian-preserving, the exponential $ e^{i\hmN(\rho)} $ is always unitary for any $\rho$ that is Hermitian, and so the unitary channel $\sigma \mapsto e^{i\hmN(\rho)} \sigma e^{-i\hmN(\rho)}$ can be implemented (if one has a universal gate set). In particular, the approximation of this unitary $e^{i\hmN(\rho)}$ can be achieved given many copies of the reference state $ \rho $, while the operations required for the implementation are oblivious to $ \rho $.
\begin{defn}[Hermitian-preserving map exponentiation~\cite{Wei2023hermpreserving}]\label{def:HME}
	For a short interval $ s \in \mathbb{R}$, a reference state $\rho$, and a linear Hermitian-preserving map $\hmN$, the Hermitian-preserving map exponentiation channel, acting on a state $ \sigma $ is defined as 
	\begin{align}\label{eq:HME_def}
		\hE{s}{\hmN,\rho}(\sigma) = \Tr_{1}\left[e^{-i\hN s}\left(\rho\otimes\sigma\right)e^{i\hN s}\right]\ ,
	\end{align}
	where $\hN$ is the partial transpose (denoted as $ \transp_{1} $) of the Choi matrix corresponding to $ \hmN $, obtained by applying the map to the second subsystem of an unnormalized maximally entangled state $ \ket{\Phi^{+}}_{12} = \sum_{j}\ket{j}_{1}\ket{j}_{2} $,
	\begin{align}
		\hN = [(\id\otimes\hmN)(\Phi^{+}_{12})]^{\transp_{1}} = \sum_{j,k}\left(\ketbra{k}{j}_{1}\otimes\hmN\left(\ketbra{j}{k}\right)_{2}\right).
	\end{align}
\end{defn}

It is useful to note that in general, the reference state $\rho$ and the target state $\sigma$ in Def.~\ref{def:HME} need not be of the same dimension. 
The map $\hmN$ is in general a map from $\mB(\mH')$ to $\mB(\mH)$ with $\dim(\mH)\neq \dim(\mH')$. 
In such cases, the reference state $\rho \in \mS(\mH')$, while the output of the map $\mN(\rho)\in \mB(\mH)$ and the target state $\sigma \in\mS(\mH)$.
Likewise, the operator $\hN \in \mB(\mH'\otimes\mH)$, since the maximally entangled state used to define the Choi matrix $\ket{\Phi^{+}}_{12}\in \mH'\otimes\mH'$.
The channel Eq.~\eqref{eq:HME_def} well approximates the unitary channel 
\begin{align}
	\Eopch{s}{\hmN,\rho} (\sigma) := e^{is\hmN(\rho)} (\sigma) e^{-is\hmN(\rho)}.
\end{align}
This can be seen by Taylor expanding with respect to $ s $, which yields   
\begin{align}
	\hE{s}{\hmN,\rho}(\sigma) &= \sigma - is\Tr_{1}\left(\left[\hN,\rho\otimes\sigma\right]\right) + \mO(s^{2}) = \sigma -is\Tr_{1}\left(\left[\sum_{j,k}\left(\ketbra{k}{j}_{1}\otimes\hmN\left(\ketbra{j}{k}\right)_{2}\right),\rho\otimes\sigma\right]\right) + \mO(s^{2})\ .
\end{align}
From the linearity of partial trace operation and commutator, 
\begin{align}\label{eq:HME_2ndform}
	\hE{s}{\hmN,\rho}(\sigma) &= \sigma -is\sum_{jk}\left[\Tr_{1}\left(\ketbra{k}{j}_{1}\rho\otimes\hmN\left(\ketbra{j}{k}\right)_{2}\sigma\right) - \Tr_{1}\left(\rho\ketbra{k}{j}_{1}\otimes\sigma\hmN\left(\ketbra{j}{k}\right)_{2}\right) \right] + \mO(s^{2}) \\
	&= \sigma - is\sum_{jk}\left[\bra{j}\rho\ket{k}\hmN\left(\ketbra{j}{k}\right),\sigma\right] + \mO(s^{2}) = \sigma - is\left[\hmN(\rho),\sigma\right] +\mO(s^{2})\ ,\nonumber
\end{align}
where $ \Tr_{1}[\ketbra{j}{k}_{1}\rho\otimes\sigma] =  \Tr_{1}[\rho\ketbra{j}{k}_{1}\otimes\sigma] = \bra{j}\rho\ket{k}\sigma$ is used for the second equality.

For simplicity, we have for now omitted the fact that the magnitude of the error also depends on the map $\hmN$. In particular, Ref.~\cite{Wei2023hermpreserving} proves a more rigorous bound
\begin{align}\label{eq:HME_diamond_dist}
	\left\lVert \hE{s}{\hmN,\rho} - \Eopch{s}{\hmN, \rho}\right\rVert_{\diamond} \leq 8\Vert \hN\Vert_{\infty}^{2}s^{2}\ ,
\end{align}
whenever $\Vert \hN\Vert_{\infty}s \in (0,0.8]$.
Setting $ s\mapsto s/M $ and repeating the channel $M$ times, we can better approximate $\Eopch{s}{\hmN,\rho}$ and achieve
\begin{align}\label{eq:HME_diamond_dist_Mtimes}
	\frac{1}{2}\left\lVert \left(\hE{s/M}{\hmN,\rho}\right)^{M} - \Eopch{s}{\hmN,\rho}\right\rVert_{\Tr}\leq \frac{1}{2}\left\lVert \left(\hE{s/M}{\hmN,\rho}\right)^{M} - \Eopch{s}{\hmN,\rho}\right\rVert_{\diamond} \leq 4\lVert \hN\rVert_{\infty}^{2}\frac{s^{2}}{M} = \mO\left(\lVert \hN\rVert_{\infty}^{2}\frac{s^{2}}{M}\right)\ ,
\end{align}
using the subadditivity $\Vert (\hE{s/M}{\hmN,\rho})^{M} - \Eopch{s}{\hmN,\rho}\Vert \leq M \Vert \hE{s/M}{\hmN,\rho} - \Eopch{s/M}{\hmN,\rho}\Vert$.
Note that the $M$ applications of $ \hE{s/M}{\hmN,\rho} $ costs $ M $ copies of $ \rho $. 
Then Eq.~\eqref{eq:HME_diamond_dist_Mtimes} can be recast as the following statement: the Hermitian-preserving map exponentiation unitary $\Eopch{s}{\hmN,\rho}$ can be implemented with an arbitrarily small error $\epsilon$ by consuming $\mO(\lVert \hN\rVert_{\infty}^{2}s^{2}/\epsilon)$ copies of $\rho$.   

Chronologically, the simplest case of $\hmN = \id$, also known as density matrix exponentiation (DME)~\cite{Lloyd2014quantum, Kimmel2017DME_OP}, was studied prior to the more general scenario of Def.~\ref{def:HME}.
\begin{defn}[Density matrix exponentiation (DME)~\cite{Lloyd2014quantum}]\label{def:DME}
	For a short interval $ s $ and a reference state $\rho$, the density matrix exponentiation channel, acting on a state $ \sigma $ having the same dimension as $ \rho $, is defined as 
	\begin{align}\label{eq:DME_def}
		\hE{s}{\rho}(\sigma) = \Tr_{1}\left[e^{-i\hat{S} s}\left(\rho\otimes\sigma\right)e^{i\hat{S} s}\right]\ ,
	\end{align}
	where $\hat{S}$ is the swap operator yielding $\hat{S}\ket{j}_{1}\ket{k}_{2} = \ket{k}_{1}\ket{j}_{2}$ for all $j,k$.
\end{defn}
The property of the swap operator $ \hat{S}^{2} = 1 $ facilitates the derivation of an explicit form of Eq.~\eqref{eq:DME_def}:
\begin{align}\label{eq:DME_explicit_def}
	\hE{s}{\rho}\left(\sigma\right) = \cos^2(s)\sigma - i\sin(s)\cos(s)[\rho,\sigma] + \sin^2(s)\rho\ .
\end{align} 
Since Def.~\ref{def:DME} is a special case of Def.~\ref{def:HME}, it approximates the corresponding unitary channel $\Eopch{s}{\id,\rho}$ according to Eq.~\eqref{eq:HME_def}, which we denote for simplicity as $\Eopch{s}{\rho}(\sigma) \coloneq e^{-i\rho s}\sigma e^{i\rho s}$.
As a special instance of Eq.~\eqref{eq:HME_diamond_dist} when $\|\hat N\|_\infty = \|\hat S\|_\infty = 1$, \begin{align}\label{eq:DME_diamond_dist_Mtimes}
	\frac{1}{2}\left\lVert \left(\hE{s/M}{\rho}\right)^{M} - \Eopch{s}{\rho}\right\rVert_{\Tr}\leq \frac{1}{2}\left\lVert \left(\hE{s/M}{\rho}\right)^{M} - \Eopch{s}{\rho}\right\rVert_{\diamond}  = \mO\left(\frac{s^{2}}{M}\right)\ .
\end{align}

Now we prove the most general version of a single memory-call operation, namely exponential of polynomials. 
\begin{lem}[Polynomial function exponentiation]\label{lem:PFE_technical}
	Let $f$ be a Hermitian-preserving polynomial of degree $d$.
	The unitary evolution $e^{if(\rho)}$ can be approximated with an arbitrarily small error $\epsilon$, using
	\begin{enumerate}\itemsep0em
		\item $\mO(d\epsilon^{-1}F)$ copies of $\rho$ and 
		\item a circuit oblivious to $\rho$, with depth $ \mO(\epsilon^{-1}F)$, where $F$ is a constant determined by the polynomial $f$. 
	\end{enumerate} 
\end{lem}
\begin{proof} 
	Using Lemma~\ref{lem:poly_corresponding_lin_map}, a linear map $g$, such that $g(\rho^{\otimes d}) = f(\rho)$ for any $f$, exists.
	Then we can approximate $e^{ig(\rho)}$ using Hermitian-preserving map exponentiations $\hE{s}{\hat{G},\rho^{\otimes d}}$, defined in Def.~\ref{def:HME}, where $\hat{G}$ is the partial transpose of the Choi matrix of $g$.
	Following Eq.~\eqref{eq:HME_diamond_dist_Mtimes}, we take $s = M^{-1}$ and concatenate the channel $\hE{M^{-1}}{\hat{G},\rho^{\otimes d}}$ for $M$ times, which leads to the implementation error $\mO(\lVert \hat{G}\rVert_{\infty}^{2}M^{-1})$. 
	Setting $M = \mO(\lVert \hat{G}\rVert_{\infty}^{2}\epsilon^{-1})$, we can achieve the approximation with arbitrarily small error $\epsilon$.
	Each channel $\hE{s}{\hat{G},\rho^{\otimes d}}$ requires a circuit of depth $\mO(1)$ and $\mO(d)$ copies of $\rho$.
	Hence, $M$ concatenations of such channels cost $\mO(d\epsilon^{-1}\lVert \hat{G}\rVert_{\infty}^{2})$ copies of $\rho$ and the circuit depth $\mO(\epsilon^{-1}\lVert \hat{G}\rVert_{\infty}^{2})$.
\end{proof}

\section{Fast spectral convergence and efficient quantum dynamic programming}\label{subsec:main_Thm_proof}

We begin by illustrating a class of fixed-point iterations that admit efficient QDP implementation.
Our recursions operate on general mixed states but unitary; hence, it is natural to consider iterations with a fixed-point shared by isospectral states (states sharing the same eigenvalue spectrum). 
In this section, we denote a generic quantum recursion in the form Eq.~\eqref{eq:dynamicunitary_poly} as $\Upsirec{\rho}$.
\begin{defn}[Fast spectral convergence]\label{def:fast_spectral_convergence}
	Consider a quantum recursion $\Upsirec{\rho}$ that defines a fixed-point iteration
	\begin{align}\label{eq:fp_iteration}
		\rho\mapsto\Upsirec{\rho}\rho \left(\Upsirec{\rho}\right)^\dagger\ .
	\end{align} 
	This iteration \emph{converges spectrally}, with respect to an initial state $\rho_0$, when there exists a fixed-point $\tau$ that is attracting for all $\tilde\rho_0$ satisfying $\spec(\tilde{\rho}_0) = \spec(\rho_0)$.
	In other words, for any sequence $\lbrace\tilde{\rho}_0\rbrace_{n=0}^\infty$ obtained by the iteration, the trace distance $ \delta_{n} \coloneq \frac{1}{2}\lVert\tau - \tilde\rho_{n} \rVert_{1} < \delta_{n-1}$ and $\displaystyle\lim_{n\rightarrow\infty} \delta_{n} = 0$. 
	We say that spectral convergence is \emph{fast}, when 
	\begin{align}
		\delta_{n+1} \leq h(\delta_n) < \delta_n,
	\end{align}
	for a function $h(x)$ whose derivative $h'(x)<r$ for some $0<r<1$ and all relevant $x$. 
\end{defn}
Spectral convergence implies that the QDP implementation will also approach the same fixed-point, when the non-unitary error is not too large. 
Furthermore, the function $h$ quantifies the worst case performance of the algorithm; the distance to the fixed-point is guaranteed to be arbitrarily close to $h^{N}(\delta_{0})$ after $N$ iterations, where $h^{N}$ indicates the $N$ times composition $h\circ\cdots\circ h$.

\subsection{Quantum dynamic programming with a mixed initial state}
We first establish a theorem that demonstrates the efficiency of quantum dynamic programming when the initial and all the instruction states are mixed.

\begin{thm}[Spectral convergence is sufficient for efficient global convergence of QDP]\label{thm:qdp}
	If a quantum recursion $\Upsirec{\rho}$ satisfies fast spectral convergence for some initial state $\rho_0$, then the QDP implementation of $N$ iterations prepares a state $\rho'_{N}$, whose distance to the fixed-point $\delta_{\rho'_{N}} \leq h^{N}(\delta_{0}) + \epsilon $, by using a circuit with depth $\mO(N^{2}\epsilon^{-1})$. 
\end{thm}
\begin{proof}
	Let us denote the unitary channel corresponding to the operator $\Upsirec{\rho}$ as $\Upsirecch{\rho}$.
	The recursion is spectrally converging.
	Hence, there exists a fixed-point $\tau$, such that the distance $ \frac{1}{2}\lVert \tau - \rho \rVert_{1} \eqcolon \delta_{\rho} $ follows
	\begin{align}\label{eq:QDP_assumption_contraction}
		\delta_{\Upsirecch{\rho}(\rho)} = \frac{1}{2}\left\lVert \tau - \Upsirecch{\rho}(\rho) \right\rVert_{1} \leq h(\delta_{\rho}) < \delta_{\rho}\ .  
	\end{align}
	Furthermore, we have  
	\begin{align}\label{eq:h_property_app}
		h(\delta + \epsilon) \leq h(\delta) + r\epsilon\ ,
	\end{align} 
	for some $r<1$ and any $\delta, \epsilon$ of interest from the assumption that the spectral convergence is fast.

	From Lemma~\ref{lem:PFE_technical}, we are able to locally accurately implement the unitary channel $\Upsirecch{\rho}$ by a \emph{non-unitary} channel $\Upsirecapp{\rho}$, such that 
	\begin{align}\label{eq:PFE_channel_eta}
		\frac{1}{2}\left\lVert\Upsirecch{\rho} - \Upsirecapp{\rho}\right\rVert_{\Tr} \leq \eta,
	\end{align}
	for any $\eta>0$, by making $\mO(\eta^{-1})$ memory-usage queries each consuming a copy of $\rho$.
	The circuit depth for this implementation is also $\mO(\eta^{-1})$. 
	
	Suppose that the sequence of states $\{\rho'_{n}\}_{n}$ is obtained from such emulation starting from $\rho'_{0} = \rho^{}_{0}$ and thus recursively defined 
	\begin{align}\label{eq:recursion_without_purification}
		\rho'_{n} = \Upsirecapp{\rho'_{n-1}}\left(\rho'_{n-1}\right)\ .
	\end{align}
	This sequence might deviate from the desired sequence $\{\rho_{n}\}_{n}$ very quickly and, in general, $\spec\{\rho'_{n}\} \neq \spec\{\tau\}$.
	
	We additionally consider a sequence of states $\{\tilde{\rho}_{n}\}_{n}$, defined by the exact unitary recursion with an erroneous instruction state $\tilde{\rho}_{n} = \Upsirecch{\rho'_{n-1}}(\tilde{\rho}_{n-1})$ starting from $\tilde{\rho}_{0} = \rho_{0}$. Hence, $\spec\{\tilde{\rho}_{n}\} = \spec\{\tau\}$.  
	Let us denote $\frac{1}{2}\lVert \tilde{\rho}_{n} - \rho'_{n}\rVert_{1} \eqcolon \varepsilon_{n}$ with $\varepsilon_{0} = 0$.
	We first analyze how $\varepsilon_{n}$ scales.
	Observe that
	\begin{align}\label{eq:varepsilon_scaling}
		\varepsilon_{n+1} = \frac{1}{2}\left \lVert \Upsirecch{\rho'_{n}}(\tilde{\rho}_{n}) - \rho'_{n+1} \right \rVert_{1} \leq \frac{1}{2}\left \lVert \Upsirecch{\rho'_{n}}(\tilde{\rho}_{n}) - \Upsirecch{\rho'_{n}}(\rho'_{n})  \right \rVert_{1} +\frac{1}{2} \left \lVert \Upsirecch{\rho'_{n}}(\rho'_{n}) - \Upsirecapp{\rho'_{n}}(\rho'_{n}) \right \rVert_{1} \leq \varepsilon_{n} + \eta\ ,
	\end{align}
	from the triangle inequality, the unitary invariance of the trace norm, and Eq.~\eqref{eq:PFE_channel_eta}.
	Therefore, $\varepsilon_{n}$ scales linearly as $\varepsilon_{n}\leq n\eta$. 
	The triangle inequality between $\tau$, $\rho'_{n}$, and $\rho_{n}$ yields $\delta_{\rho'_{n}}\leq\delta_{\tilde{\rho}_{n}} + \varepsilon_{n}$, which connects the quantity of our interest $\delta_{\rho'_{n}}$ and the term $\delta_{\tilde{\rho}_{n}}$ derived from a state $\tilde{\rho}_{n}$ isospectral to the desired fixed-point $\tau$.
	
	Next, recall that $\left \lVert e^{if(\rho)} - e^{if(\sigma)}\right \rVert \leq C'\lVert \rho - \sigma\rVert$
	for some constant $C'$, from the mean value theorem for operators.
	This implies 
	\begin{align}\label{eq:MVT}
		\left \lVert \Upsirecch{\rho}- \Upsirecch{\sigma}\right\rVert_{\Tr} \leq C\lVert \rho - \sigma\rVert_{1}\
	\end{align}
	for some constant $C$.
	Then we can write
	\begin{align}\label{eq:delta_eta_scaling}
		\delta_{\tilde{\rho}_{n+1}} \leq \frac{1}{2}\left\lVert  \tau - \Upsirecch{\tilde{\rho}_{n}}\left(\tilde{\rho}_{n}\right) \right\rVert_{1} + \frac{1}{2}\left\lVert \Upsirecch{\tilde{\rho}_{n}}\left(\tilde{\rho}_{n}\right) - \Upsirecch{\rho'_{n}}\left(\tilde{\rho}_{n}\right) \right\rVert_{1} \leq h\left(\delta_{\tilde{\rho}_{n}}\right) + \frac{1}{2}\left\lVert \Upsirecch{\tilde{\rho}_{n}} - \Upsirecch{\rho'_{n}}  \right\rVert_{\Tr} 
		\leq  h\left(\delta_{\tilde{\rho}_{n}}\right) + C\varepsilon_{n},
	\end{align}
	where we use the triangle inequality for the first, the assumption Eq.~\eqref{eq:QDP_assumption_contraction} for the second, and Eq.~\eqref{eq:MVT} for the last inequality. 
	The property of $h$, Eq.~\eqref{eq:h_property_app} gives
	\begin{align}
		\delta_{\tilde{\rho}_{n}} \leq h\left(\delta_{\tilde{\rho}_{n-1}}\right) + C\varepsilon_{n-1}  \leq h\left( h\left(\delta_{\tilde{\rho}_{n-2}}\right) + C\varepsilon_{n-2} \right) + C\varepsilon_{n-1} \leq h\circ h\Big(\delta_{\tilde{\rho}_{n-2}}\Big) + rC\varepsilon_{n-2} + C\varepsilon_{n-1}\ ,
	\end{align}
	which leads to
	\begin{align}\label{eq:delta_tilde_minus_delta_normal}
		\delta_{\tilde{\rho}_{N}} \leq h^{N}\left(\delta_{\rho_{0}}\right) + \sum_{n=0}^{N-1} r^{N-n-1}(Cn\eta) = h^{N}\left(\delta_{\rho_{0}}\right) + \mO(N\eta)\ . 
	\end{align}
	Therefore the final state $\rho'_{N}$ after $N$ iterations has the distance
	\begin{align}
		\delta_{\rho'_{N}} =  h^{N}\left(\delta_{\rho_{0}}\right) + \mO\big(N\eta\big)\ .
	\end{align}
	To obtain the fixed deviation $\epsilon>0$ from $ h^{N}(\delta_{\rho_{0}})$ after $N$ iterations, we can set $\eta = \mO(\epsilon N^{-1})$, which only requires the circuit of total depth $\mO(N^{2}\epsilon^{-1})$. 
	In case $\delta_{\rho_{N}} =  h^{N}(\delta_{\rho_{0}})$, the exact algorithm and the QDP iteration can be made to coincide after $N$ steps with an arbitrarily small deviation $\epsilon$. 
\end{proof}

\subsection{Quantum dynamic programming with a pure initial state}\label{subsec:QDP_pure}
When the initial state of the recursion---and thus also the target final state of the recursion---is pure, a more efficient QDP implementation can be made, by adopting a subroutine that reduces the mixedness of instruction states. 
The initial idea for this subroutine comes from Ref.~\cite{Cirac1999_purification}, which assumed that the system of interest is a qubit. 
Recently, generalizations to systems with an arbitrary dimension was proposed~\cite{Childs2024purification, Li2024PurityAmplification}, albeit with a less straightforward mixedness reduction performance. 
We adapt the protocol in Ref.~\cite{Childs2024purification} here.
\begin{lem}[Interferential mixedness reduction (IMR)]\label{lemma:purification_genearlised}
	Let $\rho = \sum_{i}\lambda_{i}\dm{v_{i}}$ be a $d$-dimensional density matrix with eigenvalues $\lambda_{1}>\lambda_{2}\geq\cdots\geq\lambda_{d}$ and an orthonormal basis $\{\ket{v_{i}}\}$.
	Using an auxiliary qubit and a controlled swap operation on the two copies of $\rho$, a state $\rho' = \lambda'_{i}\dm{v_{i}}$ with $\lambda'_{1} = \frac{1+\lambda_{1}}{1+\sum_{i}\lambda_{i}^{2}}\lambda_{1}>\lambda_{1}$ can be prepared with a probability $\frac{1}{2}(1+\sum_{i}\lambda_{i}^{2})$.
	When the mixedness parameter $x \coloneq 1-\lambda_{1}$ is small, $x' \coloneq 1-\lambda'_{1} = \frac{x}{2} + \mO(x^{2})$, effectively halving the mixedness of the state $\rho$, with the success probability larger than $(1 - x)$.
	This procedure can be done without any prior knowledge of the state $\rho$. 
\end{lem}
\begin{proof}
	Let us denote the space each copy of $\rho$ lives in as $1$ and $2$.
	The main idea of~\cite{Cirac1999_purification} and~\cite{Childs2024purification} is to apply the projector $ \frac{\1_{12}+S_{12}}{2} $, where $ \1_{12} $ is the identity on the total space of two copies and $ S_{12} $ is a swap between the first and the second copies. 
	This operation can be implemented by a linear combination of unitaries utilizing an additional ancilla qubit $ \ket{0}_{A} $~\cite{Long2006LCU, Childs2012LCU}, since both $ \1_{12} $ and $ S_{12} $ are unitaries. 
	To a state $ \dm{0}_{A}\otimes \rho^{\otimes 2} $ appended with ancilla, apply a unitary 
	\begin{align}
		U &= \left(H_{A}\otimes\1_{12}\right)\left(\dm{0}_{A}\otimes \1_{12} + \dm{1}_{A}\otimes S_{12}\right) \left(H_{A}\otimes\1_{12}\right)\nonumber \\&= \left(\dm{0}_{A}+\dm{1}_{A}\right)\otimes\frac{\1_{12}+S_{12}}{2} + \left(\ketbra{1}{0}_{A}+\ketbra{0}{1}_{A}\right)\otimes\frac{\1_{12}-S_{12}}{2} \ ,
	\end{align} 
	by concatenating Hadamard operators $ H_{A} $ acting on the ancilla and a controlled swap operation, applying the swap $ S_{12} $ between two copies conditioned on the ancilla state $ \ket{1} $.
	After this operation, the ancilla plus system state becomes
	\begin{align}
		U\left(\dm{0}_{A}\otimes \rho^{\otimes 2}\right)U^{\dagger} = &\dm{0}_{A}\otimes \left(\frac{\1_{12}+S_{12}}{2}\right)\rho^{\otimes 2}\left(\frac{\1_{12}+S_{12}}{2}\right) + \dm{1}_{A}\otimes \left(\frac{\1_{12}-S_{12}}{2}\right)\rho^{\otimes 2}\left(\frac{\1_{12}-S_{12}}{2}\right)\nonumber\\
		&+\left[ \ketbra{1}{0}_{A}\otimes\left(\frac{\1_{12}-S_{12}}{2}\right)\rho^{\otimes 2}\left(\frac{\1_{12}+S_{12}}{2}\right)+ \mathrm{h.c.}\right]\ .
	\end{align}
	By measuring the ancilla in $ \{\ket{0}_{A},\ket{1}_{A}\} $ basis, we effectively achieve (unnormalized) system states 
	\begin{align}\label{eq:postmeasurement_states}
		\begin{cases}
			\frac{1}{4}[2(\rho\otimes\rho) + S_{12}(\rho\otimes\rho) +(\rho\otimes\rho)S_{12}]   ,\quad &\text{if $ \ket{0}_{A} $ is measured,}\\
			\frac{1}{4}[2(\rho\otimes\rho) - S_{12}(\rho\otimes\rho) -(\rho\otimes\rho)S_{12}] \ , \quad &\text{if $ \ket{1}_{A} $ is measured,}
		\end{cases}
	\end{align}
	with probability given by their traces $\frac{1}{2}(1\pm\sum_{i}\lambda_{i}^{2})$.
	
	If the second system is traced out from Eq.~\eqref{eq:postmeasurement_states}, the remaining state becomes
	\begin{align}
		\rho' \propto
		\begin{cases}
			\frac{1}{2}(\rho + \rho^{2})   ,\quad &\text{if $ \ket{0}_{A} $ is measured,}\\
			\frac{1}{2}(\rho - \rho^{2}) \ , \quad &\text{if $ \ket{1}_{A} $ is measured.}
		\end{cases}
	\end{align}
	For our purpose, the protocol is successful when $ \ket{0}_{A} $ is measured, in which case the normalized density matrix
	\begin{align}\label{key}
		\rho' = \frac{\rho+\rho^{2}}{1+\Tr[\rho^{2}]} = \sum_{i}\frac{\lambda_{i}+\lambda_{i}^{2}}{1+\sum_{j}\lambda_{j}^{2}}\dm{v_{i}} \eqcolon \sum_{i}\lambda'_{i}\dm{v_{i}}
	\end{align}
	remains, and $\lambda'_{1} = \frac{1+\lambda_{1}}{1+\sum_{i}\lambda_{i}^{2}}\lambda_{1}$.
	From $\lambda_{1} = \lambda_{1}\sum_{i}\lambda_{i} > \sum_{i}\lambda_{i}^{2}$, we obtain $\lambda'_{1}>\lambda_{1}$ whenever $\lambda_{1}$ is the largest eigenvalue of $\rho$. 
	
	Finally, we consider the case $(1-\lambda_{1})\ll 1$.
	First, observe that 
	\begin{align}
		\frac{x'}{x} = \frac{1-\lambda'_{1}}{1-\lambda_{1}}= \frac{1+\sum_{i}\lambda_{i}^{2} - \lambda_{1}-\lambda_{1}^{2}}{(1+\sum_{i}\lambda_{i}^{2})(1-\lambda_{1})} = \frac{1- \lambda_{1}+\sum_{i>1}\lambda_{i}^{2} }{(1+\sum_{i}\lambda_{i}^{2})(1-\lambda_{1})} \leq  \frac{1- \lambda_{1}+(1-\lambda_{1})^{2} }{(1+\sum_{i}\lambda_{i}^{2})(1-\lambda_{1})} = \frac{1+(1-\lambda_{1})}{1+\sum_{i}\lambda_{i}^{2}}, 
	\end{align}
	where we used the Hölder's inequality $\sum_{k > 1}\abs{y_{k}z_{k}}\leq (\sum_{k>1} \abs{y_{k}}^{p})^{1/p} (\sum_{k>1} \abs{z_{k}}^{q})^{1/q}$ for vectors $y_{k} = z_{k} = \lambda_{k}$ and $p = q = 2$, which gives $\sum_{i>1}\lambda_{i}^{2}\leq (\sum_{i>1}\lambda_{i})^{2} = (1-\lambda_{1})^{2}$.
	On the other hand, by using $0 \leq \sum_{i>1}\lambda_{i}^{2}$, the ratio becomes
	\begin{align}\label{eq:x_ratio_upperbound}
		\frac{x'}{x}  \leq \frac{1+(1-\lambda_{1})}{1+\lambda_{1}^{2}} =  \frac{1+x}{2-2x + x^{2}} =  \frac{1}{2} + \mO(x). 
	\end{align}
	Second, the success probability $\frac{1}{2}(1+\sum_{i}\lambda_{i}^{2})\geq \frac{1}{2}(1+\lambda_{1}^{2})\geq\lambda_{1} = 1-x$, which concludes the proof.
\end{proof}

Using IMR protocol in Lemma~\ref{lemma:purification_genearlised}, we can efficiently combat the accumulation of the non-unitary error $\eta$, by spending more copies after each QDP implementation of $\Upsirecapp{\rho}$ to reduce the mixedness of the resulting state. 
Furthermore, it is guaranteed that the state can be arbitrarily close to a pure state by repeating the IMR subroutine. 
$R$ rounds of IMR protocl form a mixedness reduction subroutine that suppresses the mixedness parameter $x \coloneq 1 - \lambda_{1}$ under a certain threshold.
\begin{lem}[QDP mixedness reduction subroutine]\label{lemma:mixedness_reduction_subroutine}
	Let $ \rho $ be a density matrix with the largest eigenvalue $\lambda_{1} = 1-x$ for some mixedness parameter $x \in [0,\frac{1}{3}]$, corresponding to the eigenvector $\ket{v_{1}}$. 
	Given any maximum tolerable failure probability $q_\mathrm{th}\geq 0$, one can choose any parameter $\mathsf{g}\in\mathbb{R}$ and prepare $M$ copies of $ \rho' $, such that:
	\begin{enumerate}
		\item $\rho'$ has the largest eigenvalue $\lambda'_{1} \geq 1-\frac{x}{\mathsf{g}}$, corresponding to the same eigenvector $\ket{v_{1}}$. 
		\item a number of $ R = \mO(\log(\mathsf{g})) $ IMR rounds in Lemma \ref{lemma:purification_genearlised} is used,
		\item a total of $M \times (\frac{2}{c})^{R}$ copies of $ \rho $ is consumed, 
		with $ c\in(0,1) $ satisfying
		\begin{align}\label{eq:c_for_mixedness_reduction}
			c = 1 - x - M^{-\frac{1}{2}}\sqrt{\log\left(\frac{R}{q_{\rm th}}\right)}
		\end{align}
		\item the success probability of the entire subroutine is $q_{\rm succ} \geq 1 - q_{\rm th}$.
	\end{enumerate}
	Note that $M$ must be sufficiently large to guarantee $c>0$.
\end{lem}
\begin{proof}
	To prove this lemma, we analyze $R$ sequential rounds of the IMR protocol, with initial and final states $\rho$ and $\rho'$. 
	Let us denote the intermediate states generated via the IMR protocols as $\lbrace \chi_j\rbrace_{j=0}^R$, where $\chi_0 = \rho$ and $\chi_R = \rho'$; and let $\lbrace y_j \rbrace_{j=0}^R$ be the mixedness parameter (i.e. $1$ minus the largest eigenvalue) for each $\chi_{j}$ with $y_{0} = x$ and $y_{R} \leq \frac{x}{\mathsf{g}}$. 
	Furthermore, let us denote $M_j$ to be the number of copies of an intermediate state $\chi_j$ generated; hence, we obtain $M_{R} = M$ copies of the desired state $\rho'$ at the end, and aim to show that $M_{0} = M\times (\frac{2}{c})^{R}$ suffices. 
	
	From Eq.~\eqref{eq:x_ratio_upperbound}, the mixedness parameter after $j$ rounds of the IMR protocol becomes $y_{j}\leq \frac{1+y_{j-1}}{2-2y_{j-1}+y_{j-1}^{2}}y_{j-1}$. 
	The final reduction $ \frac{x'}{x} = \frac{y_{R}}{y_{0}} \leq \mathsf{g}^{-1} $ is achievable if $R$ is sufficiently big to satisfy
	\begin{align}
		\frac{y_{R}}{y_{0}} = \prod_{j=1}^{R}\frac{y_{j}}{y_{j-1}} \leq \prod_{j=1}^{R} \frac{1+y_{j-1}}{2-2y_{j-1}+y_{j-1}^{2}} \leq \left(\frac{1+x}{2-2x+x^{2}}\right)^{R} \leq \mathsf{g}^{-1}\ .
	\end{align}
	The second last inequality follows from two facts: i) the factor $\frac{1+y}{2-2y+y^{2}}$ monotonically increases as $y$ increases ii) $y_{j}\leq y_{0}$ for all $j$.
	Since $\log(\frac{1+x}{2-2x+x^{2}})>0$ by the assumption $x\leq \frac{1}{3}$, the number of rounds $R =  \mO(\log(\mathsf{g}))$ is sufficient as claimed. 
	
	The parameter $c$ in the statement of the lemma can be interpreted as the `survival rate' after each round of the IMR protocol.
	More specifically, the $j$-th round of the protocol is successful if at least $ \frac{c}{2}M_{j-1} $ output states $ \chi_{j} $ is prepared from $ M_{j-1} $ copies of $ \chi_{j-1} $.
	If successful, we discard all the surplus output copies and set $ M_{j} = \frac{c}{2}M_{j-1} $.
	If not, we declare that the whole subroutine has failed.
	Hence, if the entire subroutine succeeds, the final number of copies $ M = M_{R} $ is related to the initial number of copies as $M_{0} = (2/c)^{R} M$. 
	
	Finally, we estimate the success probability $ q_{\rm succ} $ given $ R, M, c, x $ and require it to be lower bounded with $ (1 - q_{\rm th}) $.
	The failure probability of each round will be bounded using Hoeffding's inequality. 
	From Lemma~\ref{lemma:purification_genearlised}, the success probability of preparing one copy of $ \chi_{j} $ from a pair $ \chi_{j-1}^{\otimes 2} $ has a lower bound $p(y_{j-1}) \geq 1-y_{j-1}$.
	At each round, $ \frac{M_{j-1}}{2} $ independent trials of this IMR protocol are conducted.
	The probability of more than $ c\frac{M_{j-1}}{2} $ attempts succeeds, i.e. the success probability of a round, is
	\begin{align}
		q_{r}^{(j)} = 1 - F\left(\left\lceil\frac{c M_{j-1}}{2}-1\right\rceil;\frac{M_{j-1}}{2},p(y_{j-1})\right)\geq   1 - F\left(\frac{c M_{j-1}}{2};\frac{M_{j-1}}{2},1-y_{j-1}\right)\ ,
	\end{align}
	where $ F(k;n,q) $ is the cumulative binomial distribution function defined as
	\begin{align}
		F(c n;n,p) = \sum_{i=0}^{\lfloor c n\rfloor}\binom{n}{i}p^{i}(1-p)^{n-i}\ .
	\end{align}
	Furthermore, the Hoeffding's inequality gives $F(c n;n,p) \leq e^{-2n(p-c)^{2}}$, and consequently sets the bound
	\begin{align}
		q_{r}^{(j)} \geq 1 - e^{-M_{j}\left(1-y_{j-1} - c\right)^{2}}\ .
	\end{align}
	The success probability of an entire subroutine can then be bounded as
	\begin{align}
		q_{\rm succ} &= \prod_{j = 1}^{R}q_{r}^{(j)} \geq \prod_{j=1}^{R}\left(1 - e^{-M_{j}\left(1-y_{j-1} - c\right)^{2}}\right) >1 - \sum_{j=1}^{R}e^{-M_{j}\left(1-y_{j-1} - c\right)^{2}}\ ,
	\end{align}
	where the last inequality uses $ \prod_{j}(1-q_{j})>1-\sum_{j}q_{j} $ that holds when $ q_{j}\in(0,1) $ for all $ j $.
	Moreover, by recalling $ x\geq y_{j-1} $ and $ M\leq M_{j} $ for all $ j $, the final bound
	\begin{align}\label{eq:qsucc_bound_c}
		q_{\rm succ}> 1 - Re^{-M\left(1-x- c\right)^{2}} 
	\end{align}
	is obtained. 
	The prescribed failure threshold $ q_{\rm th} $ holds when 
	$e^{-M\left(1-x - c\right)^{2}} = \frac{q_{\rm th}}{R}$ ,
	or equivalently, when $M$ is sufficiently large and Eq.~\eqref{eq:c_for_mixedness_reduction} is true. 
\end{proof}
During the earlier recursions, $M$ is always sufficiently large to guarantee that $c$ is well separated from $0$. 
For later recursions, in particular for the last recursion where only $M = 1$ is needed, Eq.~\eqref{eq:c_for_mixedness_reduction} can become close to zero, or even negative. 
To guarantee positive and non-vanishing $c$, we allow some redundancy of the final state $\rho_{N}$ and prepare $M>1$ copies of it.  

Now we demonstrate the advantage of using mixedness reduction subroutines for pure state QDP. 
\begin{thm}[Efficient QDP for pure states]\label{thm:qdp_pure}
	Suppose that a quantum recursion $\Upsirec{\psi}$ satisfies fast spectral convergence for some pure initial state $\psi_0 = \dm{\psi_0}$.
	Then for any $p_{\mathrm{th}}>0$, the QDP implementation of $N$ iterations prepares a state $\rho'_{N}$ with success probability $1-p_{\rm th}$, such that the distance of $\rho'_{N}$ to the pure fixed-point is 
	\begin{align}
		\delta_{\rho'_{N}} \leq h^{N}(\delta_{0}) + \epsilon ,
	\end{align}
	and $0<\epsilon<\frac{2}{3}$.
	This implementation requires $\epsilon^{-N}e^{\mO(N)}\log(p_{\rm th}^{-1})$ copies of $\psi_{0}$ and a circuit with depth $\mO(N\epsilon^{-1})$. 
\end{thm}
\begin{proof}
	The target state, i.e. the fixed-point, is again denoted as $\tau$ and the distance from it as $\frac{1}{2}\lVert \tau - \rho \rVert_{1} \eqcolon \delta_{\rho}$.
	Eqs.~\eqref{eq:QDP_assumption_contraction}--\eqref{eq:PFE_channel_eta} still holds, yet the recursion starting from $\rho'_{0} = \psi_{0}$ is now defined as
	\begin{align}
		\rho'_{n} = \mathrm{IMR}\circ\Upsirecapp{\rho'_{n-1}}(\rho'_{n-1}),
	\end{align}
	where $\mathrm{IMR}$ denotes the $R$ rounds of the IMR subroutine as in Lemma~\ref{lemma:mixedness_reduction_subroutine}. 
	Suppose that $R$ is sufficiently large to guarantee that $\rho_{n}$ has the largest eigenvalue $1-x$ with a mixedness parameter $x<\varepsilon$ for some small $\varepsilon$ and the corresponding eigenvector $\ket{\tilde{\psi}_{n}}$.
	
	Then we can write
	\begin{align}\label{eq:delta_rho_prime_tilde_psi}
		\delta_{\rho'_{n}} \leq \frac{1}{2}\left\lVert \tau - \tilde{\psi}_{n} \right\rVert_{1} + \frac{1}{2}\left\lVert \tilde{\psi}_{n} - \rho'_{n} \right\rVert_{1} \leq  \frac{1}{2}\left\lVert \tau - \tilde{\psi}_{n} \right\rVert_{1} + \varepsilon = \delta_{\tilde{\psi}_{n}} + \varepsilon\ ,
	\end{align}
	where triangle inequality and the assumption on the mixedness of $\rho'_{n}$ give rise to the first and the second inequality, respectively. 
	Then the important quantity becomes $\delta_{\tilde{\psi}_{n}}$, which evolves as
	\begin{align}
		\delta_{\tilde{\psi}_{n+1}}= \frac{1}{2}\left\lVert \tau - \tilde{\psi}_{n+1} \right\rVert_{1} \leq \frac{1}{2}\left\lVert \tau - \Upsirecch{\tilde{\psi}_{n}}(\tilde{\psi}_{n}) \right\rVert_{1} + \frac{1}{2}\left\lVert \tilde{\psi}_{n+1} - \Upsirecch{\tilde{\psi}_{n}}(\tilde{\psi}_{n}) \right\rVert_{1} \leq h(\delta_{\tilde{\psi}_{n}}) + \frac{1}{2}\left\lVert \tilde{\psi}_{n+1} - \Upsirecch{\tilde{\psi}_{n}}(\tilde{\psi}_{n}) \right\rVert_{1}, 
	\end{align}
	where the second inequality comes from Eq.~\eqref{eq:QDP_assumption_contraction}.
	The second term can be further bounded by the triangle inequality 
	\begin{align}\label{eq:psi_tilde_evolve_triangle}
		\frac{1}{2}\left\lVert \tilde{\psi}_{n+1} - \Upsirecch{\tilde{\psi}_{n}}(\tilde{\psi}_{n}) \right\rVert_{1} \leq \frac{1}{2}\left\lVert \tilde{\psi}_{n+1} - \Upsirecapp{\rho'_{n}}(\rho'_{n}) \right\rVert_{1} + \frac{1}{2}\left\lVert \Upsirecapp{\rho'_{n}}(\rho'_{n}) - \Upsirecch{\tilde{\psi}_{n}}(\tilde{\psi}_{n})\right\rVert_{1}\ . 
	\end{align}
	Note that IMR subroutines do not change the eigenvector corresponding to the largest eigenvalue of the state; hence $\tilde{\psi}_{n+1}$ is the eigenvector of $\Upsirecapp{\rho'_{n}}(\rho'_{n})$ corresponding to the largest eigenvalue of it. 
	The first term of the bound Eq.~\eqref{eq:psi_tilde_evolve_triangle} is then the mixedness parameter of $\Upsirecapp{\rho'_{n}}(\rho'_{n})$.
	Since $\rho'_{n}$ has the mixedness parameter smaller than $\varepsilon$ and $\Upsirecapp{\rho'_{n}}$ is only $\eta$ different from the unitary operator $\Upsirecch{\rho'_{n}}$, we obtain $\frac{1}{2}\lVert \tilde{\psi}_{n+1} - \Upsirecapp{\rho'_{n}}(\rho'_{n}) \rVert_{1} \leq \varepsilon + \eta$.
	The second term of the bound Eq.~\eqref{eq:psi_tilde_evolve_triangle} can be further expanded as
	\begin{align}
		\frac{1}{2}\left\lVert \Upsirecapp{\rho'_{n}}(\rho'_{n}) - \Upsirecch{\tilde{\psi}_{n}}(\tilde{\psi}_{n})\right\rVert_{1} &\leq \frac{1}{2}\left\lVert \Upsirecapp{\rho'_{n}}(\rho'_{n}) - \Upsirecapp{\rho'_{n}}(\tilde{\psi}_{n}) \right\rVert_{1} + \frac{1}{2}\left\lVert \Upsirecapp{\rho'_{n}}(\tilde{\psi}_{n}) - \Upsirecch{\rho'_{n}}(\tilde{\psi}_{n}) \right\rVert_{1} + \frac{1}{2}\left\lVert \Upsirecch{\rho'_{n}}(\tilde{\psi}_{n}) - \Upsirecch{\tilde{\psi}_{n}}(\tilde{\psi}_{n}) \right\rVert_{1} \nonumber \\
		&\leq \frac{1}{2}\left\lVert \rho'_{n} - \tilde{\psi}_{n} \right\rVert_{1} + \frac{1}{2}\left\lVert \Upsirecapp{\rho'_{n}} - \Upsirecch{\rho'_{n}} \right\rVert_{\Tr} + \frac{1}{2}\left\lVert \Upsirecch{\rho'_{n}} - \Upsirecch{\tilde{\psi}_{n}} \right\rVert_{\Tr} \leq C\varepsilon + \eta\ ,
	\end{align}
	for some constant $C$.
	The second and third inequalities follow from data processing inequality, definition of the channel distance, and Eqs.~\eqref{eq:PFE_channel_eta} and~\eqref{eq:MVT}.
	Combining everything, 
	\begin{align}
		\delta_{\tilde{\psi}_{n+1}} \leq h(\delta_{\tilde{\psi}_{n}}) + (C+1)\varepsilon + 2\eta\ .
	\end{align}
	Using Eq.~\eqref{eq:h_property_app}, and putting $z =  (C+1)\varepsilon + 2\eta$,
	\begin{align}
		\delta_{\tilde{\psi}_{N}} \leq h(\delta_{\tilde{\psi}_{N-1}}) + z \leq h(h(\delta_{\tilde{\psi}_{N-2}}) + z) + z \leq h^{2}(\delta_{\tilde{\psi}_{N-2}}) + rz +z \leq \cdots \leq h^{N}(\delta_{0}) + \sum_{k = 0}^{N-1}r^{k}z = h^{N}(\delta_{0}) + \mO(\varepsilon) + \mO(\eta).
	\end{align}
	From Eq.~\eqref{eq:delta_rho_prime_tilde_psi}, 
	\begin{align}
		\delta_{\rho'_{N}} \leq h^{N}(\delta_{0}) + \mO(\varepsilon) + \mO(\eta).
	\end{align}
	
	We can get a desired result $	\delta_{\rho'_{N}} \leq h^{N}(\delta_{0}) + \epsilon $ by setting $\eta = \mO(\epsilon)$ and $\varepsilon = \mO(\epsilon)$.
	For simplicity, let us choose $\eta = \varepsilon$. 
	$R$ rounds of IMR need to reduce the mixedness parameter of $\Upsirecapp{\rho'_{n-1}}(\rho'_{n-1})$ to that of $\rho'_{n}$ each of which is lower bounded by $2\varepsilon$ and $\varepsilon$, i.e. $\mathsf{g} = 2$.
	Since $\varepsilon = \mO(\epsilon)$ is small, $R = \mO(1)$ suffices. 
	The requirement that $\eta = \mO(\epsilon)$ can be met by  using $\mO(\epsilon^{-1})$ depth circuit for each $\Upsirecapp{\rho'_{n}}$ implementation. 
	Hence, the overall circuit depth for $N$ recursions will be $\mO(N\epsilon^{-1})$.
	
	Finally, we derive the number of initial state copies needed for the algorithm. 
	Let us denote $\mathtt{I}_{n}$ and $\mathtt{O}_{n}$ to be the number of $\rho'_{n-1}$ and $\Upsirecapp{\rho'_{n-1}}(\rho'_{n-1})$ copies before and after $n$'th recursion step. 
	We require $\mathtt{I}_{N+1}\geq1$ and $\mathtt{I}_{1}$ will be the number of $\psivec{0}$ copies we need.
	As the implementation of $\Upsirecapp{\rho'_{n-1}}$ requires $\mO(\eta^{-1}) = \mO(\epsilon^{-1})$ copies, 
	\begin{align}
		\mathtt{I}_{n} = \mO(\epsilon^{-1})\mathtt{O}_{n}.
	\end{align}

	$\mathtt{I}_{n+1}$ and $\mathtt{O}_{n}$ are related by the number of mixedness reduction rounds $R$ and survival rate $c$ as
	\begin{align}\label{eq:On_in_QDPpureThm}
		\mathtt{O}_{n} = \left(\frac{2}{c}\right)^{R}\mathtt{I}_{n+1}.
	\end{align}
	$R = \mO(1)$, while $c$ can be determined by the success probability $q_{\rm succ}$ of $R$ round mixedness reduction subroutines using Eq.~\eqref{eq:c_for_mixedness_reduction} with $M = \mathtt{I}_{n+1}$. 
	Assume that we impose the success probability $q_{\rm succ}$ to be higher than $1 - q_{\rm th}$ for some threshold value $ q_{\rm th} $. 
	The probability of all $N$ mixedness reduction protocols to be successful is
	\begin{align}\label{eq:pth_in_QDPpureThm}
		p_{\rm succ} = q_{\rm succ}^{N} \geq \left(1 - q_{\rm th}\right)^{N} \geq 1 - Nq_{\rm th}\ ,
	\end{align}
	where the second inequality follows from the Taylor expansion. 
	Thus setting $ q_{\rm th}= \frac{p_{\rm th}}{N} $ ensures that $ p_{\rm succ} \geq 1 - p_{\rm th} $ as required. 
	
	Now we fix the value of $c$.
	Recalling Eq.~\eqref{eq:c_for_mixedness_reduction} and $\varepsilon < \frac{\epsilon}{2} < \frac{1}{3}$, we have
	\begin{align}\label{eq:fix_c}
		\mathtt{I}_{N+1}^{-\frac{1}{2}}\sqrt{\log(\frac{R}{q_{\rm th}})} < \frac{1}{6} \quad \Rightarrow\quad c = \frac{1}{2},
	\end{align}
	or $\mathtt{I}_{N+1} = \mO(\log(p_{\rm th}^{-1}N))$ implies $c = \frac{1}{2}$.
	Combining everything, 
	\begin{align}
		\mathtt{I}_{n} = \epsilon^{-1}e^{\mO(1)}\mathtt{I}_{n+1}\ ,
	\end{align}
	and therefore
	\begin{align}
		\mathtt{I}_{1} = \epsilon^{-N}e^{\mO(N)}\mathtt{I}_{N+1} = \epsilon^{-N}e^{\mO(N)}\mO(\log(p_{\rm th}))\ ,
	\end{align}
	which concludes the proof. 
\end{proof}

\section{Analysis of the nested fixed-point Grover search}\label{app:Grover}

In this section, we delineate the entire protocol for the QDP implementation of nested fixed-point Grover search recursions, using it as an illustrative example. 
The exact recursion is analyzed in Sec.~\ref{appendix:nested} with the explicit formula giving the distance between the resulting state and the target state after each recursion.
QDP implementation of the recursion is carried out by combining two subroutines: approximation of the recursion unitary using DME memory-usage queries and the interferential mixedness reduction (IMR) subroutine that maintains the mixedness of the instruction states below some threshold.
The first subroutine is analysed in detail in Sec.~\ref{app:QDP_NFGS_convergence}.
In the rest of the section, we calculate the total cost of this implementation, namely the depth of the circuit and the number of copies needed, and show that the final state can be made arbitrarily close to the desired final state.

\subsection{Exact implementation of nested fixed-point Grover search}\label{appendix:nested}
The analysis of the nested fixed-point Grover search, introduced in Ref.~\cite{Yoder2014Grover}, involves Chebyshev polynomials of the first kind
\begin{align}\label{eq:Chebyshev_poly}
	T_{m}(x) = 
	\begin{cases}
		\cos\bigl(m\arccos(x)\bigr) , &\abs{x}\leq 1\ , \\
		\cosh\bigl(m\arcosh(x)\bigr),& \abs{x} >1\ .
	\end{cases}
\end{align}
We observe that if $ \abs{x} \leq 1 $ then
$ \abs{T_{m}(x)} \leq 1 $ and $ T_{l}(T_{m}(x)) = T_{lm}(x)$ .
Choosing $l = m^{-1}$ gives a special case
\begin{align}\label{eq:Chebyshev_multiplication}
	T_{m^{-1}}\left(T_{m}(x)\right) = T_1(x)= x,
\end{align}
and we will also use the property that $T_{m}(1) = 1$.

We begin with $ n = 0 $, which is to specify that we are given the initial state $ \psivec{0} $, and target state $\ket{\tau}$. 
In the $ n $th step of the iteration, we apply a composite unitary operation with $L$ memory-calls
\begin{align}\label{eq:Gammaop_redef}
	\Gammaop{L}{\psi_{n-1}} = \Eop{\alpha_{L}}{\psi_{n-1}}\Eop{\beta_{L}}{\tau}\cdots  \Eop{\alpha_{1}}{\psi_{n-1}}\Eop{\beta_{1}}{\tau}
\end{align}
to the iterated state $ \psivec{n-1} $ obtaining $\psivec{n} = \Gammaop{L}{\psi_{n-1}}\psivec{n-1}$.
The composite unitary operation in Eq.~\eqref{eq:Gammaop_redef} consists of partial reflection unitaries defined as
\begin{align}\label{key}
	\Eop{s}{\psi} \coloneq \1 - (1-e^{-is})\dm{\psi}\ ,
\end{align}
with angle $ s $, where we have adapted the notation $ \psi = \dm{\psi} $ and $ \tau = \dm{\tau} $ for pure states. 
When $ s = \pi $, we obtain the usual reflection unitary around the state $ \ket{\psi} $. 
In Eq.~\eqref{eq:Gammaop_redef}, following in the steps of ~\cite{Yoder2014Grover}, we consider $ L $ partial reflections around both $ \ket{\tau} $ and $ \psivec{n-1} $ with angles
\begin{align}\label{eq:angles}
	\alpha_{l}(q) = -\beta_{L-l+1}(q) = 2\cot^{-1}\left(\tan\left(\frac{2\pi l}{2L+1}\right)\sqrt{1-\left[T_{1/(2L+1)}(q)\right]^{-2}}\right)\ ,
\end{align}
defined with some number $ q $.
The choice of the parameter $ q $ determines the final distance to the target state as~\cite{Yoder2014Grover}
\begin{align}\label{eq:distance_recursion_1}
	\frac{1}{2}\left\lVert \tau - \Gammaop{L}{\psi}(\psi)\left(\Gammaop{L}{\psi}\right)^{\dagger}\right\rVert_{1} = qT_{2L+1}\left(T_{1/(2L+1)}(q^{-1})\frac{1}{2}\lVert \tau - \psi\rVert_{1}\right) \ .
\end{align}

\begin{lem}\label{lem:NFGS_good_parameter_q}
	Suppose that the parameter $q$ in Eq.~\eqref{eq:angles} is chosen to be $q = \delta_{n}$ for the recursion unitary $\Gammaop{L}{\psi_{n-1}}$ at any $n$, where 
	\begin{align}\label{eq:modified_step_error_app}
		\delta_{n} = \sech\left((2L+1)\arcsech\left(\delta_{n-1}\right)\right)\ ,\quad \forall n>0\ ,
	\end{align}
	starting from the base $ \delta_{0} = \frac{1}{2}\lVert \tau - \psi_{0}\rVert_{1} $.
	Then 
		$\frac{1}{2}\left\lVert \tau - \psi_{n}\right\rVert_{1} =  \delta_{n}$
		for all $n$, where $\psi_{n}$ is recursively defined as $\psivec{n} = \Gammaop{L}{\psi_{n-1}}\psivec{n-1}$.
	\end{lem}
	\begin{proof}
		We first observe that Eq.~\eqref{eq:modified_step_error_app} is equivalent to
		\begin{align}\label{eq:delta_n_exact}
			\delta_{n}^{-1}= T_{2L+1}\left(\delta_{n-1}^{-1}\right)\ ,
		\end{align}
		by recalling that ${\rm arcsech}(x) = {\rm arcosh}(x^{-1})$. In particular, this also means that by applying $T_{1/(2L+1)}$ to both sides, we get
		\begin{align}\label{eq:conc_Ts}
			T_{1/(2L+1)}(\delta_{n}^{-1}) = T_{1/(2L+1)}\left(T_{2L+1}\left(\delta_{n-1}^{-1}\right)\right) = \delta_{n-1}^{-1}.
		\end{align}
		We may now use this identity in simplifying the trace distance of interest: starting from the general relation Eq.~\eqref{eq:distance_recursion_1} with $q = \delta_{n}$,
		\begin{align}\label{eq:distance_recursion_2}
			\frac{1}{2}\lVert \tau - \psi_{n}\rVert_{1} = \delta_{n} T_{2L+1}\left(T_{1/(2L+1)}\bigl(\delta_{n}^{-1}\bigr)\frac{1}{2}\lVert \tau - \psi_{n-1}\rVert_{1}\right) = \delta_{n} T_{2L+1}\left(\delta_{n-1}^{-1}\frac{1}{2}\lVert \tau - \psi_{n-1}\rVert_{1}\right)\ ,
		\end{align}
		where the second equality comes from applying Eqs.~\eqref{eq:conc_Ts}.
		If the Lemma is true for step $n-1$, i.e. if $\frac{1}{2}\lVert \tau - \psi_{n-1}\rVert_{1} = \delta_{n-1}$, then Eq.~\eqref{eq:distance_recursion_2} leads to $\frac{1}{2}\lVert \tau - \psi_{n}\rVert_{1} = \delta_{n}$ using $T_{2L+1}(1) = 1$. 
		Since the base case $\frac{1}{2}\lVert \tau - \psi_{0} \rVert_{1} = \delta_{0}$ is true, the lemma holds by induction. 
	\end{proof}
	
	The explicit form of the $\delta_N$, i.e. the distance to the target state after $N$ iterations can also be derived. 
	From Eqs.~\eqref{eq:Chebyshev_multiplication} and~\eqref{eq:delta_n_exact}, $\delta_{N}^{-1} = T_{(2L+1)^{N}}\left(\delta_{0}^{-1}\right)$, or equivalently,
	\begin{align}\label{eq:delta_N_explicit}
		\delta_{N} = \sech\left((2L+1)^{N}\arcosh\Bigl(\delta_{0}^{-1}\Bigr)\right)\ .
	\end{align}
	For large $ N $, Eq.~\eqref{eq:delta_N_explicit} asymptotically behaves as $\delta_{N} = \mO\left(e^{-\arcosh\bigl(\delta_{0}^{-1}\bigr) (2L+1)^{N}}\right)\ $.
	A typical setting of a Grover search starts from the initial fidelity $ F_{0} = 1-\delta_{0}^{2}\ll1 $, e.g. the initial state is a superposition of all basis states.
	Then the decay rate can be rewritten as
	\begin{align}
		\arcosh\left(\delta_{0}^{-1}\right) = \arcosh((1-F_{0})^{-1/2}) = \sqrt{F_{0}} + \mO(F_{0}^{3/2})\ ,
	\end{align}
	by Taylor expansion.
	If only the leading order is taken, $\delta_{N} \sim e^{-\sqrt{F_{0}}(2L+1)^{N}}$, and thus 
	\begin{align}\label{eq:nested_length_scaling}
		(2L+1)^{N}  = \mO\left(\log\left(\delta_{N}^{-1}\right)F_{0}^{-1/2}\right)
	\end{align}
	is needed for the final distance $\delta_{N}$.
	The circuit depth for the entire algorithm is proportional to Eq.~\eqref{eq:nested_length_scaling} due to the unfolding implementation, and it contains the characteristic Grover scaling factor $ F_{0}^{-1/2} $. 
	If a non-recursive algorithm is applied, Eq.~\eqref{eq:delta_N_explicit} can still be used with $ N = 1 $. 
	Denote the final distance as 
	\begin{align}
		\delta' = \sech\left((2L'+1)\arcosh\left(\delta_{0}^{-1}\right)\right),
	\end{align}
	when the alternation length of this algorithm is $ L' $. For $ \delta' = \delta_{N} $, the alternation length of the non-nested algorithm should satisfy $ 2L'+1 = (2L+1)^{N} $, resulting in exponentially many partial reflections. 
	Note that Eq.~\eqref{eq:nested_length_scaling} then implies $ L' \sim F_{0}^{-1/2} $, which displays the quantum advantage compared to the classical search requiring $ L' \sim F_{0}^{-1} $.

	\subsection{Robustness of dynamic programming for the Grover search}\label{app:QDP_NFGS_convergence}
	In this section, we study the nested fixed-point Grover search algorithm implemented with bounded noise arising from using a dynamic programming approach.
	We establish the robustness of the algorithm against small implementation errors.

	\subsubsection{Distance bounds after a DME implementation}\label{subsec:mixed_state_propagation}
	
	In this subsection, we introduce the DME implementation of an iteration of the dynamic Grover search. 
	We first establish (through Proposition~\ref{pro:DME_stays_qubit} and Corollary~\ref{cor:NFGS_stays_qubit}) that the DME implementation still has the essential feature of the Grover search, namely all relevant states are effectively two-dimensional. 
	Then the main results Lemma~\ref{lem:mixed_step_1} and Corollary~\ref{lem:mixed_step_2} provide the upper bound of a distance after one iteration of DME implemented Grover search, starting from a mixed state. 
	In short, we observe that whenever the initially prepared state is close to being pure, then the distance to our Grover target state still follows a similar recursive relation as before, namely Eq.~\eqref{eq:delta_n_prop}.
	
	The `dynamic' parts of the recursion unitary $\Gammaop{L}{\rho}$ are the memory-calls $\Eop{\alpha}{\rho}$, i.e. partial reflectors around the state $\rho$.
	Hence, DME (Def.~\ref{def:DME}) is the right choice of memory-usage query to use for implementing the memory-call $\Eop{\alpha}{\rho}$.
	To recap, the DME memory-usage query can be written as
	\begin{align}\label{eq:DME_explicit_def_restatement}
		\hE{s}{\rho}\left(\sigma\right) = \cos^2(s)\sigma - i\sin(s)\cos(s)[\rho,\sigma] + \sin^2(s)\rho\ ,
	\end{align} 
	with the duration $s$ of our choice.
	If $M$ queries are made with $s = \frac{\alpha}{M}$, the implemented channel $(\hE{s}{\rho}(\sigma) )^{M}$ approximates the memory-call $\Eop{\alpha}{\rho}$ with error $\mO(1/M)$. 
	
	There are two key considerations in the implementation of quantum instructions via DME:
	\begin{enumerate}
		\item The implementation of a recursion step $ \Gammaop{L}{\psi} $ becomes non-unitary even when the instruction state $ \psi = \dm{\psi} $ and the input state $ \phi = \dm{\phi} $ are pure. To see this, consider a special case of DME protocol with pure instruction state $ \psi $ and input state $ \phi $,
		\begin{align}
			\hE{s}{\psi}(\phi) = \cos^{2}(s)\phi + \sin^{2}(s)\psi - i\sin(s)\cos(s)\braket{\psi}{\phi}\ketbra{\psi}{\phi}+ \mathrm{h.c.}.
		\end{align}
		We may directly compute the purity as
		\begin{align}\label{key}\Tr\left[\lset\hE{s}{\psi}(\phi)\rset^2\right] &= \cos^4(s) + \sin^4(s) + 2\sin^2(s)\cos^2(s)\left(2\lvert\braket{\psi}{\phi}\rvert^2 - \lvert\braket{\psi}{\phi}\rvert^4\right)\\
			&= 1 - 2\sin^2(s)\cos^2(s)\left(1- \lvert\braket{\psi}{\phi}\rvert^2\right)^{2}\neq1, \no
		\end{align}
		hence $ \rho_{n} $, the state after $ n $ iterations of the Grover search recursion, is no longer pure in general.
		\item Another important property to check is whether DME maps a target state supported on a two-dimensional subspace to another (now potentially mixed) state supported on the same space. 
		In the exact algorithm, state vectors always stay in the space 
		\begin{align}\label{key}
			\mH_{\rm rel} \coloneq \mathrm{span}\lset \ket{\tau},  \ket{\tp} \coloneq \frac{\psivec{0} - \braket{\tau}{\psi_0}\ket{\tau}}{\sqrt{1-\vert\braket{\tau}{\psi_0}\vert^2}}\rset,
		\end{align} 
		defined by the initial state vector $ \psivec{0} $ and the target state vector $ \ket{\tau} $. 
		This restriction is the reason that the fixed-point Grover search of~\cite{Yoder2014Grover} could utilize a $ \mathrm{SU(2)} $ representation.
		To establish similar restriction for mixed states, we define 
		\begin{align}
			\mS_{\rm rel} \coloneq \lset \rho \,\vert\, \rho: \text{density matrix},\ \mathrm{supp}(\rho)\subset\mH_{\rm rel} \rset
		\end{align}
		the set of density matrices having support only in $ \mH_{\rm rel} $.
		Proposition \ref{pro:DME_stays_qubit} then rigorously shows that all relevant density operators resulting from the DME implementation stay in $\mS_{\rm rel}$.
	\end{enumerate}

	\begin{pro}\label{pro:DME_stays_qubit} 
		Suppose two states $\rho,\sigma \in \mS_{\rm rel}$.  
		The state after a DME channel $\sigma^\prime = \hE{s}{\rho} (\sigma) $ defined in Eq.~\eqref{eq:DME_def} is also contained in $\mS_{\rm rel}$.
		In other words, for all $s\in\mbR$, we have that
		$\hE{s}{\mS_\mathrm{rel}}(\mS_\mathrm{rel})\subset \mS_\mathrm{rel}\ $.
	\end{pro}
	
	\begin{proof} 
		From direct calculation,
		\begin{align}
			\hE{s}{\rho}(\sigma)\ket{\phi} = \cos^2(s)\sigma\ket{\phi} - i\sin(s)\cos(s)[\rho,\sigma]\ket{\phi} + \sin^2(s)\rho\ket{\phi} = 0\ ,
		\end{align}
		whenever $ \ket{\phi}\in \mH_{\rm rel}^{\mathsf{C}} $.
		Therefore, $ \mathrm{supp}(\hE{s}{\rho}(\sigma))\subset\mH_{\rm rel} $. 
	\end{proof}
	Noticing that DME for dynamic Grover iterations satisfies these conditions at every step of the recursion, we can be sure that the states generated by QDP in Grover search have full support in the relevant Hilbert space representation.
	\begin{cor}\label{cor:NFGS_stays_qubit}
		The QDP implementation of Grover search using DME generates a sequence of states $\rho_n\in\mS_{\rm rel} $ for all $ n\in\mbN $ whenever $\rho_{0} = \psi_{0} \in\mS_{\rm rel}$. 
	\end{cor}
	\begin{proof}
		From Proposition~\ref{pro:DME_stays_qubit}, DME protocols map states in $ \mS_{\rm rel} $ to $ \mS_{\rm rel} $.
		Furthermore, reflections around the target state $ \tau\in\mS_{\rm rel} $, which is implemented by an exact unitary, not DME, also ensures that
		$		\Eopch{s}{\tau}\left(\sigma\right) \in \mS_{\rm rel}$
		whenever $ \sigma\in\mS_{\rm rel} $.
		Since a recursion step of Grover search comprises the above two operations,
		the state after the iteration, $ \rho_{n} $, also stays in $ \mS_{\rm rel} $.
	\end{proof}
	
	Corollary~\ref{cor:NFGS_stays_qubit} enables a simple representation of intermediate states $ \rho_{n} $. 
	Defining the rank-$ 2 $ projector onto $ \mH_{\rm rel} $, $\1_{\rm rel} = \dm{\tau}+\dm{\tp}$, we note that this operator is invariant under an exact Grover recursion step with respect to any state $ \phi\in\mS_{\rm rel} $, i.e.
	\begin{align}\label{eq:idrel_invar}
		\Gammaopch{L}{\phi}\left(\1_{\rm rel}\right) = \1_{\rm rel},
	\end{align} 
	due to unitarity of $\Gammaopch{L}{\phi}$. 
	Since $\dim (\mathcal{H}_{\rm rel}) =2$, we can decompose any $ \rho \in\mS_{\rm rel} $ in its eigenbasis
	\begin{align}\label{key}
		\rho = (1-x)\ketbra{\psi}{\psi} + x \ketbra{\psi^\perp}{\psi^\perp} = (1-2x)\ketbra{\psi}{\psi} + x\1_{\rm rel}, 
	\end{align}
	for $x \in [0,\frac{1}{2}] $ by using the fact that $\1_{\rm rel} = \ketbra{\psi}{\psi}+\ketbra{\psi^\perp}{\psi^\perp}$.
	We will refer to $\psi $ as the pure state associated with $\rho$ and $x$ as the mixedness parameter.
	Interestingly, the partial reflection around $ \rho$ with angle $ s $ can be rewritten using the observation that $[\psi,\1_{\rm rel}] = 0$:
	
	\begin{align}\label{eq:rho_reflection_approx_psi}
		\Eop{s}{\rho} = e^{-is\left[(1-2x)\psi + x\1_{\rm rel}\right]} = e^{-isx\1_{\rm rel}}\Eop{s(1-2x)}{\psi}\ ,
	\end{align}
	which is identical to the partial reflection around the corresponding pure state $ \psi $ with the angle reduced by the factor $(1-2x)$, apart from the phase factor globally acting on $\mH_{\rm rel} $. 
	If the mixedness parameter $ x $ is known, one can effectively implement the reflection around the pure state $ \psi $ using mixed states $ \rho $ by adjusting the DME duration $s\mapsto s/(1-2x)$.
	
	Nevertheless, it is usually the case that only an upper bound, but not the exact amount of mixedness $x$ is known. 
	The Proposition and the Lemma below assess the QDP implementation of a Grover search recursion step when the mixedness parameter $ x $ is unknown. 
	Instead of the exact partial reflection unitary channels $ \Eopch{s}{\psi} = \Eopch{s/(1-2x)}{\rho} $ around its associated pure state $ \psi $, we employ DME queries $ \hE{s}{\rho} $ as defined in Eq.~\eqref{eq:DME_explicit_def_restatement}.

	\begin{pro}\label{pro:dist_unitaries_psi_rho}
		Let $ \rho = (1-2x)\psi + x\1_{\rm rel} $ be a state with an associated pure state $ \psi $ and a mixedness parameter $ x\in[0,\varepsilon] $ for some $ 0\leq\varepsilon< \frac{1}{2} $.
		A dynamic unitary channel 
		\begin{align}\label{eq:gammaopch_rho}
			\Gammaopch{L}{\rho} = \Eopch{\alpha_{L}}{\rho}\circ\Eopch{\beta_{L}}{\tau}\circ\cdots\circ  \Eopch{\alpha_{1}}{\rho}\circ\Eopch{\beta_{1}}{\tau}
		\end{align}
		that uses the same angles $\{\alpha_{l},\beta_{l}\}_{l}$ used in $\Gammaopch{L}{\psi}$ has the distance 
		\begin{align}\label{eq:Lpiep_2}
			\frac{1}{2}\left\lVert \Gammaopch{L}{\psi} - \Gammaopch{L}{\rho} \right\rVert_{\rm Tr} \leq L\pi\varepsilon\ ,
		\end{align}
		from the dynamic unitary channel constructed with the associated pure state $\psi$.
	\end{pro}
	\begin{proof}
		First, observe that 
		\begin{align}
			\left\lVert \Gammaopch{L}{\psi} - \Gammaopch{L}{\rho} \right\rVert_{\rm Tr} \leq \sum_{i=1}^{L}\left\lVert \Eopch{\alpha_{i}}{\psi} - \Eopch{\alpha_{i}}{\rho} \right\rVert_{\rm Tr}\ ,
		\end{align}
		from the triangle inequality and the unitary invariance of the trace norm. 
		The distance between $\Eopch{\alpha}{\rho}$ and $\Eopch{\alpha}{\psi}$ can be written as
		\begin{align}
			\frac{1}{2}\left\lVert  \Eopch{\alpha}{\psi} - \Eopch{\alpha}{\rho}\right\rVert_{\rm Tr} = \frac{1}{2}\max_{\sigma\in\mS(\mH_{\rm rel})}\left\lVert \Eopch{\alpha}{\psi}\left(\sigma\right) -\Eopch{\alpha}{\rho}(\sigma) \right\rVert_{1} = \frac{1}{2}\max_{\sigma\in\mS(\mH_{\rm rel})}\Tr\left\lbrace \sqrt{\left(\Eopch{\alpha}{\psi}\left(\sigma\right) -\Eopch{\alpha}{\rho}(\sigma)\right)^{2}} \right\rbrace\ .
		\end{align}
		Using the matrix notation 
		\begin{align}\label{eq:variables}
			\psi = \begin{bmatrix}
				1 & 0 \\ 0 & 0
			\end{bmatrix},\quad \rho = \begin{bmatrix}
				1-x & 0 \\ 0 & x
			\end{bmatrix},\quad \sigma = \begin{bmatrix}
				p & q \\ q^* & 1-p
			\end{bmatrix}\ ,
		\end{align}
		we arrive at the explicit expression
		\begin{align}
			\left(\Eopch{\alpha}{\psi}(\sigma) - \Eopch{\alpha}{\rho}(\sigma)\right)^{2} &= \begin{bmatrix}
				4\abs{q}^{2}\sin^{2}\left(\alpha x\right) & 0 \\ 0 & 4\abs{q}^{2}\sin^{2}\left(\alpha x\right) 
			\end{bmatrix}\ .
		\end{align}
		The state $\sigma$ giving the maximum distance is the one having the maximal $\abs{q}$.
		The condition that $\sigma\geq 0$ imposes $\abs{q}\leq \frac{1}{2} $, and $\frac{1}{2}\lVert  \Eopch{\alpha}{\psi} - \Eopch{\alpha}{\rho}\rVert_{\rm Tr} = \lvert \sin\left(\alpha x\right)\rvert$. From $\alpha \in [-\pi,\pi]$ and $\abs{\alpha \varepsilon}<\frac{\pi}{2}$, we bound $\abs{\sin(\alpha x)} < \pi \varepsilon$ and achieve Eq.~\eqref{eq:Lpiep_2}.
	\end{proof}

	Now we establish the main results of this subsection.
	We first define
	\begin{align}\label{eq:delta_n_prop}
		\delta_{n}\coloneq \left(T_{(2L+1)}\left(\delta_{n-1}^{-1}\right)\right)^{-1} +\epsilon\ ,
	\end{align}
	and use it as the parameter  $q = \delta_{n}$ when choosing angles Eq.~\eqref{eq:angles}.
	At the same time, we also demonstrate that $\delta_{n}$ bounds the distance to the target state, after an iteration of DME implementation with error $\epsilon$, i.e. it is the counterpart to Eq.~\eqref{eq:modified_step_error_app} used for the exact iteration.
	Precise statements of this result are delineated in Lemma~\ref{lem:mixed_step_1} and Corollary~\ref{lem:mixed_step_2}.
	Further analysis on Eq.~\eqref{eq:delta_n_prop} itself will be made in Sec.~\ref{subsec:delta_n_analysis}.

	\begin{lem}[DME implementation subroutine]\label{lem:mixed_step_1}
		We consider iteration $n$ of the nested Grover search algorithm. 
		Assume that the initial state for the iteration $ \rho_{n-1} = (1-2x_{n-1})\psi_{n-1} + x_{n-1}\1_{\rm rel} $ is unknown to us but satisfies the following:
		\begin{itemize}
			\item An associated pure state $\psi_{n-1}$ has the distance to the target state bounded by some number $\delta_{n-1}$, i.e. $\frac{1}{2}\lVert \tau - \psi_{n-1} \rVert_{1} \leq \delta_{n-1}$.
			\item The mixedness parameter is bounded by $ x_{n-1}  \leq \frac{\epsilon}{2L\pi}$ for some number $\epsilon$.
		\end{itemize}
		Then, it is possible to prepare a state $\rho'_{n} = (1-2x_{n}')\psi_{n} + x'\1_{\rm rel} $, such that
		\begin{align}\label{eq:delta_n_prop_2}
			\frac{1}{2}\lVert \tau - \rho'_{n} \rVert_{1} &\leq \delta_{n} -\frac{\epsilon}{4},
		\end{align}
		where $\delta_{n}$ is defined by the recurrence relation Eq.~\eqref{eq:delta_n_prop}.
		The preparation requires consuming $ \mO(L^2\epsilon^{-1}) $ copies of $ \rho_{n-1} $.
	\end{lem}
	\begin{proof}
		The preparation is performed by implementing nested fixed-point Grover search with DME memory-usage queries.
		Firstly, we choose the angles $\{\alpha_{l}(q),\beta_{l}(q)\}_{l}$ for the dynamic unitary $\Gammaopch{L}{\psi_{n-1}}$. 
		Following Lemma~\ref{lem:NFGS_good_parameter_q}, we set $q = \sech((2L+1)\arcsech(\delta_{n-1}))$, or equivalently, $q = \delta_{n}-\epsilon$.
		Secondly, we denote the total number of DME queries $\hE{s}{\rho_{n-1}}$ used for implementing $\Gammaopch{L}{\psi_{n-1}}$ as $M$.
		Since $L$ memory-calls (partial reflections around $\rho_{n-1}$) are contained in $\Gammaopch{L}{\psi_{n-1}}$, each memory-call employs $\frac{M}{L}$ DME memory-usage queries.
		The implemented channel reads
		\begin{align}\label{eq:Gammaerr_def}
			\Gammaerr{L}{\rho_{n-1}} \coloneq  \left(\hE{s_{L}}{\rho_{n-1}}\right)^{\frac{M}{L}}\circ\Eopch{\beta_{L}}{\tau}\circ\cdots\circ  \left(\hE{s_{1}}{\rho_{n-1}}\right)^{\frac{M}{L}}\circ\Eopch{\beta_{1}}{\tau}\ ,
		\end{align}
		with $s_{l} = \frac{\alpha_{l}L}{M}$, the angle $\alpha_{l}$ distributed equally to $\frac{M}{L}$ DME queries.
		The output state of the channel is
		\begin{align}
			\rho'_{n} := \Gammaerr{L}{\rho_{n-1}}\left(\rho_{n-1}\right) = (1-2x_{n}')\psi_{n} + x_{n}'\1_{\rm rel}\  .
		\end{align}
		We additionally define the state after an exact unitary channel $\Gammaopch{L}{\rho_{n-1}}$ as
		\begin{align}\label{eq:tilde_rho_def}
			\tilde{\rho}_{n} := \Gammaopch{L}{\rho_{n-1}}\left(\rho_{n-1}\right) = (1-2x_{n-1})\phi_{n} + x_{n-1}\1_{\rm rel}\ .
		\end{align}
		Since $\Gammaopch{L}{\rho_{n-1}}$ is a unitary channel, it leaves the projector $\1_{\rm rel}$ intact in Eq.~\eqref{eq:tilde_rho_def}, and the associated pure state
		\begin{align}
			\phi_{n} = \Gammaopch{L}{\rho_{n-1}}(\psi_{n-1})\ ,
		\end{align}
		while the mixedness parameter for $\tilde{\rho}_{n}$ is identical to that of $\rho_{n-1}$.
		We prove the Lemma starting from the triangle inequality $\lVert \tau - \rho'_{n}\rVert_{1} \leq \lVert \tau - \phi_{n}\rVert_{1} +\lVert \phi_{n} - \rho'_{n}\rVert_{1}$ and upper bounding each term in the RHS.
		
		\begin{itemize}
			\item \emph{Upper bound for $\lVert \tau - \phi_{n} \rVert_{1}$}:
			Another triangle inequality gives 
			\begin{align}\label{eq:err1}
				\lVert \tau - \phi_{n} \rVert_{1} = \left\lVert \tau - \Gammaopch{L}{\rho_{n-1}}(\psi_{n-1})\right\rVert_{1} \leq \left\lVert \tau - \Gammaopch{L}{\psi_{n-1}}(\psi_{n-1})\right\rVert_{1} + \left\lVert \Gammaopch{L}{\psi_{n-1}}(\psi_{n-1}) - \Gammaopch{L}{\rho_{n-1}}(\psi_{n-1})\right\rVert_{1}\ .
			\end{align}
			The second term in the RHS is bounded by the trace norm distance of channels $\Gammaopch{L}{\psi_{n-1}}$ and $\Gammaopch{L}{\rho_{n-1}}$.
			Applying Proposition~\ref{pro:dist_unitaries_psi_rho} and using the assumption $x_{n-1}\leq \frac{\epsilon}{2L\pi}$,
			\begin{align}\label{eq:tracenorm_pureins_vs_mixed}
				\left\lVert \Gammaopch{L}{\psi_{n-1}}(\psi_{n-1}) - \Gammaopch{L}{\rho_{n-1}}(\psi_{n-1})\right\rVert_{1}\leq \left\lVert \Gammaopch{L}{\psi_{n-1}} - \Gammaopch{L}{\rho_{n-1}} \right\rVert_{\rm Tr} \leq \epsilon\ .
			\end{align}
			The other term $\lVert \tau - \Gammaopch{L}{\psi_{n-1}}(\psi_{n-1})\rVert_{1}$ resembles the final distance after an exact nested fixed-point Grover search in Lemma~\ref{lem:NFGS_good_parameter_q}.
			However, we cannot directly apply Lemma~\ref{lem:NFGS_good_parameter_q}, since we do not know the initial distance $\lVert \tau - \psi_{n-1}\rVert_{1}$ exactly; we only assume that it is bounded by $\delta_{n-1}$.
			Hence, we must start from the more general relation Eq.~\eqref{eq:distance_recursion_1}, which gives the equality
			\begin{align}\label{eq:distance_exact_with_unknowndistance}
				\frac{1}{2}\left\lVert \tau - \Gammaopch{L}{\psi_{n-1}}(\psi_{n-1}) \right\rVert_{1} = (\delta_{n}-\epsilon)T_{2L+1}\left(T_{1/(2L+1)}\left((\delta_{n}-\epsilon)^{-1}\right)\frac{1}{2}\left\lVert \tau - \psi_{n-1} \right\rVert_{1}\right)\ ,
			\end{align}
			substituting $q = \delta_{n}-\epsilon$.
			Recalling that $\delta_{n} - \epsilon = (T_{(2L+1)}(\delta_{n-1}^{-1}))^{-1}$ and Eq.~\eqref{eq:Chebyshev_multiplication}, we get $T_{1/(2L+1)}((\delta_{n}-\epsilon)^{-1}) = \delta_{n-1}^{-1}$.
			Eq.~\eqref{eq:distance_exact_with_unknowndistance} then becomes
			\begin{align}
				\frac{1}{2}\left\lVert \tau - \Gammaopch{L}{\psi_{n-1}}(\psi_{n-1}) \right\rVert_{1} =(\delta_{n}-\epsilon)T_{2L+1}\left(\delta_{n-1}^{-1}\frac{1}{2}\left\lVert \tau - \psi_{n-1} \right\rVert_{1}\right)\ .
			\end{align}
			Chebyshev polynomials follow $T_{2L+1}(x)\leq 1$ whenever $x\in[0,1]$, and the argument of $T_{2L+1}$ in the RHS, $\delta_{n-1}^{-1}\frac{1}{2}\lVert \tau - \psi_{n-1} \rVert_{1} $ is indeed smaller than or equal to $1$ from the assumption of the Lemma.
			Then $\frac{1}{2}\lVert \tau - \Gammaopch{L}{\psi_{n-1}}(\psi_{n-1}) \rVert_{1}\leq \delta_{n}-\epsilon$ follows, and we finally obtain the bound
			\begin{align}\label{eq:dist_tau_phi}
				\frac{1}{2}\lVert \tau - \phi_{n} \rVert_{1} \leq  \delta_{n}-\frac{\epsilon}{2}\ .
			\end{align}
			
			\item \emph{Upper bound for $\lVert \phi_{n} - \rho'_{n} \rVert_{1}$}:
			Use the triangle inequality to obtain
			\begin{align}\label{eq:phi_psi_distance}
				\lVert \phi_{n} - \rho'_{n} \rVert_{1} &\leq \left\lVert \phi_{n} - \Gammaerr{L}{\rho_{n-1}}(\psi_{n-1}) \right\rVert_{1} + \left\lVert \Gammaerr{L}{\rho_{n-1}}(\psi_{n-1}) - \rho'_{n}\right\rVert_{1}\no \\
				&= \left\lVert \Gammaopch{L}{\rho_{n-1}}(\psi_{n-1}) - \Gammaerr{L}{\rho_{n-1}}(\psi_{n-1}) \right\rVert_{1} + \left\lVert \Gammaerr{L}{\rho_{n-1}}(\psi_{n-1}) - \Gammaerr{L}{\rho_{n-1}}(\rho_{n-1})\right\rVert_{1}\ .
			\end{align}
			The first term is bounded by 
			\begin{align}
				\left\lVert \Gammaopch{L}{\rho_{n-1}}(\psi_{n-1}) - \Gammaerr{L}{\rho_{n-1}}(\psi_{n-1}) \right\rVert_{1} \leq \left\lVert \Gammaopch{L}{\rho_{n-1}} - \Gammaerr{L}{\rho_{n-1}}\right\rVert_{\rm Tr}\leq \sum_{l=1}^{L}\frac{M}{L} \left\lVert \Eopch{s_{l}}{\rho_{n-1}} - \hE{s_{l}}{\rho_{n-1}} \right\rVert_{\rm Tr}\ ,
			\end{align}
			where the second inequality follows from comparing Eqs.~\eqref{eq:gammaopch_rho} and~\eqref{eq:Gammaerr_def}.
			Making use of the explicit expressions in Eq.~\eqref{eq:variables} again, the eigenvalues of $(\Eopch{s}{\rho} - \hE{s}{\rho})^{2}$ are evaluated to be degenerate and of order $\lambda = \mO(s^{4})$. 
			Numerical maximization over all possible values of $s_{l}$ and $\rho_{n-1}$ gives
			$\frac{1}{2}\lVert \Eopch{s_{l}}{\rho_{n-1}} - \hE{s_{l}}{\rho_{n-1}} \rVert_{\rm Tr} < 0.71 s_{l}^{2}$.
			Each $s_{l}$ can be bounded as $\abs{s_{l}} = \abs{\frac{\alpha_{l}L}{M}}\leq \frac{\pi L}{M}$, which leads to a bound
			\begin{align}\label{eq:tr_diff_uni_nonuni_rho}
				\frac{1}{2}\left\lVert \Gammaopch{L}{\rho_{n-1}}(\psi_{n-1}) - \Gammaerr{L}{\rho_{n-1}}(\psi_{n-1}) \right\rVert_{1}  \leq \frac{1}{2}\left\lVert \Gammaopch{L}{\rho_{n-1}} - \Gammaerr{L}{\rho_{n-1}}\right\rVert_{\rm Tr} < \frac{7L^{2}}{M}\ .
			\end{align}
			
			The second term of Eq.~\eqref{eq:phi_psi_distance} is easier to bound; using data processing inequality, 
			\begin{align}\label{eq:tr_diff_uni_nonuni_rho_2}
				\frac{1}{2}\left\lVert \Gammaerr{L}{\rho_{n-1}}(\psi_{n-1}) - \Gammaerr{L}{\rho_{n-1}}(\rho_{n-1})\right\rVert_{1} \leq \frac{1}{2}\lVert \psi_{n-1} - \rho_{n-1}\rVert_{1} = x_{n-1} \leq \frac{\epsilon}{2L\pi}\ .
			\end{align}
			Therefore, 
			\begin{align}
				\frac{1}{2}\lVert \phi_{n} - \rho'_{n}\rVert_{1} <\frac{7L^{2}}{M} + \frac{\epsilon}{2L\pi}
			\end{align}
		\end{itemize}
		
		Combining Eqs.~\eqref{eq:dist_tau_phi} and~\eqref{eq:tr_diff_uni_nonuni_rho_2}, we achieve
		\begin{align}
			\frac{1}{2}\lVert \tau - \rho'_{n}\rVert_{1} < \delta_{n} - \frac{\epsilon}{2} + \frac{7L^{2}}{M} + \frac{\epsilon}{2L\pi}.
		\end{align}
		With $M = \mO(L^{2}\epsilon^{-1})$, we can make 
		\begin{align}\label{eq:M_bound}
			\frac{7L^{2}}{M} + \frac{\epsilon}{2L\pi} < \frac{\epsilon}{4}\ ,
		\end{align}
		which proves Eq.~\eqref{eq:delta_n_prop_2}.
		The number of $\rho_{n-1}$ needed for the implementation is $M = \mO(L^{2}\epsilon^{-1})$, which concludes the proof.
	\end{proof}
	Lemma~\ref{lem:mixed_step_1} indicates that as long as the mixedness parameter stays low throughout the iterations, the total implementation error of the QDP Grover search can be kept bounded up to a desired accuracy.
	Although the mixedness parameter generally increases, we can bound the increment. 
	Consider the distance $\frac{1}{2}\lVert \rho'_{n} - \tilde{\rho}_{n} \rVert_{1}$: it is minimized when their associated pure states are identical, i.e. 
	\begin{align}
		\frac{1}{2}\lVert \rho'_{n} - \tilde{\rho}_{n} \rVert_{1} \geq \left\lVert (x'_{n}-x_{n-1})\psi_{n} - (x'_{n}-x_{n-1})\frac{\1_{\rm rel}}{2}\right\rVert_{1} = \abs{x'_{n} - x_{n-1}}.
	\end{align}
	On the other hand, 	
	\begin{align}
		\frac{1}{2}\lVert \rho'_{n} - \tilde{\rho}_{n} \rVert_{1} \leq \frac{1}{2}\left\lVert \Gammaerr{L}{\rho_{n-1}}(\rho_{n-1}) - \Gammaopch{L}{\rho_{n-1}}(\rho_{n-1})  \right\rVert_{\rm 1} \leq \frac{1}{2}\left\lVert \Gammaopch{L}{\rho_{n-1}} - \Gammaerr{L}{\rho_{n-1}}\right\rVert_{\rm Tr} < \frac{7L^{2}}{M}\ ,
	\end{align}
	from Eq.~\eqref{eq:tr_diff_uni_nonuni_rho}.
	Therefore, 
	\begin{align}\label{eq:mixedness_para_bound}
		0<x'_{n} < x_{n-1} + \frac{7L^{2}}{M} < \frac{\epsilon}{4}\ ,
	\end{align}
	since we set $ \frac{7L^{2}}{M} < \frac{L\pi - 2}{4L\pi}\epsilon$ in Eq.~\eqref{eq:M_bound}.
	
	The mixedness factor $x'_{n}$ can be reduced into $x_{n}\leq \frac{\epsilon}{2L\pi}$ using the IMR protocol Lemma~\ref{lemma:mixedness_reduction_subroutine}, and we can continue the DME implementation without an accumulating error.
	Since we eventually change the factor $x'_{n}$, it is more important to get a bound for $\frac{1}{2}\lVert \tau - \psi_{n}\rVert_{1}$ corresponding to the associated pure state, rather than $\frac{1}{2}\lVert \tau - \rho'_{n} \rVert_{1}$ as in Lemma~\ref{lem:mixed_step_1}.

	\begin{cor}\label{lem:mixed_step_2}
		The resulting state $\rho'_{n} = (1-2x_{n}')\psi_{n} + x'\1_{\rm rel} $ of Lemma~\ref{lem:mixed_step_1} has the associated pure state, such that
		\begin{align}\label{eq:delta_n_prop_3}
			\frac{1}{2}\lVert \tau - \psi_{n} \rVert_{1} &\leq \delta_{n}\ ,
		\end{align}
		where $\delta_{n}$ is defined in Eq.~\eqref{eq:delta_n_prop}.
	\end{cor}
	\begin{proof}
		The distance in Eq.~\eqref{eq:delta_n_prop_3} follows the triangle inequality $\frac{1}{2}\lVert \tau - \psi_{n} \rVert_{1} \leq \frac{1}{2}\lVert \tau - \rho'_{n} \rVert_{1} +\frac{1}{2}\lVert \rho'_{n} - \psi_{n} \rVert_{1}$.
		The first term in the RHS is already bounded by $\delta_{n}-\frac{\epsilon}{4}$ in Lemma~\ref{lem:mixed_step_1}.
		The second term $\frac{1}{2}\lVert\rho'_{n}- \psi_{n} \rVert_{1} = x'_{n} <\frac{\epsilon}{4} $ using Eq.~\eqref{eq:mixedness_para_bound}. 
	\end{proof}

		\subsection{QDP Grover search with linear circuit depth and exponential circuit width}\label{subsec:delta_n_analysis}
		
		In this subsection, we repeat the analysis in Theorem~\ref{thm:qdp_pure} for QDP nested fixed-point Grover search algorithm. 
		This time, we can obtain scaling factors more precisely by combining the results from Section~\ref{app:QDP_NFGS_convergence}.
		As a result, we obtain the required number of initial state as a function of $L$, $\epsilon$, and $p_{\rm th}$ denoting the number of partial reflections of each recursion step, the threshold implementation error parameter, and the threshold failure probability.
		
		\begin{lem}[QDP Grover search]\label{lem:dynamic_Grover_technical}
			Consider the nested fixed-point Grover search with $ N $ recursions and $2L$ partial reflections at each recursion step, and let
			\begin{itemize}
				\item $\tau = \dm{\tau}$ be the target state of the search,
				\item $ \psi_{0} = \dm{\psi_{0}} $ be the initial state whose distance to the target $ \frac{1}{2}\lVert \tau - \psi_{0}\rVert_{1} \eqcolon \delta_{0} $, 
				\item $\epsilon\in (0,\frac{2}{3})$ be the threshold implementation error parameter,
				\item $p_{\mathrm{th}} \in (0,1)$ be the threshold failure probability of the algorithm.
			\end{itemize}
			Then, the QDP implementation of the search prepares the final state $\rho_{N}$ with
			\begin{itemize}
				\item the distance to the target state $\frac{1}{2}\lVert \tau - \rho_{N} \rVert_{1} \leq  \delta_{N} + \frac{\epsilon}{2L\pi}$, where $\delta_{N}$ is defined recursively for any $n>0$ as
				\begin{align}\label{eq:modified_step_error}
					\delta_{n} = \sech\left((2L+1)\arcsech\left(\delta_{n-1}\right)\right) + \epsilon,
				\end{align}
				\item the success probability greater than $1-p_{\mathrm{th}}$,
				\item a circuit of depth $ \mO(L^{2}N\epsilon^{-1}) $, and 
				\item a number of initial state $\psi_{0}$ scaling as $ L^{(2+\Delta)N}\epsilon^{-N}\log(p_{\rm th}^{-1}) $ with some positive constant $\Delta$.
			\end{itemize}
			In contrast, the unfolding implementation of the nested fixed-point Grover search has a circuit depth $\mO((2L)^{N})$ and a constant width.
		\end{lem}
		\begin{proof}
			
			The algorithm consists of $ N $ iterations of the recursion step, where each iteration involves two subroutines: a DME implementation in Lemma~\ref{lem:mixed_step_1} followed by a mixedness reduction subroutine in Lemma~\ref{lemma:mixedness_reduction_subroutine}.
			The proof is very similar to that of Theorem~\ref{thm:qdp_pure}, except the parts where we explicitly consider the dependence on $L$.
			\begin{enumerate}
				\item At a DME implementation subroutine, $ \mathtt{I}_{n} $ copies of $ \rho_{n-1} $ are transformed to $ \mathtt{O}_{n} $ copies of $ \rho_{n}' $ deterministically.
				Input state 
				\begin{align}
					\rho_{n-1} = (1-2x_{n-1})\psi_{n-1} + x_{n-1}\1_{\rm rel} 
				\end{align}
				is assumed to satisfy $ x_{n-1}\leq\frac{\epsilon}{2L\pi} $ and $ \frac{1}{2}\lVert \tau - \psi_{n-1}\rVert_{1} \leq \delta_{n-1} $, while the output state
				\begin{align}\label{eq:rho_n_prime}
					\rho_{n}' = \Gammaerr{L}{\rho_{n-1}}(\rho_{n-1}) = (1-x_{n}')\psi_{n} + x_{n}'\frac{\1_{\rm rel}}{2}\ ,
				\end{align}
				satisfies $ x'_{n}<\frac{\epsilon}{4} $ from Eq.~\eqref{eq:mixedness_para_bound}, and $ \frac{1}{2}\lVert \tau - \psi_{n} \rVert_{1} \leq \delta_{n} $ by Corollary~\ref{lem:mixed_step_2}. 
				For each channel $ \Gammaerr{L}{\rho_{n-1}} $ implementation, $ \mO(L^{2}\epsilon^{-1}) $ DME queries are made and consequently $ \mO(L^{2}\epsilon^{-1}) $ copies of $ \rho_{n-1} $ are consumed. 
				Plus, $ L $ partial reflections $ \Eopch{s}{\tau} $ around the target state are also applied.
				In sum, we need 
				\begin{align}
					\mathtt{I}_{n} = \mO\left(\frac{L^{2}}{\epsilon}\right)\mathtt{O}_{n}
				\end{align} 
				copies of $ \rho_{n-1} $ with depth $ \mO(\frac{L^{2}}{\epsilon}) $ elementary gates for the subroutine. 
				
				\item At a mixedness reduction subroutine, described in Lemma~\ref{lemma:mixedness_reduction_subroutine}, $ \mathtt{O}_{n} $ copies of $ \rho_{n}' $, output from the DME implementation subroutine (Eq.~\eqref{eq:rho_n_prime}), are transformed into $ \mathtt{I}_{n+1} $ copies of $ \rho_{n} = (1-2x_{n})\psi_{n} + x_{n}\1_{\rm rel} $, such that $ x_{n}\leq\frac{\epsilon}{2L\pi} $, with a guarantee that the success probability $ q_{\rm succ} $ is higher than $ 1 - q_{\rm th} $ for some threshold value $ q_{\rm th} $. 
				Following Eqs.~\eqref{eq:On_in_QDPpureThm} and~\eqref{eq:pth_in_QDPpureThm}, we can set $ q_{\rm th}= \frac{p_{\rm th}}{N} $ and $\mathtt{O}_{n} = (\frac{2}{c})^{R}\mathtt{I}_{n+1}$, where $ R $ is the number of IMR protocol rounds proportional to the $ \log $ of rate of reduction $ \frac{x_{n}'}{x_{n}} $, while $ c $ is the factor that decides how many transformed states survive after a round of the IMR protocol.
				The rate of reduction $ \frac{x_{n}'}{x_{n}} = \frac{L\pi}{2}$, giving $R = \mO(\log(\mathsf{g})) = \mO(\log(L)) $, while we can set $c = \frac{1}{2}$ from Eq.~\eqref{eq:fix_c}.
				Combining all,
				\begin{align}
					\mathtt{O}_{n} = 4^{\mO\left(\log(L)\right)}\mathtt{I}_{n+1} = \mathrm{poly}(L)\mathtt{I}_{n+1}\ ,
				\end{align}
				and circuit depth $ \mO(R) = \mO(\log(L)) $ from $ R $ repetitions of interferential mixedness reduction protocols (Lemma~\ref{lemma:purification_genearlised}) is added. 
			\end{enumerate}
			Finally, we account for the redundancy of the final number of copies $I_{N+1}= \mO(\log(Np_{\rm th}^{-1}))$ in Eq.~\eqref{eq:fix_c}, introduced to guarantee high success probability and fix the survival rate $c$.
			Then the initial number of copies needed 
			\begin{align}
				I_{1} = L^{(2+\Delta)N}\epsilon^{-N}\log(p_{\rm th}^{-1}),
			\end{align}
			where $\Delta$ is the additional exponent coming from the mixedness reduction rounds. 
			The number of initial copies is exponential in $N$ and grows faster than $L^{2N}$ as a function of $L$.
			Furthermore, the circuit depth scales as 
			\begin{align}
				\mO\left(\log(L)+\frac{L^{2}}{\epsilon}\right)N = \mO\left(\frac{NL^{2}}{\epsilon}\right)\ ,
			\end{align}
			which is linear in $ N $.
		\end{proof}

		\newcommand{\tdel}{\tilde{\delta}}
		\newcommand{\tL}{\tilde{L}}
		\newcommand{\tilh}{\tilde{h}}
		\subsection{From Lemma~\ref{lem:dynamic_Grover_technical} to Theorem~\ref{thm:dynamic_Grover}}\label{app:delta_convergence}
		
		In Lemma~\ref{lem:dynamic_Grover_technical}, we have shown that the distance $ \delta_{n} $ between the final state and the target state $\tau$ after $ n $ QDP Grover search iterations follows the recursion relation Eq.~\eqref{eq:modified_step_error}. 
		In this subsection, we demonstrate that for a reasonably small implementation error $ \epsilon $, the sequence $ \delta_{n} $ does not deviate from the sequence of distance obtained from ideal, errorless iterations.
		To do this, we set $ \delta_{n+1} = h(\delta_{n}) + \epsilon $, with the function 
		\begin{align}\label{eq:function_f}
			h(x) = \sech\left((2L+1)\arcsech(x)\right) \ .
		\end{align}
		Note that Eq.~\eqref{eq:function_f} gives distance recurrence relations for the exact nested fixed-point Grover search (Lemma~\ref{lem:NFGS_good_parameter_q}), i.e. $\frac{1}{2}\lVert \tau - \psi_{n}\rVert_{1} = h(\frac{1}{2}\lVert \tau - \psi_{n-1}\rVert_{1})$ for $\psivec{n} = \Gammaopch{L}{\psi_{n-1}}\psivec{n-1}$.
		By finding the range for fast spectral convergence (Def.~\ref{def:fast_spectral_convergence}), i.e. the range where its derivative $h'(x) <r$ for some $r<1$ and $h(x)+\varepsilon < x$, we demonstrate that the resulting distance from $N$ iterations of the QDP implementation $\delta_{N}$ can still be arbitrarily close to that of the exact implementation.
		
		As a first step, we establish the convexity of $ h(x) $ in the range $ x\in(0,1) $.
		\begin{pro}\label{pro:convexity}
			The function $ h(x) $ defined as in Eq.~\eqref{eq:function_f} is monotonically increasing and convex in the range $ x\in(0,1) $.
		\end{pro}
		\begin{proof}
			Let us define $\tL \coloneq 2L+1$ for simplicity. 
			The derivatives of $h(x)$ become
			\begin{align}\label{eq:hprime_def}
				h'(x) &\coloneq \frac{dh(x)}{dx} = \frac{\tL\sech\left(\tL\arcsech(x)\right)\tanh\left(\tL\arcsech(x)\right)}{x\sqrt{1-x^{2}}} >0\ ,
			\end{align}
			from which the monotonicity of $h(x)$ follows.
			We define another variable $ y\coloneq\arcsech\left(x\right)\in(0,\infty) $ and 
			\begin{align}
				g(y) \coloneq \frac{\tL\sech\left(\tL y\right)\tanh\left(\tL y\right)}{\sech(y)\tanh(y)}\ ,
			\end{align}
			satisfying $g(\arcsech(x)) = h'(x)$. 
			Thus,
			\begin{align}
				h''(x) \coloneq \frac{d^{2}h(x)}{dx^{2}} = \frac{dg(y)}{dy}\frac{dy}{dx} = -\frac{g'(y)}{x\sqrt{1-x^{2}}}\ ,
			\end{align}
			and $h(x)$ is convex when $ g'(y)\coloneq \frac{dg(y)}{dy} $ is negative for $y>0$.
			From direct calculation
			\begin{align}
				g'(y) = \frac{\tL}{y}\frac{\cosh^{2}(y)\sinh(\tL y)}{\sinh(y)\cosh^{2}(\tL y)}\left[\tL y\frac{2\sech^{2}(\tL y) - 1}{\tanh(\tL y)} - y\frac{2\sech^{2}(y)-1}{\tanh(y)}\right]\ .
			\end{align}
			Note that the sign of $g'(y)$ is determined by terms inside the square bracket, which can be written as $f(\tL y) - f(y)$ by defining 
			\begin{align}
				f(y) \coloneq y\frac{2\sech^{2}(y)-1}{\tanh(y)}\ .
			\end{align}
			$f(y)$ is a monotonically decreasing function, hence $f(\tL y) - f(y)<0$ and $h(x)$ is convex.
		\end{proof}

		Due to convexity, $h'(x)<r$ can be guaranteed whenever $x<\delta^{*}$ such that $h'(\delta^{*}) = r$. 
		The threshold $\delta^{*}$ is the smallest when $L =1$ and grows as $L$ increases. 
		Setting $L = 1$, we numerically obtain that $\delta^{*} = 0.73$ is sufficient for $h'(x)<1$ for any $L$.
		With a more reasonable choice $L = 5$, the threshold becomes $\delta^{*}\simeq 0.93$. 
		Furthermore, Proposition~\ref{pro:convexity} also implies that when $\epsilon$ is small, $h(x)+\epsilon = x$ has two solutions $x_{1}$ and $x_{2}$ with $x_{1}<x_{2}$ in the range $x\in(0,1)$. 
		
		Assuming $\epsilon\ll 1$, we can evaluate the approximate value of $x_{1}$ and $x_{2}$.
		\begin{pro}\label{pro:x2}
			When $\epsilon \ll 1$, a larger solution $x_{2}$ satisfying $h(x)+\epsilon = x$ is
			\begin{align}
				x_{2} = 1 - \frac{1}{4(L^{2}+L)}\epsilon + \mO(\epsilon^{2})\ .
			\end{align}
		\end{pro} 
		\begin{proof}
			The equation $h(x)+\epsilon = x$ can be rewritten as $\tL \arcsech(x) = \arcsech(x-\epsilon)$.
			Let us define $y\coloneq1-x$ and assume that $y$ and $\epsilon$ are both small. 
			The equation becomes
			\begin{align}
				\tL\sqrt{2y}  = \sqrt{2(y+\epsilon)} + \mO(y^{3/2})\  
			\end{align}
			or $y = \frac{1}{\tL^{2}-1}\epsilon + \mO(\epsilon^{2})$, which proves the proposition. 
		\end{proof}
		Proposition~\ref{pro:x2} indicates that for any reasonably small $\epsilon$, we have $\delta^{*}<x_{2}$.
		Hence, for the fast spectral convergence of the recursion, we only need to ensure that the initial distance $\delta_{0}<\delta^{*}$.
		Now we demonstrate that the QDP implementation of the Grover search can converge to the exact algorithm with an arbitrary error $\eta>0$.
		\begin{thm}[Efficient QDP Grover search algorithm]\label{thm:dynamic_Grover}
			Consider the nested fixed-point Grover search with $ N $ recursions and $2L$ partial reflections at each recursion step.
			Furthermore, assume that the distance between the initial state $\psi_{0}$ and the target state $\tau$ satisfies $ \frac{1}{2}\lVert \tau - \psi_{0}\rVert_{1} \eqcolon \delta_{0} <\delta^{*} $, where $\delta^{*}$ is determined by $h'(\delta^{*}) = 1$ with $h'$ defined in Eq.~\eqref{eq:hprime_def}. 
			Then, QDP implementation of the algorithm can output a state $\rho_{N}$, whose distance to the target state $\tau$ is arbitrarily close to that of the $\psivec{N}$, an output from the exact algorithm, i.e.
			\begin{align}
				\frac{1}{2}\lVert \tau - \rho_{N} \rVert_{1} \leq \frac{1}{2}\lVert \tau - \psi_{N} \rVert_{1} + \eta\ ,
			\end{align}
			for any $\eta>0$.
			In addition, the success probability of the QDP implementation exceeds $1-p_{\mathrm{th}}$ for any $p_{\mathrm{th}} \in(0,1)$.
			This implementation requires $ L^{(2+\Delta)N}\eta^{-N}\log(p_{\rm th}^{-1}) $ copies of $\psivec{0}$ and a circuit of depth $ \mO(L^{2}N\eta^{-1}) $ with some constant $\Delta>0$.
		\end{thm}
		\begin{proof}
			After $N$ iterations of the exact algorithm, output state $\psi_{N}$ satisfies 
			\begin{align}
				\frac{1}{2}\lVert \tau - \psi_{N} \rVert_{1} = \tilde{\delta}_{N},\quad \frac{1}{2}\lVert \tau-\psi_{0}\rVert_{1} = \delta_{0}\ ,
			\end{align}
			where $\tilde{\delta}_{N}$ is recursively defined by $\tilde{\delta}_{n+1} = h(\tilde{\delta}_{n})$ and $\tilde{\delta}_{0} = \delta_{0}$. 
			On the other hand, the QDP implementation (Lemma~\ref{lem:dynamic_Grover_technical}) gives
			\begin{align}
				\frac{1}{2}\lVert \tau - \rho_{N}\rVert_{1} \leq \delta_{N} + \frac{\epsilon}{2L\pi}\ ,
			\end{align}
			by using $ L^{\mO(N)}\epsilon^{-N}\log(p_{\rm th}^{-1}) $ copies of the initial state $\psi_0$,
			where $\delta_{n+1} = h(\delta_{n}) + \epsilon$ and $\epsilon$ is the parameter that decides circuit depth and width of the QDP algorithm. 
			
			Since $\delta_{0}<\delta^{*}$, we have $h'(\delta_{0})<r$ for some $r<1$ and $h(\delta_{0})+\epsilon<\delta_{0}$. 
			Then $\delta_{1} = h(\delta_{0})+\epsilon$ and $\delta_{2} = h(h(\delta_{0})+\epsilon)+\epsilon<h^{2}(\delta_{0})+(r+1)\epsilon$.
			Hence, we can bound $\delta_{N}$ by
			\begin{align}
				\delta_{N} < h^{N}(\delta_{0})  + \sum_{n=0}^{N-1}r^{n}\epsilon =  \tilde{\delta}_{N} + \frac{1-r^{N}}{1-r}\epsilon\ .
			\end{align}
			Setting $\epsilon = \mO(\eta)$, we achieve $\frac{1}{2}\lVert \tau - \rho_{N}\rVert_{1} \leq \tilde{\delta}_{N} + \eta$.
			In other words, by using circuit depth linear in $\eta$ and width exponential in $\eta$, QDP algorithm converges to the exact algorithm with error bounded by an arbitrary number $\eta>0$. 
		\end{proof}
		
		In the unfolding implementation, the circuit depth scales as $ (2L+1)^{N} \sim \log(\delta_{N}^{-1})(1-\delta_{0}^{2})^{-1/2} $ with the Grover scaling factor $(1-\delta_{0}^{2})^{-1/2}$.
		However, for sampling-based Grover search algorithms, i.e. if copies of the initial state $\psi_{0}$ are used instead of the oracle $\Eop{\alpha}{\psi_{0}}$, Ref.~\cite{Kimmel2017DME_OP} proved that at least $\mO(L^{2N}) = \mO((1-\delta_{0}^{2})^{-1})$ copies of the initial state are needed. 
		The QDP implementation also obeys this scaling, as can be seen in the exponent $(2+\Delta)N$ for the width scaling in $L$.

		\def\A{\hat A}
		\def\B{\hat B}
		\def\linops{{\mathcal L(\mathbb C^{\times D})}}
		
		\def\r{\rho^{(A)}}
		\def\rp{\tilde \rho^{(A)}}
		\def\Z{\hat Z}
		\def\X{\hat X}
		\def\Y{\hat Y}
		\def\p{\Z}
		\def\d{\Delta}
		\def\o{\sigma}
		\def\h{\hat H}
		\def\j{\hat J}
		\def\w{\hat W}
		\def\rxi{\mathcal R}
		\def\wosd{\hat W^{\mathrm{(OSD)}}}

		\section{Double bracket iterations: quantum imaginary-time evolution and oblivious Schmidt decomposition}\label{app:OSD}
		
		Two example algorithms---quantum imaginary-time evolution and oblivious schmidt decomposition---can be considered as an ordered diagonalization task: from a generic pure state to the ground state in the Hamiltonian basis in the former case and from a generic reduced state to the diagonal state in computational basis in the latter case.  
		In both cases, we use a recently-proposed quantum algorithm~\cite{Gluza2024DBI}, which we call double-bracket iterations, that adapts the continuous gradient flow for matrices for discrete gate-based quantum computing models. 
		
		\subsection{Double-bracket iterations}\label{app:DBI}
		In this section, we expound on the double-bracket iteration algorithm, which can be used for scenarios beyond quantum imaginary-time evolution or oblivious Schmidt decomposition.
		A double-bracket iteration step with an instruction operator $\hat{P}_{n}$ is defined as
		\begin{align}\label{eq:discrete_flow4}
			\hat{P}_{n+1} = e^{s_{n} [\hat{D},\hat{P}_{n}] }\hat{P}_{n}e^{-s_{n} [\hat{D},\hat{P}_{n}] }\ .
		\end{align} 
		In other words, this is a single memory-call recursion unitary, where the memory-call is defined by the linear Hermitian-preserving map $\mN(\hat{P}) = -is_{n}[\hat{D},\hat{P}]$.
		Flow duration $s_{n}$ and diagonal operator $\hat{D}$ are not part of the quantum instructions and can be chosen obliviously to $\hat{P}_{n}$, or strategically with some a priori information on $\hat{P}_{n}$.
		
		This recursion can be understood better when considering a function
		\begin{align}\label{eq:cost_func_DBI}
			f\left(\hat{P}\right) \coloneq \lVert \hat{P} - \hat{D}\rVert_{2}^{2}\ .
		\end{align}
		The generator $\frac{d}{d s_{n}}\hat{P}_{n+1} = [[\hat{D},\hat{P}_{n}],\hat{P}_{n}]$ is shown to be $-\mathrm{grad}f(\hat{P}_{n})$ for some Riemannian metric defined on the manifold of matrices isospectral to $\hat{P}_{0}$~\cite{HelmkeMoore1994Book}.
		Hence, the iteration effectively applies the gradient flow of $f$ and we expect $f$ to reduce over the iterations towards the minimum. This is in fact true for suitable choices of $s_{n}$ and $\hat{D}$.
		\begin{lem}[Theorem~4.4 in Ref.~\cite{Moore1994DiscreteDBI}]\label{lem:DBI_convergence_Moore}
			When $\hat{D}$ is a non-degenerate diagonal matrix such that $\hat{D} = \sum_{i}\mu_{i}\dm{i}$ with $\mu_{i}>\mu_{j}$ for all $i>j$ and $s_{n} = \frac{1}{4\lVert \hat{P}_{0}\rVert_{2}\lVert \hat{D}\rVert_{2}}$, the recursion Eq.~\eqref{eq:discrete_flow4} has a unique stable fixed-point $\hat{P}_{\infty} = \sum_{i}\lambda_{i}\dm{i}$, where $\lambda_{i}$ are eigenvalues of $\hat{P}_{0}$ arranged in a non-increasing order. 
			Furthermore, this fixed-point is locally exponentially stable. 
		\end{lem}
		In particular, the off-diagonal element of $\hat{P}_{k}$ is suppressed after each step 
		\begin{align}
			(\hat{P}_{k+1})_{ij} = \left[1 - \frac{(\lambda_{i}-\lambda_{j})(\mu_{i}-\mu_{j})}{4\lVert \hat{P}_{0}\rVert_{2}\lVert \hat{D}\rVert_{2}}\right](\hat{P}_{k})_{ij}\ ,
		\end{align}
		up to the first order of $(\hat{P}_{k})_{ij}$, where the suppression factor inside the square bracket is in the range $(0,1)$.
		Hence, it is possible to establish the exponential suppression of the distance.
		\begin{cor}\label{cor:DBI_distance_suppression}
			Let $\hat{P}_{\infty}$ be a fixed-point of the double-bracket iterations starting from a matrix $\hat{P}_{0}$ and let $M(\hat{P}_{0})$ be a manifold of matrices isospectral to $\hat{P}_{0}$.
			Then there exists $r<1$ and a neighbourhood $N$ of $\hat{P}_{\infty}$ in the manifold $M(\hat{P}_{0})$, such that 
			\begin{align}\label{eq:DBI_exponential_suppression}
				\lVert \hat{P}_{k+1} - \hat{P}_{\infty}\rVert_{2} < r\lVert \hat{P}_{k} - \hat{P}_{\infty}\rVert_{2}\ ,
			\end{align}
			for any $\hat{P}_{k}\in N$ and $\hat{P}_{k+1}$ obtained from an iteration from $\hat{P}_{k}$.
		\end{cor}

		\subsection{Unfolding implementation for double-bracket iterations with black box queries}
		We now describe the unfolding implementation of double-bracket iterations, given access to the black box queries to the evolution $e^{\pm it\hat{P}_{0}}$ by the root operator $\hat{P}_{0}$, following Ref.~\cite{Gluza2024DBI}.
		To implement the memory-call $e^{i\mN(\hat{P})} = e^{s_{n}[\hat{D}_{n},\hat{P}]}$ with queries $e^{\pm it\hat{P}_{0}}$, the group commutator approximation is used. 
		\begin{lem}[Group commutator, \cite{Gluza2024DBI}]
			Let $ \A $ and $ \B $ be Hermitian operators mapping $ \mH $ to $ \mH $.
			For any unitarily invariant norm $\|\cdot\|$ we have
			\begin{align}
				\left\lVert e^{-i\sqrt{s}\B} e^{-i\sqrt{s}\A} e^{i\sqrt{s}\B}e^{i\sqrt{s}\A} -e^{s[\A,\B]}\right\rVert\le s^{3/2}\left(\left\lVert [\A,[\A,\B]\right\rVert + \left\lVert [\B,[\B,\A]\right\rVert\right)\ .
			\end{align}
			\label{groupcom}
		\end{lem}

		Setting $ \A = \hat{P} $ and $ \B = -\hat{D} $ we obtain an locally accurate approximation to the recursion step
		\begin{align}\label{eq:Marek_DBI_recursion}
			\hat{W}_{n} = \left(e^{i\sqrt{\frac{s_{n}}{M}}\hat{D}_{n}}e^{-i\sqrt{\frac{s_{n}}{M}}\hat{P}_{n}}e^{-i\sqrt{\frac{s_{n}}{M}}\hat{D}_{n}}e^{i\sqrt{\frac{s_{n}}{M}}\hat{P}_{n}}\right)^{M}
		\end{align}
		that approximates $e^{s_{n}[\hat{D},\hat{P}_{n}]}$ with error $\mO(s_{n}^{3/2}M^{-1/2})$ by making $L = 2M$ calls to $P_{n}$.
		Instead of Eq.~\eqref{eq:Marek_DBI_recursion}, one can employ higher order product formulae~\cite{Chen2022ProductFormulae} that can be more complicated to execute but have higher efficiency.
		Typically, we apply this recursion to a state 
		\begin{align}\label{eq:H_encoded_state}
			\rho_{n} = \frac{\hat{P}_{n}+\lambda\1}{\Tr\left[\hat{P}_{n}+\lambda\1\right]}\ ,
		\end{align}
		encoding the matrix $\hat{P}_{n}$.
		The convergence of the recursion $\rho_{k+1} = \hat{W}_{k}\rho_{k}\hat{W}_{k}^{\dagger}$ to the diagonal state $\rho_{\infty}$ has been established, despite the error arising from the group commutator approximation~\cite{Gluza2024DBI}.
		
		The memory-call to $\hat{P}_{n}$ needed for Eq.~\eqref{eq:Marek_DBI_recursion} is not directly accessible to us.
		However, it is covariant, i.e. $e^{-it\hat{P}_{n}} = \hat{W}_{n-1}e^{-it\hat{P}_{n-1}}\hat{W}_{n-1}^{\dagger}$.
		Unfolding then follows from replacing all memory-calls in Eq.~\eqref{eq:Marek_DBI_recursion} by the call to the root state $e^{\pm it\hat{P}_{0}}$, dressed with previous recursion steps. 
		The circuit depth scaling can also be calculated for the unfolding implementation. 
		Since $L = 2M$, the depth required for performing $N$ steps of the iteration would scale as $\mO((4M+1)^{N})$.
		
		\subsection{QDP implementation for double-bracket iterations}\label{app:DBI_QDP}
		Now we analyze the QDP implementation using states $\rho_{n}$ in Eq.~\eqref{eq:H_encoded_state} as instructions for the recursion $e^{s_{n}[\hat{D},\hat{P}_{n}]}$.
		First, we establish the locally accurate implementation of a single recursion step, whose memory-call is of a linear Hermitian-preserving map $\mN(\rho) = -is[\hat{D},\rho]$.
		Observe that
		\begin{align}
			e^{s_{n}[\hat{D},\hat{P}_{n}]} = e^{i\Tr[\hat{P}_{0}+\lambda\1]\hat{\mN}(\rho_{n})}\ .
		\end{align}
		Henceforth, we rescale the flow duration $s_{n}\mapsto\Tr[\hat{P}_{0}+\lambda\1]s_{n}$ and define our recursion as $ e^{i\mN(\rho_{n})}$.
		
		The memory-call is then implemented by memory-usage queries to Hermitian-preserving matrix exponentiation channel $ \hE{r}{\mN, \rho_{n}} $, given as
		\begin{align}
			\hE{r}{\mN, \rho_{n}}(\sigma) = \Tr_{1}\left[e^{-ir\hat{N}}(\rho\otimes\sigma)e^{ir\hat{N}}\right]\ .
		\end{align}
		The fixed operation $\mathcal{Q} = e^{-ir\hat{N}}$ oblivious to $\rho_{n}$ is generated by $ \hat{N} = is_{n}\sum_{jk} (\mu_{k}-\mu_{j})\ketbra{kj}{jk}_{12} $, where $ \{\mu_{j}\}_{j} $ are the diagonal entries of $ \hat{D} $. 
		Then, $M$ queries to this channel with $r = 1/M$ gives the trace distance to the exact recursion channel
		\begin{align}
			\frac{1}{2}\left\lVert \left(\hE{1/M}{\mN, \rho_{n}}\right)^{M} - \Eopch{}{\mN, \rho_{n}}\right\rVert_{\Tr} = \mO\left(\frac{\lVert \hat{N}\rVert_{\infty}^{2}}{M}\right)\ ,
		\end{align}
		from Eq.~\eqref{eq:HME_diamond_dist_Mtimes}.
		The flow duration is chosen to be $s_{n} = \mO(\lVert \hat{D}\rVert^{-1})$ as in Lemma~\ref{lem:DBI_convergence_Moore}, which makes the error scaling independent of the rescaling of $\hat{D}$.
		Therefore, one can make the locally accurate implementation, by increasing the number of copies $M$.
		
		Now we consider the error accumulation over multiple iterations. 
		Theorem~\ref{thm:qdp} is applicable in this case from the exponential convergence in Corollary~\ref{cor:DBI_distance_suppression}, given that the initial state $\rho_{0}$ already has a small off-diagonal part, since $h(\delta) = r\delta$ for $r<1$ satisfies the requirement of the theorem. 
		Consequently, QDP double-bracket iterations can be implemented with a circuit whose depth is at most quadratic to the number of iterations $N$. 
		However, we could not find a method to curb the accumulation of non-unitary error, as opposed to the Grover search example, and thus cannot achieve linear depth and exponential width.

		\subsection{Quantum imaginary-time evolution}\label{app:QITE_detail}
		
		Up to this point, we have focused on the diagonalization of the state $\rho_{n}$ (or $\hat{P}_{n}$) as a result of double-bracket iterations. 
		However, Lemma~\ref{lem:DBI_convergence_Moore} tells us an additional property: the fixed-point is not only diagonal in the basis of $\hat{D}$, but its eigenvalues are also rearranged in an non-increasing order.
		In other words, the eigenvector of $\rho_{0}$ with the highest eigenvalue will align with the eigenvector of $\hat{D}$ that has the lowest eigenvalue. 
		
		Now suppose that $\rho_{0} = \dm{\psi}$ is a pure state. 
		Then the highest eigenvalue of it is $1$ with the corresponding eigenvector $\ket{\psi}$.
		Then, as the result of the double-bracket iterations, $\rho_{\infty}$ will become an eigenvector of $\hat{D}$ with the largest eigenvalue. 
		If we set $\hat{D} = -\hat{H}$, this means that the double-bracket iteration prepares the ground state of the Hamiltonian $\hat{H}$.
		
		In Ref.~\cite{Gluza2024DBI}, it has been found that this process is equivalent to the discretized version of quantum imaginary-time evolution.
		The goal is to apply an operator similar to the usual Hamiltonian evolution $e^{-it\hat{H}}$ but with an imaginary-time $t = -i\beta$. 
		Clearly, this is not a unitary operator, and will not result in a quantum state after the application.
		Hence, we instead need to consider a unitary flow of a pure state $\ket{\psi(\beta)} \propto e^{-\beta H}\ket{\psi(0)}$ starting $\beta = 0$ and evolving towards $\beta\to\infty$.
		The corresponding differential equation is then written as
		\begin{align}
			\partial_{\beta}\ket{\psi(\beta)} = -(\hat{H} - E_{\beta}) \ket{\psi(\beta)} \propto -\hat{H}\ket{\psi(\beta)}\ ,
		\end{align}
		where $E_{\beta} = \bra{\psi(\beta)}\hat{H}\ket{\psi(\beta)}$.
		Surprisingly, this equation is equivalent to another equation
		\begin{align}
			\partial_{\beta}\ket{\psi(\beta)} = \left[\dm{\psi(\beta)},\hat{H}\right]\ket{\psi(\beta)}\ ,
		\end{align}
		which can be discretized into finite steps with size $s_{n}$ as
		\begin{align}\label{eq:DBI-QITE}
			\psivec{n+1} = e^{s_{n}[\dm{\psi_{n}},\hat{H}]}\psivec{n}\ .
		\end{align}
		This comes from the fact that the cost function $f$ in Eq.~\eqref{eq:cost_func_DBI} can be written as $f(\dm{\psi}) = 1+ \Tr[\hat{H}^{2}] + 2\bra{\psi}\hat{H}\ket{\psi}$, and the double-bracket iteration, which is a gradient flow in terms of this function, is effectively giving the gradient flow in terms of the average energy of the pure state $\ket{\psi}$.
		Note that Eq.~\eqref{eq:DBI-QITE} is exactly in the form of Eq.~\eqref{eq:discrete_flow4} with $\hat{D} = -\hat{H}$ and $\hat{P}_{n} = \dm{\psi_{n}}$.
		
		Furthermore, Theorem~2 of Ref.~\cite{Gluza2024DBI} guarantees the fast spectral convergence of this algorithm. 
		Therefore, we can apply Thm.~\ref{thm:qdp_pure}, which ensures the efficient QDP implementation. 
		
		One potential difficulty of implementing the memory-call $e^{s_{n}[\dm{\psi_{n}},\hat{H}]}$ directly from memory-usage queries is the requirement for a prior information about the Hamiltonian $\hat{H}$. 
		Hence, in the main text, we considered a sampling based Hamiltonian simulation scenario studied in Refs.~\cite{Lloyd2014quantum, Kimmel2017DME_OP}.
		In such settings, many copies of the state $\chi \propto \hat{H} - \lambda_0 \1$ are given. 
		Then, instead of considering the single variable function $\hmN(\rho) = -i[\rho, \hat{H}]$ as the one defining the memory-call, we can work with a two-variable function $f(\rho,\chi) = -i[\rho,\chi]$. 
		Such memory-usage queries can be made using the ideas from Lemma~\ref{lem:poly_corresponding_lin_map}.

		\subsection{Oblivious Schmidt decomposition}\label{app:OSD_detail}
		
		Any bi-partite pure state $ \psi $ has a Schmidt decomposition, which reads
		\begin{align}
			\ket{\psi}_{AB} = \sum_{j}\sqrt{\lambda_{j}}\ket{\phi_{j}}_{A}\ket{\chi_{j}}_{B}\ ,
		\end{align}
		where $ \{\ket{\phi_{j}}_{A}\}_{j} $ and $ \{\ket{\chi_{j}}_{B}\}_{j} $ form orthonormal bases for respective Hilbert spaces $ \mH_{A} $ and $ \mH_{B} $. 
		We will refer to such bases as Schmidt bases and (squares of) coefficients $ \lambda_{j} $ as Schmidt coefficients. 
		Schmidt coefficients, in particular, provide the full information on how entangled the state is~\cite{Vidal2000Schmidt}. 
		Schmidt bases can be connected to computational bases $ \{\ket{j}_{A}\}_{j} $ and $ \{\ket{j}_{B}\}_{j} $ through local unitary transfornations	$ \ket{\phi_{j}}_{A} = \hat{V}_{A}^{\dagger}\ket{j}_{A} $ and $ \ket{\chi_{j}}_{A} = \hat{V}_{B}^{\dagger}\ket{j}_{B} $, for all $ j $.
		
		The oblivious Schmidt decomposition algorithm aims to access the Schmidt coefficients and Schmidt basis of a given pure state, without learning the classical description of the state first via, e.g. state tomography. 
		In this section, we demonstrate that this task is achievable by diagonalizing the reduced state of a given pure state $ \psi $ via double-bracket iteration.   
		
		We first describe how double-bracket iterations on the reduced system can be implemented with memory-usage queries using the entire bi-partite state. 
		Suppose that multiple copies of the instruction state $\ket{\psi}_{AB}$ are given. 
		The double-bracket iteration for the Oblivious Schmidt decomposition is given as
		\begin{align}
			e^{i\mN(\psi_{n})} = e^{s_{n}[\hat{D}, \rho_{n}^{(A)}]} \otimes \1_{B}\ ,
		\end{align}
		where $\rho_{n}^{(A)} = \Tr_B[\psi_{n}]$.
		Consider the memory-usage query
		\begin{align}
			\hE{r}{\mN, \psi_{A_{1}B_{1}}}(\sigma_{A_{2}B_{2}}) = \Tr_{A_{1}B_{1}}\left[e^{-ir\hat{N}}(\psi_{A_{1}B_{1}}\otimes\sigma_{A_{2}B_{2}})e^{ir\hat{N}}\right]\ ,
		\end{align}
		where $\hat{N} =  is_{n}\sum_{jk} (\mu_{k}-\mu_{j})\ketbra{kj}{jk}_{A_{1}A_{2}} \otimes \1_{B_{1}B_{2}} $ and $ \{\mu_{j}\}_{j} $ are the diagonal entries of $ \hat{D} $. 
		Each query $\hE{r}{\mN, \psi_{A_{1}B_{1}}}(\sigma_{A_{2}B_{2}}) $ consumes a copy of $\psi_{A_{1}B_{1}}$.
		Furthermore, following the same calculation from the previous subsection, we obtain
		\begin{align}
			\frac{1}{2}\left\lVert \left(\hE{1/M}{\mN, \psi_{A_{1}B_{1}}}\right)^{M} -e^{i\mN(\psi_{n})} \right\rVert_{\Tr} = \mO\left(\frac{\lVert \hat{N}\rVert_{\infty}^{2}}{M}\right)\ ,
		\end{align}
		the locally accurate implementation of the recursion unitary.

		If recursions are implemented with sufficiently good precision, Theorem~1 in the main text holds and we approach the fixed-point of the recursion, which is a diagonal matrix $\rho_{\infty}^{(A)}$ isospectral to the initial reduced state $\rho_{0}^{(A)}$. 
		In other words, $ \rho_{N}^{(A)} $ with some large number $ N $ becomes close to the desired diagonal matrix
		\begin{align}
			\rho_{N}\simeq \sum_{j}\lambda_{j}\dm{j}_{A}\ ,
		\end{align}
		from which Schmidt coefficients $ \{\lambda_{j}\}_{j} $ can be extracted.


\begin{thebibliography}{64}%
	\makeatletter
	\providecommand \@ifxundefined [1]{%
		\@ifx{#1\undefined}
	}%
	\providecommand \@ifnum [1]{%
		\ifnum #1\expandafter \@firstoftwo
		\else \expandafter \@secondoftwo
		\fi
	}%
	\providecommand \@ifx [1]{%
		\ifx #1\expandafter \@firstoftwo
		\else \expandafter \@secondoftwo
		\fi
	}%
	\providecommand \natexlab [1]{#1}%
	\providecommand \enquote  [1]{``#1''}%
	\providecommand \bibnamefont  [1]{#1}%
	\providecommand \bibfnamefont [1]{#1}%
	\providecommand \citenamefont [1]{#1}%
	\providecommand \href@noop [0]{\@secondoftwo}%
	\providecommand \href [0]{\begingroup \@sanitize@url \@href}%
	\providecommand \@href[1]{\@@startlink{#1}\@@href}%
	\providecommand \@@href[1]{\endgroup#1\@@endlink}%
	\providecommand \@sanitize@url [0]{\catcode `\\12\catcode `\$12\catcode
		`\&12\catcode `\#12\catcode `\^12\catcode `\_12\catcode `\%12\relax}%
	\providecommand \@@startlink[1]{}%
	\providecommand \@@endlink[0]{}%
	\providecommand \url  [0]{\begingroup\@sanitize@url \@url }%
	\providecommand \@url [1]{\endgroup\@href {#1}{\urlprefix }}%
	\providecommand \urlprefix  [0]{URL }%
	\providecommand \Eprint [0]{\href }%
	\providecommand \doibase [0]{https://doi.org/}%
	\providecommand \selectlanguage [0]{\@gobble}%
	\providecommand \bibinfo  [0]{\@secondoftwo}%
	\providecommand \bibfield  [0]{\@secondoftwo}%
	\providecommand \translation [1]{[#1]}%
	\providecommand \BibitemOpen [0]{}%
	\providecommand \bibitemStop [0]{}%
	\providecommand \bibitemNoStop [0]{.\EOS\space}%
	\providecommand \EOS [0]{\spacefactor3000\relax}%
	\providecommand \BibitemShut  [1]{\csname bibitem#1\endcsname}%
	\let\auto@bib@innerbib\@empty
	\bibitem [{\citenamefont {Michie}(1968)}]{Michie1968memoization}%
	\BibitemOpen
	\bibfield  {author} {\bibinfo {author} {\bibfnamefont {D.}~\bibnamefont
			{Michie}},\ }\bibfield  {title} {\bibinfo {title} {``{M}emo'' functions and
			machine learning},\ }\href {https://doi.org/10.1038/218019a0} {\bibfield
		{journal} {\bibinfo  {journal} {Nature}\ }\textbf {\bibinfo {volume} {218}},\
		\bibinfo {pages} {19} (\bibinfo {year} {1968})}\BibitemShut {NoStop}%
	\bibitem [{\citenamefont {Bellman}(1952)}]{Bellman1952_DP}%
	\BibitemOpen
	\bibfield  {author} {\bibinfo {author} {\bibfnamefont {R.}~\bibnamefont
			{Bellman}},\ }\bibfield  {title} {\bibinfo {title} {On the theory of dynamic
			programming},\ }\href {https://doi.org/https://doi.org/10.1073/pnas.38.8.716}
	{\bibfield  {journal} {\bibinfo  {journal} {Proc. Natl. Acad. Sci. U.S.A.}\
		}\textbf {\bibinfo {volume} {38}},\ \bibinfo {pages} {716} (\bibinfo {year}
		{1952})}\BibitemShut {NoStop}%
	\bibitem [{\citenamefont {Bellman}(1966)}]{Bellman_dynamic}%
	\BibitemOpen
	\bibfield  {author} {\bibinfo {author} {\bibfnamefont {R.}~\bibnamefont
			{Bellman}},\ }\bibfield  {title} {\bibinfo {title} {Dynamic programming},\
	}\href {https://doi.org/10.1126/science.153.3731.34} {\bibfield  {journal}
		{\bibinfo  {journal} {Science}\ }\textbf {\bibinfo {volume} {153}},\ \bibinfo
		{pages} {34} (\bibinfo {year} {1966})}\BibitemShut {NoStop}%
	\bibitem [{\citenamefont {Grover}(2005)}]{Grover2005FP}%
	\BibitemOpen
	\bibfield  {author} {\bibinfo {author} {\bibfnamefont {L.~K.}\ \bibnamefont
			{Grover}},\ }\bibfield  {title} {\bibinfo {title} {Fixed-point quantum
			search},\ }\href {https://doi.org/10.1103/PhysRevLett.95.150501} {\bibfield
		{journal} {\bibinfo  {journal} {Phys. Rev. Lett.}\ }\textbf {\bibinfo
			{volume} {95}},\ \bibinfo {pages} {150501} (\bibinfo {year}
		{2005})}\BibitemShut {NoStop}%
	\bibitem [{\citenamefont {Yoder}\ \emph {et~al.}(2014)\citenamefont {Yoder},
		\citenamefont {Low},\ and\ \citenamefont {Chuang}}]{Yoder2014Grover}%
	\BibitemOpen
	\bibfield  {author} {\bibinfo {author} {\bibfnamefont {T.~J.}\ \bibnamefont
			{Yoder}}, \bibinfo {author} {\bibfnamefont {G.~H.}\ \bibnamefont {Low}},\
		and\ \bibinfo {author} {\bibfnamefont {I.~L.}\ \bibnamefont {Chuang}},\
	}\bibfield  {title} {\bibinfo {title} {Fixed-point quantum search with an
			optimal number of queries},\ }\href
	{https://doi.org/10.1103/PhysRevLett.113.210501} {\bibfield  {journal}
		{\bibinfo  {journal} {Phys. Rev. Lett.}\ }\textbf {\bibinfo {volume} {113}},\
		\bibinfo {pages} {210501} (\bibinfo {year} {2014})}\BibitemShut {NoStop}%
	\bibitem [{\citenamefont {Gluza}(2024)}]{Gluza2024DBI}%
	\BibitemOpen
	\bibfield  {author} {\bibinfo {author} {\bibfnamefont {M.}~\bibnamefont
			{Gluza}},\ }\bibfield  {title} {\bibinfo {title} {Double-bracket quantum
			algorithms for diagonalization},\ }\href
	{https://doi.org/10.22331/q-2024-04-09-1316} {\bibfield  {journal} {\bibinfo
			{journal} {{Quantum}}\ }\textbf {\bibinfo {volume} {8}},\ \bibinfo {pages}
		{1316} (\bibinfo {year} {2024})}\BibitemShut {NoStop}%
	\bibitem [{\citenamefont {Robbiati}\ \emph {et~al.}(2024)\citenamefont
		{Robbiati}, \citenamefont {Pedicillo}, \citenamefont {Pasquale},
		\citenamefont {Li}, \citenamefont {Wright}, \citenamefont {Farias},
		\citenamefont {Giang}, \citenamefont {Son}, \citenamefont {Kn{\"o}rzer},
		\citenamefont {Goh}, \citenamefont {Khoo}, \citenamefont {Ng}, \citenamefont
		{Holmes}, \citenamefont {Carrazza},\ and\ \citenamefont
		{Gluza}}]{Robbiati2024DBQA}%
	\BibitemOpen
	\bibfield  {author} {\bibinfo {author} {\bibfnamefont {M.}~\bibnamefont
			{Robbiati}}, \bibinfo {author} {\bibfnamefont {E.}~\bibnamefont {Pedicillo}},
		\bibinfo {author} {\bibfnamefont {A.}~\bibnamefont {Pasquale}}, \bibinfo
		{author} {\bibfnamefont {X.}~\bibnamefont {Li}}, \bibinfo {author}
		{\bibfnamefont {A.}~\bibnamefont {Wright}}, \bibinfo {author} {\bibfnamefont
			{R.~M.~S.}\ \bibnamefont {Farias}}, \bibinfo {author} {\bibfnamefont {K.~U.}\
			\bibnamefont {Giang}}, \bibinfo {author} {\bibfnamefont {J.}~\bibnamefont
			{Son}}, \bibinfo {author} {\bibfnamefont {J.}~\bibnamefont {Kn{\"o}rzer}},
		\bibinfo {author} {\bibfnamefont {S.~T.}\ \bibnamefont {Goh}}, \bibinfo
		{author} {\bibfnamefont {J.~Y.}\ \bibnamefont {Khoo}}, \bibinfo {author}
		{\bibfnamefont {N.~H.~Y.}\ \bibnamefont {Ng}}, \bibinfo {author}
		{\bibfnamefont {Z.}~\bibnamefont {Holmes}}, \bibinfo {author} {\bibfnamefont
			{S.}~\bibnamefont {Carrazza}},\ and\ \bibinfo {author} {\bibfnamefont
			{M.}~\bibnamefont {Gluza}},\ }\href {https://arxiv.org/abs/2408.03987}
	{\bibinfo {title} {Double-bracket quantum algorithms for high-fidelity ground
			state preparation}} (\bibinfo {year} {2024}),\ \Eprint
	{https://arxiv.org/abs/2408.03987} {arXiv:2408.03987 [quant-ph]} \BibitemShut
	{NoStop}%
	\bibitem [{\citenamefont {Xiaoyue}\ \emph {et~al.}(2024)\citenamefont
		{Xiaoyue}, \citenamefont {Robbiati}, \citenamefont {Pasquale}, \citenamefont
		{Pedicillo}, \citenamefont {Wright}, \citenamefont {Carrazza},\ and\
		\citenamefont {Gluza}}]{Xiaoyue2024DBQA}%
	\BibitemOpen
	\bibfield  {author} {\bibinfo {author} {\bibfnamefont {L.}~\bibnamefont
			{Xiaoyue}}, \bibinfo {author} {\bibfnamefont {M.}~\bibnamefont {Robbiati}},
		\bibinfo {author} {\bibfnamefont {A.}~\bibnamefont {Pasquale}}, \bibinfo
		{author} {\bibfnamefont {E.}~\bibnamefont {Pedicillo}}, \bibinfo {author}
		{\bibfnamefont {A.}~\bibnamefont {Wright}}, \bibinfo {author} {\bibfnamefont
			{S.}~\bibnamefont {Carrazza}},\ and\ \bibinfo {author} {\bibfnamefont
			{M.}~\bibnamefont {Gluza}},\ }\href {https://arxiv.org/abs/2408.07431}
	{\bibinfo {title} {Strategies for optimizing double-bracket quantum
			algorithms}} (\bibinfo {year} {2024}),\ \Eprint
	{https://arxiv.org/abs/2408.07431} {arXiv:2408.07431 [quant-ph]} \BibitemShut
	{NoStop}%
	\bibitem [{\citenamefont {Gluza}\ \emph {et~al.}(2024)\citenamefont {Gluza},
		\citenamefont {Son}, \citenamefont {Tiang}, \citenamefont {Suzuki},
		\citenamefont {Holmes},\ and\ \citenamefont {Ng}}]{QITEDBF}%
	\BibitemOpen
	\bibfield  {author} {\bibinfo {author} {\bibfnamefont {M.}~\bibnamefont
			{Gluza}}, \bibinfo {author} {\bibfnamefont {J.}~\bibnamefont {Son}}, \bibinfo
		{author} {\bibfnamefont {B.~H.}\ \bibnamefont {Tiang}}, \bibinfo {author}
		{\bibfnamefont {Y.}~\bibnamefont {Suzuki}}, \bibinfo {author} {\bibfnamefont
			{Z.}~\bibnamefont {Holmes}},\ and\ \bibinfo {author} {\bibfnamefont
			{N.~H.~Y.}\ \bibnamefont {Ng}},\ }\href {https://arxiv.org/abs/2412.04554}
	{\bibinfo {title} {Double-bracket quantum algorithms for quantum
			imaginary-time evolution}} (\bibinfo {year} {2024}),\ \Eprint
	{https://arxiv.org/abs/2412.04554} {arXiv:2412.04554 [quant-ph]} \BibitemShut
	{NoStop}%
	\bibitem [{\citenamefont {Gisin}\ and\ \citenamefont
		{Massar}(1997)}]{Gisin1997Cloning}%
	\BibitemOpen
	\bibfield  {author} {\bibinfo {author} {\bibfnamefont {N.}~\bibnamefont
			{Gisin}}\ and\ \bibinfo {author} {\bibfnamefont {S.}~\bibnamefont {Massar}},\
	}\bibfield  {title} {\bibinfo {title} {Optimal quantum cloning machines},\
	}\href {https://doi.org/10.1103/PhysRevLett.79.2153} {\bibfield  {journal}
		{\bibinfo  {journal} {Phys. Rev. Lett.}\ }\textbf {\bibinfo {volume} {79}},\
		\bibinfo {pages} {2153} (\bibinfo {year} {1997})}\BibitemShut {NoStop}%
	\bibitem [{\citenamefont {Wootters}\ and\ \citenamefont
		{Zurek}(1982)}]{Wootters1982Nocloning}%
	\BibitemOpen
	\bibfield  {author} {\bibinfo {author} {\bibfnamefont {W.~K.}\ \bibnamefont
			{Wootters}}\ and\ \bibinfo {author} {\bibfnamefont {W.~H.}\ \bibnamefont
			{Zurek}},\ }\bibfield  {title} {\bibinfo {title} {A single quantum cannot be
			cloned},\ }\href {https://doi.org/10.1038/299802a0} {\bibfield  {journal}
		{\bibinfo  {journal} {Nature}\ }\textbf {\bibinfo {volume} {299}},\ \bibinfo
		{pages} {802} (\bibinfo {year} {1982})}\BibitemShut {NoStop}%
	\bibitem [{\citenamefont {Haah}\ \emph {et~al.}(2016)\citenamefont {Haah},
		\citenamefont {Harrow}, \citenamefont {Ji}, \citenamefont {Wu},\ and\
		\citenamefont {Yu}}]{Haah2016_tomography}%
	\BibitemOpen
	\bibfield  {author} {\bibinfo {author} {\bibfnamefont {J.}~\bibnamefont
			{Haah}}, \bibinfo {author} {\bibfnamefont {A.~W.}\ \bibnamefont {Harrow}},
		\bibinfo {author} {\bibfnamefont {Z.}~\bibnamefont {Ji}}, \bibinfo {author}
		{\bibfnamefont {X.}~\bibnamefont {Wu}},\ and\ \bibinfo {author}
		{\bibfnamefont {N.}~\bibnamefont {Yu}},\ }\bibfield  {title} {\bibinfo
		{title} {Sample-optimal tomography of quantum states},\ }in\ \href
	{https://doi.org/10.1145/2897518.2897585} {\emph {\bibinfo {booktitle}
			{Proceedings of the Forty-Eighth Annual ACM Symposium on Theory of
				Computing}}},\ \bibinfo {series and number} {STOC '16}\ (\bibinfo
	{publisher} {Association for Computing Machinery},\ \bibinfo {address} {New
		York, NY, USA},\ \bibinfo {year} {2016})\ pp.\ \bibinfo {pages}
	{913--925}\BibitemShut {NoStop}%
	\bibitem [{\citenamefont {Lloyd}\ \emph {et~al.}(2014)\citenamefont {Lloyd},
		\citenamefont {Mohseni},\ and\ \citenamefont
		{Rebentrost}}]{Lloyd2014quantum}%
	\BibitemOpen
	\bibfield  {author} {\bibinfo {author} {\bibfnamefont {S.}~\bibnamefont
			{Lloyd}}, \bibinfo {author} {\bibfnamefont {M.}~\bibnamefont {Mohseni}},\
		and\ \bibinfo {author} {\bibfnamefont {P.}~\bibnamefont {Rebentrost}},\
	}\bibfield  {title} {\bibinfo {title} {Quantum principal component
			analysis},\ }\href {https://doi.org/10.1038/nphys3029} {\bibfield  {journal}
		{\bibinfo  {journal} {Nat. Phys.}\ }\textbf {\bibinfo {volume} {10}},\
		\bibinfo {pages} {631} (\bibinfo {year} {2014})}\BibitemShut {NoStop}%
	\bibitem [{\citenamefont {Marvian}\ and\ \citenamefont
		{Lloyd}(2016)}]{Marvian2016_emulator}%
	\BibitemOpen
	\bibfield  {author} {\bibinfo {author} {\bibfnamefont {I.}~\bibnamefont
			{Marvian}}\ and\ \bibinfo {author} {\bibfnamefont {S.}~\bibnamefont
			{Lloyd}},\ }\href@noop {} {\bibinfo {title} {Universal quantum emulator}}
	(\bibinfo {year} {2016}),\ \Eprint {https://arxiv.org/abs/1606.02734}
	{arXiv:1606.02734 [quant-ph]} \BibitemShut {NoStop}%
	\bibitem [{\citenamefont {Pichler}\ \emph {et~al.}(2016)\citenamefont
		{Pichler}, \citenamefont {Zhu}, \citenamefont {Seif}, \citenamefont
		{Zoller},\ and\ \citenamefont {Hafezi}}]{Pichler2016DME}%
	\BibitemOpen
	\bibfield  {author} {\bibinfo {author} {\bibfnamefont {H.}~\bibnamefont
			{Pichler}}, \bibinfo {author} {\bibfnamefont {G.}~\bibnamefont {Zhu}},
		\bibinfo {author} {\bibfnamefont {A.}~\bibnamefont {Seif}}, \bibinfo {author}
		{\bibfnamefont {P.}~\bibnamefont {Zoller}},\ and\ \bibinfo {author}
		{\bibfnamefont {M.}~\bibnamefont {Hafezi}},\ }\bibfield  {title} {\bibinfo
		{title} {Measurement protocol for the entanglement spectrum of cold atoms},\
	}\href {https://doi.org/10.1103/PhysRevX.6.041033} {\bibfield  {journal}
		{\bibinfo  {journal} {Phys. Rev. X}\ }\textbf {\bibinfo {volume} {6}},\
		\bibinfo {pages} {041033} (\bibinfo {year} {2016})}\BibitemShut {NoStop}%
	\bibitem [{\citenamefont {Kimmel}\ \emph {et~al.}(2017)\citenamefont {Kimmel},
		\citenamefont {Lin}, \citenamefont {Low}, \citenamefont {Ozols},\ and\
		\citenamefont {Yoder}}]{Kimmel2017DME_OP}%
	\BibitemOpen
	\bibfield  {author} {\bibinfo {author} {\bibfnamefont {S.}~\bibnamefont
			{Kimmel}}, \bibinfo {author} {\bibfnamefont {C.~Y.-Y.}\ \bibnamefont {Lin}},
		\bibinfo {author} {\bibfnamefont {G.~H.}\ \bibnamefont {Low}}, \bibinfo
		{author} {\bibfnamefont {M.}~\bibnamefont {Ozols}},\ and\ \bibinfo {author}
		{\bibfnamefont {T.~J.}\ \bibnamefont {Yoder}},\ }\bibfield  {title} {\bibinfo
		{title} {Hamiltonian simulation with optimal sample complexity},\ }\href
	{https://doi.org/10.1038/s41534-017-0013-7} {\bibfield  {journal} {\bibinfo
			{journal} {npj Quantum Inf.}\ }\textbf {\bibinfo {volume} {3}},\ \bibinfo
		{pages} {13} (\bibinfo {year} {2017})}\BibitemShut {NoStop}%
	\bibitem [{\citenamefont {Kjaergaard}\ \emph {et~al.}(2022)\citenamefont
		{Kjaergaard}, \citenamefont {Schwartz}, \citenamefont {Greene}, \citenamefont
		{Samach}, \citenamefont {Bengtsson}, \citenamefont {O'Keeffe}, \citenamefont
		{McNally}, \citenamefont {Braum\"uller}, \citenamefont {Kim}, \citenamefont
		{Krantz}, \citenamefont {Marvian}, \citenamefont {Melville}, \citenamefont
		{Niedzielski}, \citenamefont {Sung}, \citenamefont {Winik}, \citenamefont
		{Yoder}, \citenamefont {Rosenberg}, \citenamefont {Obenland}, \citenamefont
		{Lloyd}, \citenamefont {Orlando}, \citenamefont {Marvian}, \citenamefont
		{Gustavsson},\ and\ \citenamefont {Oliver}}]{Kjaergaard2022DME}%
	\BibitemOpen
	\bibfield  {author} {\bibinfo {author} {\bibfnamefont {M.}~\bibnamefont
			{Kjaergaard}}, \bibinfo {author} {\bibfnamefont {M.~E.}\ \bibnamefont
			{Schwartz}}, \bibinfo {author} {\bibfnamefont {A.}~\bibnamefont {Greene}},
		\bibinfo {author} {\bibfnamefont {G.~O.}\ \bibnamefont {Samach}}, \bibinfo
		{author} {\bibfnamefont {A.}~\bibnamefont {Bengtsson}}, \bibinfo {author}
		{\bibfnamefont {M.}~\bibnamefont {O'Keeffe}}, \bibinfo {author}
		{\bibfnamefont {C.~M.}\ \bibnamefont {McNally}}, \bibinfo {author}
		{\bibfnamefont {J.}~\bibnamefont {Braum\"uller}}, \bibinfo {author}
		{\bibfnamefont {D.~K.}\ \bibnamefont {Kim}}, \bibinfo {author} {\bibfnamefont
			{P.}~\bibnamefont {Krantz}}, \bibinfo {author} {\bibfnamefont
			{M.}~\bibnamefont {Marvian}}, \bibinfo {author} {\bibfnamefont
			{A.}~\bibnamefont {Melville}}, \bibinfo {author} {\bibfnamefont {B.~M.}\
			\bibnamefont {Niedzielski}}, \bibinfo {author} {\bibfnamefont
			{Y.}~\bibnamefont {Sung}}, \bibinfo {author} {\bibfnamefont {R.}~\bibnamefont
			{Winik}}, \bibinfo {author} {\bibfnamefont {J.}~\bibnamefont {Yoder}},
		\bibinfo {author} {\bibfnamefont {D.}~\bibnamefont {Rosenberg}}, \bibinfo
		{author} {\bibfnamefont {K.}~\bibnamefont {Obenland}}, \bibinfo {author}
		{\bibfnamefont {S.}~\bibnamefont {Lloyd}}, \bibinfo {author} {\bibfnamefont
			{T.~P.}\ \bibnamefont {Orlando}}, \bibinfo {author} {\bibfnamefont
			{I.}~\bibnamefont {Marvian}}, \bibinfo {author} {\bibfnamefont
			{S.}~\bibnamefont {Gustavsson}},\ and\ \bibinfo {author} {\bibfnamefont
			{W.~D.}\ \bibnamefont {Oliver}},\ }\bibfield  {title} {\bibinfo {title}
		{Demonstration of density matrix exponentiation using a superconducting
			quantum processor},\ }\href
	{https://doi.org/https://doi.org/10.1103/PhysRevX.12.011005} {\bibfield
		{journal} {\bibinfo  {journal} {Phys. Rev. X}\ }\textbf {\bibinfo {volume}
			{12}},\ \bibinfo {pages} {011005} (\bibinfo {year} {2022})}\BibitemShut
	{NoStop}%
	\bibitem [{\citenamefont {Wei}\ \emph {et~al.}(2023)\citenamefont {Wei},
		\citenamefont {Liu}, \citenamefont {Liu}, \citenamefont {Han}, \citenamefont
		{Ma}, \citenamefont {Deng},\ and\ \citenamefont
		{Liu}}]{Wei2023hermpreserving}%
	\BibitemOpen
	\bibfield  {author} {\bibinfo {author} {\bibfnamefont {F.}~\bibnamefont
			{Wei}}, \bibinfo {author} {\bibfnamefont {Z.}~\bibnamefont {Liu}}, \bibinfo
		{author} {\bibfnamefont {G.}~\bibnamefont {Liu}}, \bibinfo {author}
		{\bibfnamefont {Z.}~\bibnamefont {Han}}, \bibinfo {author} {\bibfnamefont
			{X.}~\bibnamefont {Ma}}, \bibinfo {author} {\bibfnamefont {D.-L.}\
			\bibnamefont {Deng}},\ and\ \bibinfo {author} {\bibfnamefont
			{Z.}~\bibnamefont {Liu}},\ }\href {https://arxiv.org/abs/2308.07956}
	{\bibinfo {title} {Realizing non-physical actions through
			hermitian-preserving map exponentiation}} (\bibinfo {year} {2023}),\ \Eprint
	{https://arxiv.org/abs/2308.07956} {arXiv:2308.07956 [quant-ph]} \BibitemShut
	{NoStop}%
	\bibitem [{\citenamefont {Patel}\ and\ \citenamefont
		{Wilde}(2023{\natexlab{a}})}]{Patel2023WML1}%
	\BibitemOpen
	\bibfield  {author} {\bibinfo {author} {\bibfnamefont {D.}~\bibnamefont
			{Patel}}\ and\ \bibinfo {author} {\bibfnamefont {M.~M.}\ \bibnamefont
			{Wilde}},\ }\bibfield  {title} {\bibinfo {title} {Wave matrix lindbladization
			i: Quantum programs for simulating markovian dynamics},\ }\href
	{https://doi.org/10.1142/S1230161223500105} {\bibfield  {journal} {\bibinfo
			{journal} {Open Syst. Inf. Dyn.}\ }\textbf {\bibinfo {volume} {30}},\
		\bibinfo {pages} {2350010} (\bibinfo {year} {2023}{\natexlab{a}})},\ \Eprint
	{https://arxiv.org/abs/https://doi.org/10.1142/S1230161223500105}
	{https://doi.org/10.1142/S1230161223500105} \BibitemShut {NoStop}%
	\bibitem [{\citenamefont {Patel}\ and\ \citenamefont
		{Wilde}(2023{\natexlab{b}})}]{Patel2023WML2}%
	\BibitemOpen
	\bibfield  {author} {\bibinfo {author} {\bibfnamefont {D.}~\bibnamefont
			{Patel}}\ and\ \bibinfo {author} {\bibfnamefont {M.~M.}\ \bibnamefont
			{Wilde}},\ }\bibfield  {title} {\bibinfo {title} {Wave matrix lindbladization
			ii: General lindbladians, linear combinations, and polynomials},\ }\href
	{https://doi.org/10.1142/S1230161223500142} {\bibfield  {journal} {\bibinfo
			{journal} {Open Syst. Inf. Dyn.}\ }\textbf {\bibinfo {volume} {30}},\
		\bibinfo {pages} {2350014} (\bibinfo {year} {2023}{\natexlab{b}})},\ \Eprint
	{https://arxiv.org/abs/https://doi.org/10.1142/S1230161223500142}
	{https://doi.org/10.1142/S1230161223500142} \BibitemShut {NoStop}%
	\bibitem [{\citenamefont {Rodriguez-Grasa}\ \emph {et~al.}(2023)\citenamefont
		{Rodriguez-Grasa}, \citenamefont {Ibarrondo}, \citenamefont {Gonzalez-Conde},
		\citenamefont {Ban}, \citenamefont {Rebentrost},\ and\ \citenamefont
		{Sanz}}]{Rodriguezgrasa2023cloningDME}%
	\BibitemOpen
	\bibfield  {author} {\bibinfo {author} {\bibfnamefont {P.}~\bibnamefont
			{Rodriguez-Grasa}}, \bibinfo {author} {\bibfnamefont {R.}~\bibnamefont
			{Ibarrondo}}, \bibinfo {author} {\bibfnamefont {J.}~\bibnamefont
			{Gonzalez-Conde}}, \bibinfo {author} {\bibfnamefont {Y.}~\bibnamefont {Ban}},
		\bibinfo {author} {\bibfnamefont {P.}~\bibnamefont {Rebentrost}},\ and\
		\bibinfo {author} {\bibfnamefont {M.}~\bibnamefont {Sanz}},\ }\href
	{https://arxiv.org/abs/2311.11751} {\bibinfo {title} {Quantum approximated
			cloning-assisted density matrix exponentiation}} (\bibinfo {year} {2023}),\
	\Eprint {https://arxiv.org/abs/2311.11751} {arXiv:2311.11751 [quant-ph]}
	\BibitemShut {NoStop}%
	\bibitem [{\citenamefont {Go}\ \emph {et~al.}(2024)\citenamefont {Go},
		\citenamefont {Kwon}, \citenamefont {Park}, \citenamefont {Patel},\ and\
		\citenamefont {Wilde}}]{Go2024DME}%
	\BibitemOpen
	\bibfield  {author} {\bibinfo {author} {\bibfnamefont {B.}~\bibnamefont
			{Go}}, \bibinfo {author} {\bibfnamefont {H.}~\bibnamefont {Kwon}}, \bibinfo
		{author} {\bibfnamefont {S.}~\bibnamefont {Park}}, \bibinfo {author}
		{\bibfnamefont {D.}~\bibnamefont {Patel}},\ and\ \bibinfo {author}
		{\bibfnamefont {M.~M.}\ \bibnamefont {Wilde}},\ }\href
	{https://arxiv.org/abs/2412.02134} {\bibinfo {title} {Density matrix
			exponentiation and sample-based hamiltonian simulation: Non-asymptotic
			analysis of sample complexity}} (\bibinfo {year} {2024}),\ \Eprint
	{https://arxiv.org/abs/2412.02134} {arXiv:2412.02134 [quant-ph]} \BibitemShut
	{NoStop}%
	\bibitem [{\citenamefont {Schoute}\ \emph {et~al.}(2024)\citenamefont
		{Schoute}, \citenamefont {Grinko}, \citenamefont {Subasi},\ and\
		\citenamefont {Volkoff}}]{Schoute2024QProgrammableReflections}%
	\BibitemOpen
	\bibfield  {author} {\bibinfo {author} {\bibfnamefont {E.}~\bibnamefont
			{Schoute}}, \bibinfo {author} {\bibfnamefont {D.}~\bibnamefont {Grinko}},
		\bibinfo {author} {\bibfnamefont {Y.}~\bibnamefont {Subasi}},\ and\ \bibinfo
		{author} {\bibfnamefont {T.}~\bibnamefont {Volkoff}},\ }\href
	{https://arxiv.org/abs/2411.03648} {\bibinfo {title} {Quantum programmable
			reflections}} (\bibinfo {year} {2024}),\ \Eprint
	{https://arxiv.org/abs/2411.03648} {arXiv:2411.03648 [quant-ph]} \BibitemShut
	{NoStop}%
	\bibitem [{\citenamefont {Huggins}\ and\ \citenamefont
		{McClean}(2024)}]{Huggins2023_precomputation}%
	\BibitemOpen
	\bibfield  {author} {\bibinfo {author} {\bibfnamefont {W.~J.}\ \bibnamefont
			{Huggins}}\ and\ \bibinfo {author} {\bibfnamefont {J.~R.}\ \bibnamefont
			{McClean}},\ }\bibfield  {title} {\bibinfo {title} {Accelerating {Q}uantum
			{A}lgorithms with {P}recomputation},\ }\href
	{https://doi.org/10.22331/q-2024-02-22-1264} {\bibfield  {journal} {\bibinfo
			{journal} {{Quantum}}\ }\textbf {\bibinfo {volume} {8}},\ \bibinfo {pages}
		{1264} (\bibinfo {year} {2024})}\BibitemShut {NoStop}%
	\bibitem [{\citenamefont {Grover}(1996)}]{Grover96}%
	\BibitemOpen
	\bibfield  {author} {\bibinfo {author} {\bibfnamefont {L.~K.}\ \bibnamefont
			{Grover}},\ }\bibfield  {title} {\bibinfo {title} {A fast quantum mechanical
			algorithm for database search},\ }in\ \href
	{https://doi.org/10.1145/237814.237866} {\emph {\bibinfo {booktitle}
			{Proceedings of the Twenty-Eighth Annual ACM Symposium on Theory of
				Computing}}},\ \bibinfo {series and number} {STOC '96}\ (\bibinfo
	{publisher} {Association for Computing Machinery},\ \bibinfo {address} {New
		York},\ \bibinfo {year} {1996})\ pp.\ \bibinfo {pages} {212--219}\BibitemShut
	{NoStop}%
	\bibitem [{\citenamefont {Dawson}\ and\ \citenamefont
		{Nielsen}(2006)}]{dawson2006solovay}%
	\BibitemOpen
	\bibfield  {author} {\bibinfo {author} {\bibfnamefont {C.~M.}\ \bibnamefont
			{Dawson}}\ and\ \bibinfo {author} {\bibfnamefont {M.~A.}\ \bibnamefont
			{Nielsen}},\ }\bibfield  {title} {\bibinfo {title} {The solovay-kitaev
			algorithm},\ }\href@noop {} {\bibfield  {journal} {\bibinfo  {journal}
			{Quantum Info. Comput.}\ }\textbf {\bibinfo {volume} {6}},\ \bibinfo {pages}
		{81} (\bibinfo {year} {2006})}\BibitemShut {NoStop}%
	\bibitem [{sup()}]{suppl}%
	\BibitemOpen
	\href@noop {} {}\bibinfo {note} {See Supplemental Material for details and
		Refs.~\cite{Odake2023higherorder, Li2024PurityAmplification, Long2006LCU,
			Childs2012LCU, Vidal2000Schmidt, Chen2022ProductFormulae}
		therein.}\BibitemShut {Stop}%
	\bibitem [{\citenamefont {Karimi}\ \emph {et~al.}(2016)\citenamefont {Karimi},
		\citenamefont {Nutini},\ and\ \citenamefont {Schmidt}}]{karimi2016linear}%
	\BibitemOpen
	\bibfield  {author} {\bibinfo {author} {\bibfnamefont {H.}~\bibnamefont
			{Karimi}}, \bibinfo {author} {\bibfnamefont {J.}~\bibnamefont {Nutini}},\
		and\ \bibinfo {author} {\bibfnamefont {M.}~\bibnamefont {Schmidt}},\
	}\bibfield  {title} {\bibinfo {title} {Linear convergence of gradient and
			proximal-gradient methods under the polyak-{\l}ojasiewicz condition},\ }in\
	\href@noop {} {\emph {\bibinfo {booktitle} {Machine Learning and Knowledge
				Discovery in Databases: European Conference, ECML PKDD 2016, Riva del Garda,
				Italy, September 19-23, 2016, Proceedings, Part I 16}}}\ (\bibinfo
	{organization} {Springer},\ \bibinfo {year} {2016})\ pp.\ \bibinfo {pages}
	{795--811}\BibitemShut {NoStop}%
	\bibitem [{\citenamefont {Moore}\ \emph {et~al.}(1994)\citenamefont {Moore},
		\citenamefont {Mahony},\ and\ \citenamefont {Helmke}}]{Moore1994DiscreteDBI}%
	\BibitemOpen
	\bibfield  {author} {\bibinfo {author} {\bibfnamefont {J.~B.}\ \bibnamefont
			{Moore}}, \bibinfo {author} {\bibfnamefont {R.~E.}\ \bibnamefont {Mahony}},\
		and\ \bibinfo {author} {\bibfnamefont {U.}~\bibnamefont {Helmke}},\
	}\bibfield  {title} {\bibinfo {title} {Numerical gradient algorithms for
			eigenvalue and singular value calculations},\ }\href
	{https://doi.org/10.1137/S0036141092229732} {\bibfield  {journal} {\bibinfo
			{journal} {SIAM J. Matrix Anal. Appl.}\ }\textbf {\bibinfo {volume} {15}},\
		\bibinfo {pages} {881} (\bibinfo {year} {1994})}\BibitemShut {NoStop}%
	\bibitem [{\citenamefont {Cirac}\ \emph {et~al.}(1999)\citenamefont {Cirac},
		\citenamefont {Ekert},\ and\ \citenamefont
		{Macchiavello}}]{Cirac1999_purification}%
	\BibitemOpen
	\bibfield  {author} {\bibinfo {author} {\bibfnamefont {J.~I.}\ \bibnamefont
			{Cirac}}, \bibinfo {author} {\bibfnamefont {A.~K.}\ \bibnamefont {Ekert}},\
		and\ \bibinfo {author} {\bibfnamefont {C.}~\bibnamefont {Macchiavello}},\
	}\bibfield  {title} {\bibinfo {title} {Optimal purification of single
			qubits},\ }\href {https://doi.org/10.1103/PhysRevLett.82.4344} {\bibfield
		{journal} {\bibinfo  {journal} {Phys. Rev. Lett.}\ }\textbf {\bibinfo
			{volume} {82}},\ \bibinfo {pages} {4344} (\bibinfo {year}
		{1999})}\BibitemShut {NoStop}%
	\bibitem [{\citenamefont {Childs}\ \emph {et~al.}(2025)\citenamefont {Childs},
		\citenamefont {Fu}, \citenamefont {Leung}, \citenamefont {Li}, \citenamefont
		{Ozols},\ and\ \citenamefont {Vyas}}]{Childs2024purification}%
	\BibitemOpen
	\bibfield  {author} {\bibinfo {author} {\bibfnamefont {A.~M.}\ \bibnamefont
			{Childs}}, \bibinfo {author} {\bibfnamefont {H.}~\bibnamefont {Fu}}, \bibinfo
		{author} {\bibfnamefont {D.}~\bibnamefont {Leung}}, \bibinfo {author}
		{\bibfnamefont {Z.}~\bibnamefont {Li}}, \bibinfo {author} {\bibfnamefont
			{M.}~\bibnamefont {Ozols}},\ and\ \bibinfo {author} {\bibfnamefont
			{V.}~\bibnamefont {Vyas}},\ }\bibfield  {title} {\bibinfo {title} {Streaming
			quantum state purification},\ }\href
	{https://doi.org/10.22331/q-2025-01-21-1603} {\bibfield  {journal} {\bibinfo
			{journal} {{Quantum}}\ }\textbf {\bibinfo {volume} {9}},\ \bibinfo {pages}
		{1603} (\bibinfo {year} {2025})}\BibitemShut {NoStop}%
	\bibitem [{\citenamefont {Takagi}\ \emph {et~al.}(2022)\citenamefont {Takagi},
		\citenamefont {Endo}, \citenamefont {Minagawa},\ and\ \citenamefont
		{Gu}}]{Takagi2022fundamental}%
	\BibitemOpen
	\bibfield  {author} {\bibinfo {author} {\bibfnamefont {R.}~\bibnamefont
			{Takagi}}, \bibinfo {author} {\bibfnamefont {S.}~\bibnamefont {Endo}},
		\bibinfo {author} {\bibfnamefont {S.}~\bibnamefont {Minagawa}},\ and\
		\bibinfo {author} {\bibfnamefont {M.}~\bibnamefont {Gu}},\ }\bibfield
	{title} {\bibinfo {title} {Fundamental limits of quantum error mitigation},\
	}\href {https://doi.org/10.1038/s41534-022-00618-z} {\bibfield  {journal}
		{\bibinfo  {journal} {npj Quantum Information}\ }\textbf {\bibinfo {volume}
			{8}},\ \bibinfo {pages} {1} (\bibinfo {year} {2022})}\BibitemShut {NoStop}%
	\bibitem [{\citenamefont {Takagi}\ \emph {et~al.}(2023)\citenamefont {Takagi},
		\citenamefont {Tajima},\ and\ \citenamefont {Gu}}]{Takagi2023universal}%
	\BibitemOpen
	\bibfield  {author} {\bibinfo {author} {\bibfnamefont {R.}~\bibnamefont
			{Takagi}}, \bibinfo {author} {\bibfnamefont {H.}~\bibnamefont {Tajima}},\
		and\ \bibinfo {author} {\bibfnamefont {M.}~\bibnamefont {Gu}},\ }\bibfield
	{title} {\bibinfo {title} {Universal sampling lower bounds for quantum error
			mitigation},\ }\href {https://doi.org/10.1103/PhysRevLett.131.210602}
	{\bibfield  {journal} {\bibinfo  {journal} {Phys. Rev. Lett.}\ }\textbf
		{\bibinfo {volume} {131}},\ \bibinfo {pages} {210602} (\bibinfo {year}
		{2023})}\BibitemShut {NoStop}%
	\bibitem [{\citenamefont {Tsubouchi}\ \emph {et~al.}(2023)\citenamefont
		{Tsubouchi}, \citenamefont {Sagawa},\ and\ \citenamefont
		{Yoshioka}}]{Tsubouchi2023universal}%
	\BibitemOpen
	\bibfield  {author} {\bibinfo {author} {\bibfnamefont {K.}~\bibnamefont
			{Tsubouchi}}, \bibinfo {author} {\bibfnamefont {T.}~\bibnamefont {Sagawa}},\
		and\ \bibinfo {author} {\bibfnamefont {N.}~\bibnamefont {Yoshioka}},\
	}\bibfield  {title} {\bibinfo {title} {Universal cost bound of quantum error
			mitigation based on quantum estimation theory},\ }\href
	{https://doi.org/10.1103/PhysRevLett.131.210601} {\bibfield  {journal}
		{\bibinfo  {journal} {Phys. Rev. Lett.}\ }\textbf {\bibinfo {volume} {131}},\
		\bibinfo {pages} {210601} (\bibinfo {year} {2023})}\BibitemShut {NoStop}%
	\bibitem [{\citenamefont {Temme}\ \emph {et~al.}(2017)\citenamefont {Temme},
		\citenamefont {Bravyi},\ and\ \citenamefont {Gambetta}}]{Temme2017error}%
	\BibitemOpen
	\bibfield  {author} {\bibinfo {author} {\bibfnamefont {K.}~\bibnamefont
			{Temme}}, \bibinfo {author} {\bibfnamefont {S.}~\bibnamefont {Bravyi}},\ and\
		\bibinfo {author} {\bibfnamefont {J.~M.}\ \bibnamefont {Gambetta}},\
	}\bibfield  {title} {\bibinfo {title} {Error mitigation for short-depth
			quantum circuits},\ }\href {https://doi.org/10.1103/PhysRevLett.119.180509}
	{\bibfield  {journal} {\bibinfo  {journal} {Phys. Rev. Lett.}\ }\textbf
		{\bibinfo {volume} {119}},\ \bibinfo {pages} {180509} (\bibinfo {year}
		{2017})}\BibitemShut {NoStop}%
	\bibitem [{\citenamefont {Quek}\ \emph {et~al.}(2024)\citenamefont {Quek},
		\citenamefont {Stilck~França}, \citenamefont {Khatri}, \citenamefont
		{Meyer},\ and\ \citenamefont {Eisert}}]{Quek2024exponentially}%
	\BibitemOpen
	\bibfield  {author} {\bibinfo {author} {\bibfnamefont {Y.}~\bibnamefont
			{Quek}}, \bibinfo {author} {\bibfnamefont {D.}~\bibnamefont
			{Stilck~França}}, \bibinfo {author} {\bibfnamefont {S.}~\bibnamefont
			{Khatri}}, \bibinfo {author} {\bibfnamefont {J.~J.}\ \bibnamefont {Meyer}},\
		and\ \bibinfo {author} {\bibfnamefont {J.}~\bibnamefont {Eisert}},\
	}\bibfield  {title} {\bibinfo {title} {Exponentially tighter bounds on
			limitations of quantum error mitigation},\ }\href
	{https://doi.org/10.1038/s41567-024-02536-7} {\bibfield  {journal} {\bibinfo
			{journal} {Nature Physics}\ }\textbf {\bibinfo {volume} {20}},\ \bibinfo
		{pages} {1648} (\bibinfo {year} {2024})}\BibitemShut {NoStop}%
	\bibitem [{\citenamefont {Yoder}\ \emph {et~al.}(2016)\citenamefont {Yoder},
		\citenamefont {Takagi},\ and\ \citenamefont {Chuang}}]{Yoder2016universal}%
	\BibitemOpen
	\bibfield  {author} {\bibinfo {author} {\bibfnamefont {T.~J.}\ \bibnamefont
			{Yoder}}, \bibinfo {author} {\bibfnamefont {R.}~\bibnamefont {Takagi}},\ and\
		\bibinfo {author} {\bibfnamefont {I.~L.}\ \bibnamefont {Chuang}},\ }\bibfield
	{title} {\bibinfo {title} {Universal fault-tolerant gates on concatenated
			stabilizer codes},\ }\href {https://doi.org/10.1103/PhysRevX.6.031039}
	{\bibfield  {journal} {\bibinfo  {journal} {Phys. Rev. X}\ }\textbf {\bibinfo
			{volume} {6}},\ \bibinfo {pages} {031039} (\bibinfo {year}
		{2016})}\BibitemShut {NoStop}%
	\bibitem [{\citenamefont {Takagi}\ \emph {et~al.}(2017)\citenamefont {Takagi},
		\citenamefont {Yoder},\ and\ \citenamefont {Chuang}}]{Takagi2017error}%
	\BibitemOpen
	\bibfield  {author} {\bibinfo {author} {\bibfnamefont {R.}~\bibnamefont
			{Takagi}}, \bibinfo {author} {\bibfnamefont {T.~J.}\ \bibnamefont {Yoder}},\
		and\ \bibinfo {author} {\bibfnamefont {I.~L.}\ \bibnamefont {Chuang}},\
	}\bibfield  {title} {\bibinfo {title} {Error rates and resource overheads of
			encoded three-qubit gates},\ }\href
	{https://doi.org/10.1103/PhysRevA.96.042302} {\bibfield  {journal} {\bibinfo
			{journal} {Phys. Rev. A}\ }\textbf {\bibinfo {volume} {96}},\ \bibinfo
		{pages} {042302} (\bibinfo {year} {2017})}\BibitemShut {NoStop}%
	\bibitem [{\citenamefont {Wehner}\ \emph {et~al.}(2018)\citenamefont {Wehner},
		\citenamefont {Elkouss},\ and\ \citenamefont
		{Hanson}}]{Wehner2018_QInternet}%
	\BibitemOpen
	\bibfield  {author} {\bibinfo {author} {\bibfnamefont {S.}~\bibnamefont
			{Wehner}}, \bibinfo {author} {\bibfnamefont {D.}~\bibnamefont {Elkouss}},\
		and\ \bibinfo {author} {\bibfnamefont {R.}~\bibnamefont {Hanson}},\
	}\bibfield  {title} {\bibinfo {title} {Quantum internet: A vision for the
			road ahead},\ }\href {https://doi.org/10.1126/science.aam9288} {\bibfield
		{journal} {\bibinfo  {journal} {Science}\ }\textbf {\bibinfo {volume}
			{362}},\ \bibinfo {pages} {eaam9288} (\bibinfo {year} {2018})}\BibitemShut
	{NoStop}%
	\bibitem [{\citenamefont {Cacciapuoti}\ \emph {et~al.}(2019)\citenamefont
		{Cacciapuoti}, \citenamefont {Caleffi}, \citenamefont {Tafuri}, \citenamefont
		{Cataliotti}, \citenamefont {Gherardini},\ and\ \citenamefont
		{Bianchi}}]{Cacciapuoti2019Distributed}%
	\BibitemOpen
	\bibfield  {author} {\bibinfo {author} {\bibfnamefont {A.~S.}\ \bibnamefont
			{Cacciapuoti}}, \bibinfo {author} {\bibfnamefont {M.}~\bibnamefont
			{Caleffi}}, \bibinfo {author} {\bibfnamefont {F.}~\bibnamefont {Tafuri}},
		\bibinfo {author} {\bibfnamefont {F.~S.}\ \bibnamefont {Cataliotti}},
		\bibinfo {author} {\bibfnamefont {S.}~\bibnamefont {Gherardini}},\ and\
		\bibinfo {author} {\bibfnamefont {G.}~\bibnamefont {Bianchi}},\ }\bibfield
	{title} {\bibinfo {title} {Quantum internet: Networking challenges in
			distributed quantum computing},\ }\href
	{https://doi.org/https://doi.org/10.1109/MNET.001.1900092} {\bibfield
		{journal} {\bibinfo  {journal} {IEEE Netw.}\ }\textbf {\bibinfo {volume}
			{34}},\ \bibinfo {pages} {137} (\bibinfo {year} {2019})}\BibitemShut
	{NoStop}%
	\bibitem [{\citenamefont {Davarzani}\ \emph {et~al.}(2020)\citenamefont
		{Davarzani}, \citenamefont {Zomorodi-Moghadam}, \citenamefont {Houshmand},\
		and\ \citenamefont {Nouri-baygi}}]{Davarzani2020Distributed}%
	\BibitemOpen
	\bibfield  {author} {\bibinfo {author} {\bibfnamefont {Z.}~\bibnamefont
			{Davarzani}}, \bibinfo {author} {\bibfnamefont {M.}~\bibnamefont
			{Zomorodi-Moghadam}}, \bibinfo {author} {\bibfnamefont {M.}~\bibnamefont
			{Houshmand}},\ and\ \bibinfo {author} {\bibfnamefont {M.}~\bibnamefont
			{Nouri-baygi}},\ }\bibfield  {title} {\bibinfo {title} {A dynamic programming
			approach for distributing quantum circuits by bipartite graphs},\ }\href
	{https://doi.org/10.1007/s11128-020-02871-7} {\bibfield  {journal} {\bibinfo
			{journal} {Quantum Information Processing}\ }\textbf {\bibinfo {volume}
			{19}},\ \bibinfo {pages} {360} (\bibinfo {year} {2020})}\BibitemShut
	{NoStop}%
	\bibitem [{\citenamefont {Wang}\ \emph {et~al.}(2024)\citenamefont {Wang},
		\citenamefont {Yang},\ and\ \citenamefont {Gu}}]{Wang2024Decoupling}%
	\BibitemOpen
	\bibfield  {author} {\bibinfo {author} {\bibfnamefont {X.}~\bibnamefont
			{Wang}}, \bibinfo {author} {\bibfnamefont {C.}~\bibnamefont {Yang}},\ and\
		\bibinfo {author} {\bibfnamefont {M.}~\bibnamefont {Gu}},\ }\bibfield
	{title} {\bibinfo {title} {Variational quantum circuit decoupling},\ }\href
	{https://doi.org/10.1103/PhysRevLett.133.130602} {\bibfield  {journal}
		{\bibinfo  {journal} {Phys. Rev. Lett.}\ }\textbf {\bibinfo {volume} {133}},\
		\bibinfo {pages} {130602} (\bibinfo {year} {2024})}\BibitemShut {NoStop}%
	\bibitem [{\citenamefont {Brassard}\ \emph {et~al.}(2002)\citenamefont
		{Brassard}, \citenamefont {Hoyer}, \citenamefont {Mosca},\ and\ \citenamefont
		{Tapp}}]{Brassard2002ampamp}%
	\BibitemOpen
	\bibfield  {author} {\bibinfo {author} {\bibfnamefont {G.}~\bibnamefont
			{Brassard}}, \bibinfo {author} {\bibfnamefont {P.}~\bibnamefont {Hoyer}},
		\bibinfo {author} {\bibfnamefont {M.}~\bibnamefont {Mosca}},\ and\ \bibinfo
		{author} {\bibfnamefont {A.}~\bibnamefont {Tapp}},\ }\bibfield  {title}
	{\bibinfo {title} {Quantum amplitude amplification and estimation},\ }\href
	{https://doi.org/https://doi.org/10.1090/conm/305/05215} {\bibfield
		{journal} {\bibinfo  {journal} {Contemp. Math.}\ }\textbf {\bibinfo {volume}
			{305}},\ \bibinfo {pages} {53} (\bibinfo {year} {2002})}\BibitemShut
	{NoStop}%
	\bibitem [{\citenamefont {Wick}(1954)}]{Wick1954QITE}%
	\BibitemOpen
	\bibfield  {author} {\bibinfo {author} {\bibfnamefont {G.~C.}\ \bibnamefont
			{Wick}},\ }\bibfield  {title} {\bibinfo {title} {Properties of bethe-salpeter
			wave functions},\ }\href {https://doi.org/10.1103/PhysRev.96.1124} {\bibfield
		{journal} {\bibinfo  {journal} {Phys. Rev.}\ }\textbf {\bibinfo {volume}
			{96}},\ \bibinfo {pages} {1124} (\bibinfo {year} {1954})}\BibitemShut
	{NoStop}%
	\bibitem [{\citenamefont {Schmidt}(1907)}]{Schmidt1907}%
	\BibitemOpen
	\bibfield  {author} {\bibinfo {author} {\bibfnamefont {E.}~\bibnamefont
			{Schmidt}},\ }\bibfield  {title} {\bibinfo {title} {{Z}ur {T}heorie der
			linearen und nichtlinearen {I}ntegralgleichungen},\ }\href
	{https://doi.org/https://doi.org/10.1007/BF01449770} {\bibfield  {journal}
		{\bibinfo  {journal} {Math. Annalen}\ }\textbf {\bibinfo {volume} {63}},\
		\bibinfo {pages} {433} (\bibinfo {year} {1907})}\BibitemShut {NoStop}%
	\bibitem [{\citenamefont {Ekert}\ and\ \citenamefont
		{Knight}(1995)}]{Ekert1995Schmidt}%
	\BibitemOpen
	\bibfield  {author} {\bibinfo {author} {\bibfnamefont {A.}~\bibnamefont
			{Ekert}}\ and\ \bibinfo {author} {\bibfnamefont {P.~L.}\ \bibnamefont
			{Knight}},\ }\bibfield  {title} {\bibinfo {title} {Entangled quantum systems
			and the {S}chmidt decomposition},\ }\href
	{https://doi.org/https://doi.org/10.1119/1.17904} {\bibfield  {journal}
		{\bibinfo  {journal} {Am. J. Phys.}\ }\textbf {\bibinfo {volume} {63}},\
		\bibinfo {pages} {415} (\bibinfo {year} {1995})}\BibitemShut {NoStop}%
	\bibitem [{\citenamefont {Helmke}\ and\ \citenamefont
		{Moore}(1994)}]{HelmkeMoore1994Book}%
	\BibitemOpen
	\bibfield  {author} {\bibinfo {author} {\bibfnamefont {U.}~\bibnamefont
			{Helmke}}\ and\ \bibinfo {author} {\bibfnamefont {J.~B.}\ \bibnamefont
			{Moore}},\ }\href {https://doi.org/https://doi.org/10.1007/978-1-4471-3467-1}
	{\emph {\bibinfo {title} {Optimization and Dynamical Systems}}},\ \bibinfo
	{edition} {1st}\ ed.,\ Communications and Control Engineering\ (\bibinfo
	{publisher} {Springer London},\ \bibinfo {year} {1994})\BibitemShut {NoStop}%
	\bibitem [{\citenamefont {Smith}(1993)}]{Smith_Thesis}%
	\BibitemOpen
	\bibfield  {author} {\bibinfo {author} {\bibfnamefont {S.~T.}\ \bibnamefont
			{Smith}},\ }\emph {\bibinfo {title} {Geometric Optimization Methods for
			Adaptive Filtering}},\ \href@noop {} {Ph.D. thesis},\ \bibinfo  {school}
	{Harvard University} (\bibinfo {year} {1993})\BibitemShut {NoStop}%
	\bibitem [{\citenamefont {Horodecki}\ and\ \citenamefont
		{Ekert}(2002)}]{Horodecki2002method}%
	\BibitemOpen
	\bibfield  {author} {\bibinfo {author} {\bibfnamefont {P.}~\bibnamefont
			{Horodecki}}\ and\ \bibinfo {author} {\bibfnamefont {A.}~\bibnamefont
			{Ekert}},\ }\bibfield  {title} {\bibinfo {title} {Method for direct detection
			of quantum entanglement},\ }\href
	{https://doi.org/10.1103/PhysRevLett.89.127902} {\bibfield  {journal}
		{\bibinfo  {journal} {Phys. Rev. Lett.}\ }\textbf {\bibinfo {volume} {89}},\
		\bibinfo {pages} {127902} (\bibinfo {year} {2002})}\BibitemShut {NoStop}%
	\bibitem [{\citenamefont {Bennett}\ \emph {et~al.}(1996)\citenamefont
		{Bennett}, \citenamefont {Bernstein}, \citenamefont {Popescu},\ and\
		\citenamefont {Schumacher}}]{Bennett1996concentrating}%
	\BibitemOpen
	\bibfield  {author} {\bibinfo {author} {\bibfnamefont {C.~H.}\ \bibnamefont
			{Bennett}}, \bibinfo {author} {\bibfnamefont {H.~J.}\ \bibnamefont
			{Bernstein}}, \bibinfo {author} {\bibfnamefont {S.}~\bibnamefont {Popescu}},\
		and\ \bibinfo {author} {\bibfnamefont {B.}~\bibnamefont {Schumacher}},\
	}\bibfield  {title} {\bibinfo {title} {Concentrating partial entanglement by
			local operations},\ }\href {https://doi.org/10.1103/PhysRevA.53.2046}
	{\bibfield  {journal} {\bibinfo  {journal} {Phys. Rev. A}\ }\textbf {\bibinfo
			{volume} {53}},\ \bibinfo {pages} {2046} (\bibinfo {year}
		{1996})}\BibitemShut {NoStop}%
	\bibitem [{\citenamefont {Matsumoto}\ and\ \citenamefont
		{Hayashi}(2007)}]{Matsumoto2007universal}%
	\BibitemOpen
	\bibfield  {author} {\bibinfo {author} {\bibfnamefont {K.}~\bibnamefont
			{Matsumoto}}\ and\ \bibinfo {author} {\bibfnamefont {M.}~\bibnamefont
			{Hayashi}},\ }\bibfield  {title} {\bibinfo {title} {Universal distortion-free
			entanglement concentration},\ }\href
	{https://doi.org/10.1103/PhysRevA.75.062338} {\bibfield  {journal} {\bibinfo
			{journal} {Phys. Rev. A}\ }\textbf {\bibinfo {volume} {75}},\ \bibinfo
		{pages} {062338} (\bibinfo {year} {2007})}\BibitemShut {NoStop}%
	\bibitem [{\citenamefont {Low}\ and\ \citenamefont
		{Chuang}(2019)}]{Low2019qubitization}%
	\BibitemOpen
	\bibfield  {author} {\bibinfo {author} {\bibfnamefont {G.~H.}\ \bibnamefont
			{Low}}\ and\ \bibinfo {author} {\bibfnamefont {I.~L.}\ \bibnamefont
			{Chuang}},\ }\bibfield  {title} {\bibinfo {title} {Hamiltonian {S}imulation
			by {Q}ubitization},\ }\href {https://doi.org/10.22331/q-2019-07-12-163}
	{\bibfield  {journal} {\bibinfo  {journal} {{Quantum}}\ }\textbf {\bibinfo
			{volume} {3}},\ \bibinfo {pages} {163} (\bibinfo {year} {2019})}\BibitemShut
	{NoStop}%
	\bibitem [{\citenamefont {Scarani}\ \emph {et~al.}(2005)\citenamefont
		{Scarani}, \citenamefont {Iblisdir}, \citenamefont {Gisin},\ and\
		\citenamefont {Ac\'{\i}n}}]{cloningAcinRevModPhys.77.1225}%
	\BibitemOpen
	\bibfield  {author} {\bibinfo {author} {\bibfnamefont {V.}~\bibnamefont
			{Scarani}}, \bibinfo {author} {\bibfnamefont {S.}~\bibnamefont {Iblisdir}},
		\bibinfo {author} {\bibfnamefont {N.}~\bibnamefont {Gisin}},\ and\ \bibinfo
		{author} {\bibfnamefont {A.}~\bibnamefont {Ac\'{\i}n}},\ }\bibfield  {title}
	{\bibinfo {title} {Quantum cloning},\ }\href
	{https://doi.org/10.1103/RevModPhys.77.1225} {\bibfield  {journal} {\bibinfo
			{journal} {Rev. Mod. Phys.}\ }\textbf {\bibinfo {volume} {77}},\ \bibinfo
		{pages} {1225} (\bibinfo {year} {2005})}\BibitemShut {NoStop}%
	\bibitem [{\citenamefont {Vidal}(2003)}]{VidalPhysRevLett.91.147902}%
	\BibitemOpen
	\bibfield  {author} {\bibinfo {author} {\bibfnamefont {G.}~\bibnamefont
			{Vidal}},\ }\bibfield  {title} {\bibinfo {title} {Efficient classical
			simulation of slightly entangled quantum computations},\ }\href
	{https://doi.org/10.1103/PhysRevLett.91.147902} {\bibfield  {journal}
		{\bibinfo  {journal} {Phys. Rev. Lett.}\ }\textbf {\bibinfo {volume} {91}},\
		\bibinfo {pages} {147902} (\bibinfo {year} {2003})}\BibitemShut {NoStop}%
	\bibitem [{\citenamefont {Odake}\ \emph {et~al.}(2024)\citenamefont {Odake},
		\citenamefont {Kristj\'ansson}, \citenamefont {Soeda},\ and\ \citenamefont
		{Murao}}]{Odake2023higherorder}%
	\BibitemOpen
	\bibfield  {author} {\bibinfo {author} {\bibfnamefont {T.}~\bibnamefont
			{Odake}}, \bibinfo {author} {\bibfnamefont {H.}~\bibnamefont
			{Kristj\'ansson}}, \bibinfo {author} {\bibfnamefont {A.}~\bibnamefont
			{Soeda}},\ and\ \bibinfo {author} {\bibfnamefont {M.}~\bibnamefont {Murao}},\
	}\bibfield  {title} {\bibinfo {title} {Higher-order quantum transformations
			of hamiltonian dynamics},\ }\href
	{https://doi.org/10.1103/PhysRevResearch.6.L012063} {\bibfield  {journal}
		{\bibinfo  {journal} {Phys. Rev. Res.}\ }\textbf {\bibinfo {volume} {6}},\
		\bibinfo {pages} {L012063} (\bibinfo {year} {2024})}\BibitemShut {NoStop}%
	\bibitem [{\citenamefont {Li}\ \emph {et~al.}(2024)\citenamefont {Li},
		\citenamefont {Fu}, \citenamefont {Isogawa},\ and\ \citenamefont
		{Chuang}}]{Li2024PurityAmplification}%
	\BibitemOpen
	\bibfield  {author} {\bibinfo {author} {\bibfnamefont {Z.}~\bibnamefont
			{Li}}, \bibinfo {author} {\bibfnamefont {H.}~\bibnamefont {Fu}}, \bibinfo
		{author} {\bibfnamefont {T.}~\bibnamefont {Isogawa}},\ and\ \bibinfo {author}
		{\bibfnamefont {I.}~\bibnamefont {Chuang}},\ }\href
	{https://arxiv.org/abs/2409.18167} {\bibinfo {title} {Optimal quantum purity
			amplification}} (\bibinfo {year} {2024}),\ \Eprint
	{https://arxiv.org/abs/2409.18167} {arXiv:2409.18167 [quant-ph]} \BibitemShut
	{NoStop}%
	\bibitem [{\citenamefont {Long}(2006)}]{Long2006LCU}%
	\BibitemOpen
	\bibfield  {author} {\bibinfo {author} {\bibfnamefont {G.-L.}\ \bibnamefont
			{Long}},\ }\bibfield  {title} {\bibinfo {title} {General quantum interference
			principle and duality computer},\ }\href
	{https://doi.org/10.1088/0253-6102/45/5/013} {\bibfield  {journal} {\bibinfo
			{journal} {Commun. Theor. Phys.}\ }\textbf {\bibinfo {volume} {45}},\
		\bibinfo {pages} {825} (\bibinfo {year} {2006})}\BibitemShut {NoStop}%
	\bibitem [{\citenamefont {Childs}\ and\ \citenamefont
		{Wiebe}(2012)}]{Childs2012LCU}%
	\BibitemOpen
	\bibfield  {author} {\bibinfo {author} {\bibfnamefont {A.~M.}\ \bibnamefont
			{Childs}}\ and\ \bibinfo {author} {\bibfnamefont {N.}~\bibnamefont {Wiebe}},\
	}\bibfield  {title} {\bibinfo {title} {Hamiltonian simulation using linear
			combinations of unitary operations},\ }\href
	{https://doi.org/https://doi.org/10.26421/QIC12.11-12-1} {\bibfield
		{journal} {\bibinfo  {journal} {Quantum Inf. Comput.}\ }\textbf {\bibinfo
			{volume} {12}},\ \bibinfo {pages} {901} (\bibinfo {year} {2012})}\BibitemShut
	{NoStop}%
	\bibitem [{\citenamefont {Vidal}(2000)}]{Vidal2000Schmidt}%
	\BibitemOpen
	\bibfield  {author} {\bibinfo {author} {\bibfnamefont {G.}~\bibnamefont
			{Vidal}},\ }\bibfield  {title} {\bibinfo {title} {Entanglement monotones},\
	}\href {https://doi.org/10.1080/09500340008244048} {\bibfield  {journal}
		{\bibinfo  {journal} {J. Mod. Opt.}\ }\textbf {\bibinfo {volume} {47}},\
		\bibinfo {pages} {355} (\bibinfo {year} {2000})}\BibitemShut {NoStop}%
	\bibitem [{\citenamefont {Chen}\ \emph {et~al.}(2022)\citenamefont {Chen},
		\citenamefont {Childs}, \citenamefont {Hafezi}, \citenamefont {Jiang},
		\citenamefont {Kim},\ and\ \citenamefont {Xu}}]{Chen2022ProductFormulae}%
	\BibitemOpen
	\bibfield  {author} {\bibinfo {author} {\bibfnamefont {Y.-A.}\ \bibnamefont
			{Chen}}, \bibinfo {author} {\bibfnamefont {A.~M.}\ \bibnamefont {Childs}},
		\bibinfo {author} {\bibfnamefont {M.}~\bibnamefont {Hafezi}}, \bibinfo
		{author} {\bibfnamefont {Z.}~\bibnamefont {Jiang}}, \bibinfo {author}
		{\bibfnamefont {H.}~\bibnamefont {Kim}},\ and\ \bibinfo {author}
		{\bibfnamefont {Y.}~\bibnamefont {Xu}},\ }\bibfield  {title} {\bibinfo
		{title} {Efficient product formulas for commutators and applications to
			quantum simulation},\ }\href
	{https://doi.org/10.1103/PhysRevResearch.4.013191} {\bibfield  {journal}
		{\bibinfo  {journal} {Phys. Rev. Res.}\ }\textbf {\bibinfo {volume} {4}},\
		\bibinfo {pages} {013191} (\bibinfo {year} {2022})}\BibitemShut {NoStop}%
	\bibitem [{\citenamefont {Liu}\ \emph {et~al.}(2024)\citenamefont {Liu},
		\citenamefont {Zhang}, \citenamefont {Fei},\ and\ \citenamefont
		{Cai}}]{Liu2024ChannelPurification}%
	\BibitemOpen
	\bibfield  {author} {\bibinfo {author} {\bibfnamefont {Z.}~\bibnamefont
			{Liu}}, \bibinfo {author} {\bibfnamefont {X.}~\bibnamefont {Zhang}}, \bibinfo
		{author} {\bibfnamefont {Y.-Y.}\ \bibnamefont {Fei}},\ and\ \bibinfo {author}
		{\bibfnamefont {Z.}~\bibnamefont {Cai}},\ }\href
	{https://arxiv.org/abs/2402.07866} {\bibinfo {title} {Virtual channel
			purification}} (\bibinfo {year} {2024}),\ \Eprint
	{https://arxiv.org/abs/2402.07866} {arXiv:2402.07866 [quant-ph]} \BibitemShut
	{NoStop}%
	\bibitem [{\citenamefont {Bharti}\ \emph {et~al.}(2022)\citenamefont {Bharti},
		\citenamefont {Cervera-Lierta}, \citenamefont {Kyaw}, \citenamefont {Haug},
		\citenamefont {Alperin-Lea}, \citenamefont {Anand}, \citenamefont {Degroote},
		\citenamefont {Heimonen}, \citenamefont {Kottmann}, \citenamefont {Menke}
		\emph {et~al.}}]{bharti2022noisy}%
	\BibitemOpen
	\bibfield  {author} {\bibinfo {author} {\bibfnamefont {K.}~\bibnamefont
			{Bharti}}, \bibinfo {author} {\bibfnamefont {A.}~\bibnamefont
			{Cervera-Lierta}}, \bibinfo {author} {\bibfnamefont {T.~H.}\ \bibnamefont
			{Kyaw}}, \bibinfo {author} {\bibfnamefont {T.}~\bibnamefont {Haug}}, \bibinfo
		{author} {\bibfnamefont {S.}~\bibnamefont {Alperin-Lea}}, \bibinfo {author}
		{\bibfnamefont {A.}~\bibnamefont {Anand}}, \bibinfo {author} {\bibfnamefont
			{M.}~\bibnamefont {Degroote}}, \bibinfo {author} {\bibfnamefont
			{H.}~\bibnamefont {Heimonen}}, \bibinfo {author} {\bibfnamefont {J.~S.}\
			\bibnamefont {Kottmann}}, \bibinfo {author} {\bibfnamefont {T.}~\bibnamefont
			{Menke}}, \emph {et~al.},\ }\bibfield  {title} {\bibinfo {title} {Noisy
			intermediate-scale quantum algorithms},\ }\href@noop {} {\bibfield  {journal}
		{\bibinfo  {journal} {Reviews of Modern Physics}\ }\textbf {\bibinfo {volume}
			{94}},\ \bibinfo {pages} {015004} (\bibinfo {year} {2022})}\BibitemShut
	{NoStop}%
	\bibitem [{\citenamefont {Larocca}\ \emph {et~al.}(2024)\citenamefont
		{Larocca}, \citenamefont {Thanasilp}, \citenamefont {Wang}, \citenamefont
		{Sharma}, \citenamefont {Biamonte}, \citenamefont {Coles}, \citenamefont
		{Cincio}, \citenamefont {McClean}, \citenamefont {Holmes},\ and\
		\citenamefont {Cerezo}}]{Larocca2024BPreview}%
	\BibitemOpen
	\bibfield  {author} {\bibinfo {author} {\bibfnamefont {M.}~\bibnamefont
			{Larocca}}, \bibinfo {author} {\bibfnamefont {S.}~\bibnamefont {Thanasilp}},
		\bibinfo {author} {\bibfnamefont {S.}~\bibnamefont {Wang}}, \bibinfo {author}
		{\bibfnamefont {K.}~\bibnamefont {Sharma}}, \bibinfo {author} {\bibfnamefont
			{J.}~\bibnamefont {Biamonte}}, \bibinfo {author} {\bibfnamefont {P.~J.}\
			\bibnamefont {Coles}}, \bibinfo {author} {\bibfnamefont {L.}~\bibnamefont
			{Cincio}}, \bibinfo {author} {\bibfnamefont {J.~R.}\ \bibnamefont {McClean}},
		\bibinfo {author} {\bibfnamefont {Z.}~\bibnamefont {Holmes}},\ and\ \bibinfo
		{author} {\bibfnamefont {M.}~\bibnamefont {Cerezo}},\ }\href
	{https://arxiv.org/abs/2405.00781} {\bibinfo {title} {A review of barren
			plateaus in variational quantum computing}} (\bibinfo {year} {2024}),\
	\Eprint {https://arxiv.org/abs/2405.00781} {arXiv:2405.00781 [quant-ph]}
	\BibitemShut {NoStop}%
	\bibitem [{\citenamefont {Buchbinder}\ \emph {et~al.}(2013)\citenamefont
		{Buchbinder}, \citenamefont {Huang},\ and\ \citenamefont
		{Weinstein}}]{Buchbinder2013}%
	\BibitemOpen
	\bibfield  {author} {\bibinfo {author} {\bibfnamefont {S.~D.}\ \bibnamefont
			{Buchbinder}}, \bibinfo {author} {\bibfnamefont {C.~L.}\ \bibnamefont
			{Huang}},\ and\ \bibinfo {author} {\bibfnamefont {Y.~S.}\ \bibnamefont
			{Weinstein}},\ }\bibfield  {title} {\bibinfo {title} {Encoding an arbitrary
			state in a {$[$}7,1,3{$]$} quantum error correction code},\ }\href
	{https://doi.org/10.1007/s11128-012-0414-7} {\bibfield  {journal} {\bibinfo
			{journal} {Quantum Inf. Process.}\ }\textbf {\bibinfo {volume} {12}},\
		\bibinfo {pages} {699} (\bibinfo {year} {2013})}\BibitemShut {NoStop}%
\end{thebibliography}
\end{document}